\documentclass[11pt]{article}
\usepackage[margin=1in]{geometry}
\usepackage{amsmath,amssymb,amsfonts}
\allowdisplaybreaks[4]
\usepackage{amsthm}
\usepackage{dsfont}
\usepackage{bbm}
\usepackage{amscd}
\usepackage{booktabs}
\usepackage{array}
\usepackage{multirow}
\usepackage{float}
\usepackage{graphicx}
\usepackage{color} 
\usepackage{xcolor}
\usepackage[title]{appendix}
\usepackage{setspace}
\usepackage{verbatim}
\usepackage{url}
\usepackage{enumitem}
\usepackage{manyfoot}
\usepackage{newlfont}
\usepackage{upref}
\usepackage{subcaption}
\usepackage[numbers,sort&compress]{natbib}

\numberwithin{equation}{section}
\newcolumntype{C}[1]{>{\centering\arraybackslash}m{#1}}
\usepackage[
    colorlinks=true,
    linkcolor=blue,
    citecolor=blue,
    urlcolor=blue
]{hyperref}

\theoremstyle{plain}
\newtheorem{theorem}{Theorem}[section]
\newtheorem{lemma}[theorem]{Lemma}
\newtheorem{proposition}[theorem]{Proposition}
\newtheorem{corollary}[theorem]{Corollary}

\theoremstyle{definition}

\theoremstyle{remark}
\newtheorem{remark}[theorem]{Remark}

\begin{document}

\title{Additional Food Enhances the Bifurcation Structure of Predator Competition Models}

\author{
Kanishka Goyal$^{1}$,
Chanaka Kottegoda$^{2}$,
Urvashi Verma$^{1}$,
and Rana D. Parshad$^{1}$
\\[1ex]
\small $^{1}$Department of Mathematics, Iowa State University,
Ames, Iowa 50011, USA
\\
\small $^{2}$Department of Mathematics and Physics, Marshall University,
Huntington, West Virginia 25755, USA
\\[1ex]
\small \texttt{\href{mailto:kgoyal@iastate.edu}{kgoyal@iastate.edu}},
\texttt{\href{mailto:kottegoda@marshall.edu}{kottegoda@marshall.edu}},
\texttt{\href{mailto:uverma@iastate.edu}{uverma@iastate.edu}},
\texttt{\href{mailto:rparshad@iastate.edu}{rparshad@iastate.edu}},
}

\maketitle

\begin{abstract}
Additional food sources and predator competition are both known to  impact the dynamics of predator-prey models. The Bazykin model of predator competition, with Holling type II functional response,  possesses a rich bifurcation structure consisting of a focus type degenerate Bogdanov-Takens bifurcation of codimension $3$, and a degenerate Hopf bifurcation of codimension at most $2$. Additional food models on the other hand are able to drive pest populations lower, with vast applicability in biological control. Despite these models being studied rigorously in the literature, the global bifurcation structure, of their possible complex dynamics, in a unified model, is unknown. In this work, we study an additional food model with generalized predator competition and Holling type-II functional response. Depending on the parameter values, the system can have up to three interior equilibria. Further, we show that this system exhibits a cusp-type (or focus-type) Bogdanov-Takens bifurcation of codimension at least $4$ (or $3$), a global Hopf bifurcation of codimension $3$, and a homoclinic bifurcation of codimension $3$. This shows there could exist three limit cycles around the BT point. These results demonstrate that additional food in Bazykin type models, can enhance their bifurcation structure. We discuss the applicability of these results to integrated pest management programs for the soybean aphid, wherein long term field data in the North-Central United States, shows two distinct limit cycles in aphid populations and their predators. Our results suggest  biological control with additional food, is an effective management tactic for invasive pests.
\end{abstract}


\medskip

\section{Introduction}\label{sec:intro}
\subsection{Motivation}
Invasive species have increasingly become a major driving factor of global agricultural loss \citep{roy2024curbing, PSCEWT16, PZM05, seebens2017no}, causing significant harm both ecologically and economically \citep{PB14, TS14, CO18, BOP22}. Their control is imperative, yet challenging \cite{PZM05}. Our dependence on pesticides in controlling invasive species pose negative environmental impacts, such as pests developing resistance against pesticides, wildlife losses and groundwater contamination. \citep{PB14, hladik2018environmental, bass2015global}. To address this issue, top-down biological control has been a sustainable and eco-friendly choice. This process refers to introducing natural enemies of the pest species \citep{V96, S99}, which subsequently prey upon the pest population, thereby reducing the reliance on frequent insecticide applications in the field \citep{heimpel2013environmental}. When predators alone fail to adequately control pests, their efficiency can be enhanced by supplying alternative/additional food (AF) different from the target prey \citep{SV06, T15}. To this end, several mathematical models have been developed to study predator-prey interactions with AF \citep{SP07, SP10, SP11, S02, SPD17, SPV18, VA22, verma2026t}, where they indicate that providing predators with sufficient high-quality AF may drive the pest population to extinction. Although biological control is considered environmentally friendly, it may also create unintended ecological consequences \citep{SV06}. One major concern is that the introduced predator may grow uncontrollably and cause non-target impacts \citep{PQB16}. Under these circumstances, intraspecific competition can naturally regulate such predator populations. Competition-based strategies have been widely explored in invasion management \citep{kettenring2011lessons, britton2018trophic}, particularly through approaches that manipulate resource availability, habitat restoration, and promote native species competition to limit invasive spread \citep{james2010principles, halassy2023meta}, yet the role of increasing competition among introduced predators as a tool to improve biological control and prevent uncontrolled predator growth has not been extensively studied. It was first noted qualitatively by Kuznetsov \citep{kuznetsov1998elements}, and Bazykin \citep{bazykin1998nonlinear}, that predator competition can significantly enrich the dynamics of prey-predator systems leading to complex higher-order bifurcations, while rigorous analytical results have only recently been established \citep{lu2021global}. Related techniques have also been applied to models involving type III and IV functional responses \citep{huang12004analyses, huang2013bifurcation, huang2014bifurcations}. 

A specific pest whose control, particularly in the North-Central United States, is imperative and thus highly investigated, is the invasive soybean aphid, \textit{Aphis glycines} Matsumura (Hemiptera: Aphididae). This pest is native to Asia, and was first detected in the United States in southeastern Wisconsin in July 2000 \cite{ragsdale2004soybean}. Due to its high reproductive potential, ability to migrate across long distances, and favorable environmental conditions, it rapidly became a major pest of soybean production \cite{rhainds2008toward}. The life cycle of the soybean aphid is relatively complex because it is heteroecious, requiring two different host plants during its development. The primary host is common Buckthorn \textit{(Rhamnus cathartica L.)} and the secondary host is soybean plant \textit{Glycine max (L.) Merr.} \cite{tilmon2011biology}.  During the summer, soybean aphids feed on soybean plants and reproduce parthenogenetically, meaning reproduction occurs without males. Female aphids give birth to live offspring, allowing populations to increase rapidly, with as many as 18 generations produced during a single growing season, depending on environmental conditions such as temperature. As temperatures decline in the fall, winged male and female aphids emerge and migrate to buckthorn plants, where mating occurs, and eggs are laid for overwintering. In the following spring, the overwintering eggs hatch into female aphids, which eventually produce winged offspring that migrate back to soybean fields during the summer to establish new colonies. 
Soybean aphids can colonize nearly all above-ground parts of the soybean plant and are phloem feeders. Severe plant injury occurs when heavy aphid infestations remove substantial amounts of water and nutrients from leaves and stems during feeding. They also excrete honeydew, a sticky sugar-rich waste product that promotes the growth of sooty mold on plant surfaces, further reducing photosynthetic activity. These processes can cause yield losses of up to 40\%  \citep{song2006profitability,kim2008economic}, thus calling for effective management strategies.

Prior to the invasion of the soybean aphid in North America, insecticide use in soybean production was minimal. But following its establishment, insecticide applications increased dramatically, with treated soybean acreage in the northcentral United States rising from less than 0.1\% in 2000 to over 13\% by 2006, representing nearly a 130-fold increase in insecticide use. One of the primary control methods involves the use of chemical treatments, including foliar applications of pyrethroids and organophosphates in soybean production \citep{ragsdale2011ecology}.
Although chemical insecticides are widely used for soybean aphid management, their negative environmental impacts are well documented, including harmful effects on non-target organisms and beneficial biological control agents \citep{bahlai2010choosing,desneux2004effects,desneux2007sublethal,johnson2008preventative,varenhorst2012response}. Furthermore, soybean aphid populations have developed resistance to pyrethroid insecticides \citep{hanson2017evidence}, highlighting the limitations of relying solely on chemical treatments. These challenges have further increased interest in alternative and sustainable integrated pest management (IPM) strategies that combine multiple approaches such as cultural methods, chemical controls, use of aphid-resistant cultivars, and biological control \cite{rutledge2004soybean,hodgson2012management,johnson2009probability}. Population data on the aphid in the 2000-2023 time period, shows clear multiple peak dynamics, wherein two disticnt population cycles are seen. One in the 2004-2009 period, and another from 2009 on wards. This corresponds to different management phases of the aphid, during which predators and parasitoids were introduced, and insecticide use increased \cite{bahlai2015shifts, bahlai2026asynchronous}. From a dynamical systems point of view, this points to the occurrence of two distinct stable limit cylces in the phase space - which however remains unproven. Despite several attempts being made to understand the multiple peak dynamics in aphid populations, via non-autonomous models \cite{banerjee2025two}, there is no bifurcation approach to analyzing these dynamics in autonomous systems, to the best of our knowledge. 

\subsubsection{Biological Control of Soybean Aphids}

Greenhouse experiments have been conducted to better understand the interactions between soybean aphids and their natural enemies \cite{VERMA2026111721}. In these experiments, soybean plants infested with aphids were exposed to different densities of green lacewing larvae (\textit{Chrysoperla rufilabris}), and aphid consumption was measured under controlled greenhouse conditions. The resulting experimental observations informed the development and parameterization of the mathematical models \cite{VERMA2026111721}. However, one of the key predators of soybean aphid in North America is \textit{Harmonia axyridis} (Coleoptera: Coccinellidae) \cite{rutledge2004soybean}. Fig.~\ref{fig:ltd} represents the long-term dynamics of soybean aphid populations over the period $2000-2013$, where the predator is considered to be the third instar of \textit{H. axyridis} with the corresponding attack rate and handling time reported in \cite{xue2009predation}. What is clearly observed are two distinct limit cycles of different amplitudes in phases 2 and 3. The transition seems to occur as we move into phase 3, when biological control practices increased.
Note, very similar qualitative results, showing distinctly different limit cycles  have been reported elsewhere in the literature \cite{bahlai2015shifts, bahlai2026asynchronous}.
 \begin{figure}[H]
 \begin{minipage}[b]{0.48\textwidth}
\centering
 \includegraphics[width=5.79cm, height=5cm]{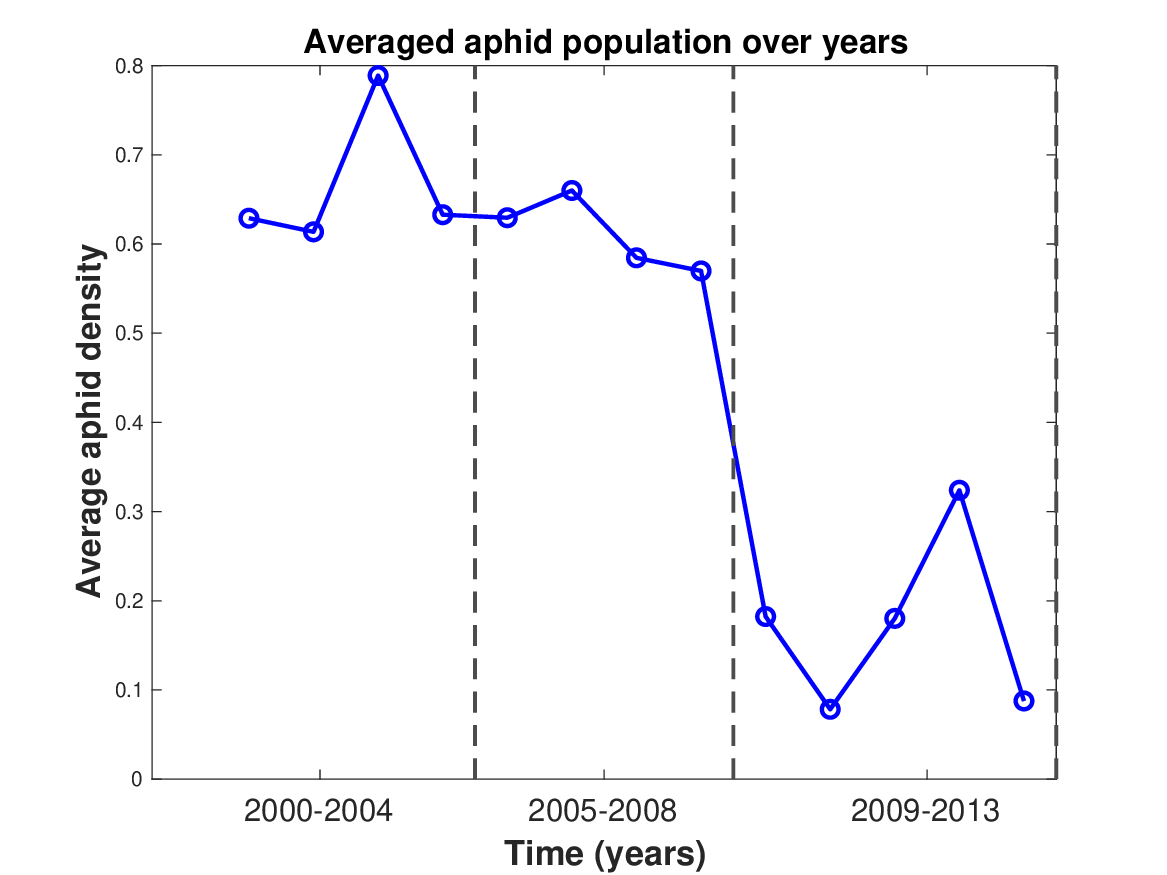}
  \\(a)
  \end{minipage}
 \hfill
\begin{minipage}[b]{0.48\textwidth}
  \centering
  \includegraphics[width= 5.79cm, height=5cm]{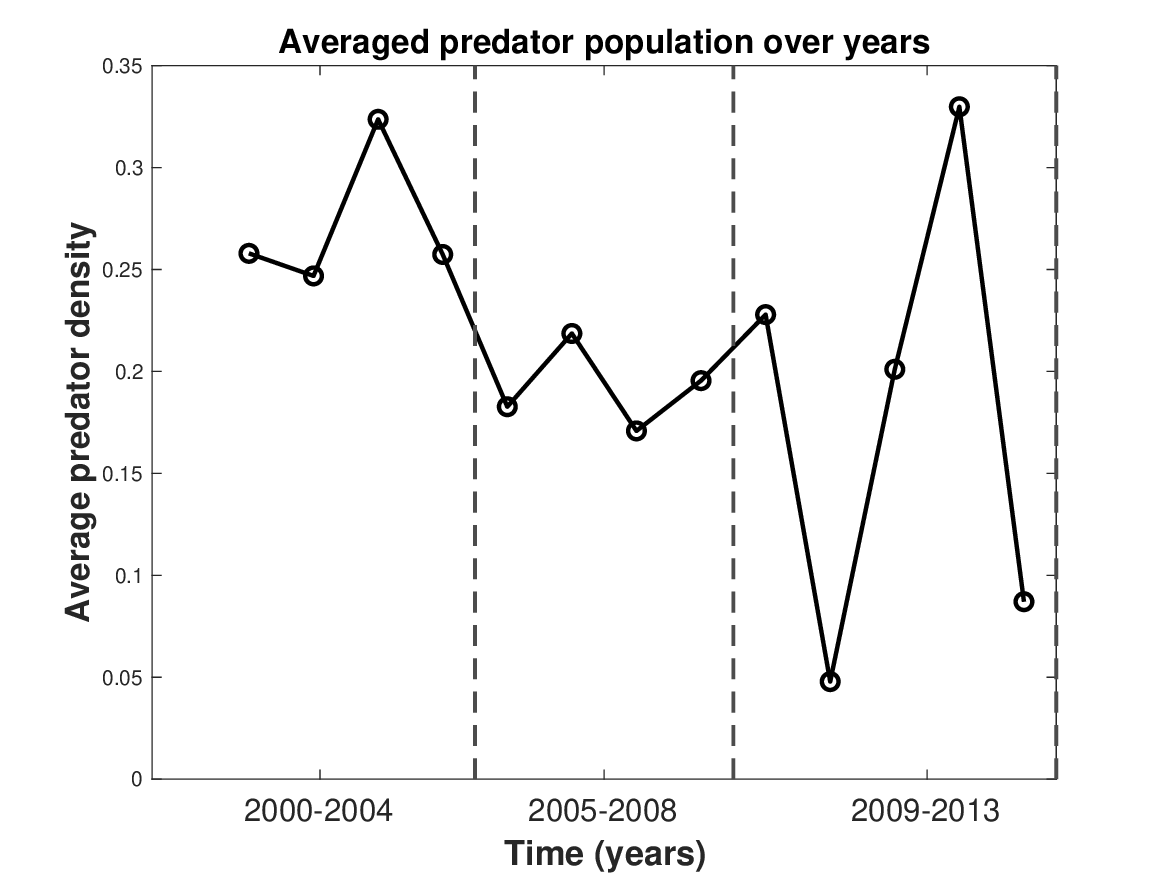}
  \\(b)
 \end{minipage}
 \caption{Population densities for aphids, predators, and parasitoids over 13 years $(2000 - 2013)$. The time series plot is divided into three phases, each corresponding to a model representing the population dynamics during the given period. Phase 1: $(2000-2004)$ when only predators were originally present, Phase 2: $(2005-2008)$ when the use of insecticides increased rapidly, and Phase 3: $(2009-2013)$ when parasitoids were introduced along with the predator and insecticide treatment.}
   \label{fig:ltd}
\end{figure}

\subsection{Related literature}
The dynamics of Gause type predator-prey systems have attracted considerable attention and have been studied extensively in a variety of ecological and mathematical settings \citep{mukherjee2020bifurcation, huang2013bifurcation, huang2014bifurcations}. The inclusion of predator competition in these models - or the so called Bazykin model, has strong ecological motivations \cite{deangelis1975model, hsu1978competing, bazykin1998nonlinear, kang2024intraspecific}, and has been fairly well studied \cite{bazykin1998nonlinear}. However, it is only recently that a rigorous program of studying the global bifurcation structure of Bazykin type models was begun \cite{lu2021global}.
Depending on the choice of functional response, these systems may exhibit a wide range of dynamical behavior, including multi-stability, periodic oscillations, and various local and global bifurcations. In particular, extensive bifurcation analyses have been carried out for predator-prey models with Holling type II, III and IV functional responses, revealing the occurrance of saddle-node, Hopf, homoclinic and Bogdanov-Takens (BT) bifurcations of higher codimension. A summary of some of these results is presented in Table~\ref{Table:1}. Besides that, predator-prey models have been studied with various other response functions in the literature \citep{freedman1980deterministic, huang2014bifurcations, xiao2007multiple, geritz2012mechanistic, baek2010qualitative, lamontagne2008bifurcation}, or systems with different ecological mechanism have also been studied rigorously. These include harvesting, refuge effects, predator interference, and Allee effects, among many others \citep{etoua2010bifurcation, lin2013bifurcations} and the references therein. These mechanisms are known to generate rich dynamical behavior and bifurcation structures \citep{arsie2022predator, arsie2023high}. 

Motivated by the potential benefits of supplementary food in biological control \citep{van2012biological}, several authors extended these classical models by incorporating AF in predator dynamics. A general predator-prey model describing an introduced predator population $(y(t))$ feeding on a target pest population $(x(t))$, while also receiving additional food, was proposed in the literature \citep{SP07, SP10, SP11}. For consistency with the notation used throughout this paper, we denote the carrying capacity and conversion efficiency parameters by $K$ and $\eta$, respectively. The model is given by
\begin{equation}
\label{Eq:1}
\frac{dx}{dt} = x\left(1-\frac{x}{K}\right) - \frac{xy}{1+\alpha \xi + x}, \quad    \frac{dy}{dt} =\frac{\eta\, xy}{1+\alpha \xi + x} + \frac{\eta\, \xi\, y}{1+\alpha \xi + x} - \delta y.
\end{equation}
wherein all parameters $(K, \alpha, \xi, \eta, \delta)$ are assumed to be positive constants. The parameter $K$ denotes the carrying capacity of the pest population, $\eta$ represents the conversion efficiency of the predator, measuring how effectively consumed prey biomass is converted into predator growth and reproduction, while $\delta$ corresponds to its natural death rate. Moreover, $\xi$ and $\frac{1}{\alpha}$ characterize the quantity and quality of AF supplied to the predator, respectively. The term $\frac{x}{1+\alpha\xi+x}$ describes the type II functional response of the predator, which depends on both the prey density and the AF supplied, whereas $\frac{\eta\, x}{1+\alpha \xi + x} + \frac{\eta\, \xi }{1+\alpha \xi + x}$ represents the corresponding numerical response of the predator. The dynamical behavior of system \eqref{Eq:1} has been extensively investigated in the literature. Srinivasu \textit{et. al.} \citep{SP07, SP10, SP11} demonstrated that the quantity and quality of AF play a crucial role in determining the effectiveness of biological control. Depending on the parameter values, the system may exhibit coexistence equilibria, stable oscillations arising through Hopf bifurcation, and significant reductions in prey density. These studies further suggested that sufficiently large amounts of high-quality AF can substantially enhance pest suppression. The literature covers a wide range of AF predator-prey models with Type III and Type IV functional responses, and AF models with different biological applications where analyses have revealed diverse bifurcation structures \citep{VA22, W23, sen2015global}. More recently, Verma \textit{et. al.} \citep{verma2026t} developed a two-patch AF model incorporating predator dispersal and drift, showing that habitat heterogeneity can improve pest suppression while generating rich dynamical behaviors such as periodic oscillations and chaos. 

Despite the extensive literature on biological control with AF (summarized in Table~\ref{Table:1}) and on predator competition separately \citep{lu2021global, broer2007dynamics, naji2013dynamics}, the combined effects of AF and predator competition on the global dynamics of predator-prey systems remain insufficiently understood. Only very recently have these been formulated and investigated \citep{parshad2023additional, verma2026additional, prakash2025dynamics, prakash2026role}, yet their higher-codimensional bifurcation structure is relatively understudied.

\begin{table}[ht]
\caption{Bifurcations/Dynamics in predator-prey models with/without AF}
\label{Table:1}
\centering
{%
\begin{tabular}{|c|C{3cm}|C{3cm}|C{4cm}|C{3.5cm}|}
\hline
& \multicolumn{2}{c|}{Without AF} & \multicolumn{2}{c|}{With AF} \\
\hline
& Functional form & Dynamics studied & Functional form & Dynamics studied \\
\hline\hline

\multirow{3}{*}{(i)} & \multirow{3}{2.5cm}{\centering $f(x) = \frac{x}{1+ x}$}
& \multirow{3}{3cm}{\centering Uniqueness, global stability of limit cycles}
&\vspace{2mm} $f(x,\xi,\alpha) = \frac{x}{1+\alpha \xi + x}$
& \multirow{3}{3cm}{\centering Saddle-node, Hopf,  Transcritical} \\
& 
& 
& \vspace{2mm} $g(x,\xi,\alpha) = \frac{\eta (x+\xi)}{1+\alpha \xi + x}$ & \\ & & \citep{cheng1981uniqueness, kuang1988uniqueness}  & & \citep{SP07}  \\ 
\hline

\multirow{2}{*}{(ii)}
& $f(x) = \frac{x^2}{ax^2+bx+ 1}$, &\vspace{1mm} Saddle-node, Nilpotent saddle (codim-3), Hopf (codim-1 \& 2), BT (codim-2 \& 3)
& $f(x,\xi,\alpha) = \frac{x^2}{1+\alpha\xi^2+x^2}$ & Transcritical, Hopf \\
& $f(x)$ $+$ \small harvesting & \citep{lamontagne2008bifurcation, etoua2010bifurcation} & $g(x,\xi,\alpha) = \frac{\eta (x^{2}+\xi^{2})}{1+\alpha \xi^{2} + x^{2}}$ & \citep{SPV18, ananth2021influence} \\
\hline

\multirow{4}{*}{(iii)}
& \multirow{4}{3cm}{\centering $f(x) = \frac{x}{ax^2+bx+ 1}$}& \multirow{4}{3cm}{\centering Saddle-node, Hopf (codim-2), Homoclinic, BT (codim-3)}
& \vspace{2mm} $f(x,\xi,\alpha) = \frac{x}{(1+\alpha \xi)( \omega x^{2}+1) + x}$ & \multirow{4}{3cm}{\centering Transcritical, Hopf}\\ 
&  &  & \vspace{2mm} $g(x,\xi,\alpha) = \frac{\eta (x+\xi(\omega x^{2}+1))}{(1+\alpha \xi)(\omega x^{2} + 1) + x}$ & \\
& & \citep{zhu2003bifurcation, xiao2001global, xiao2006multiple, rothe1992multiple}
& 
& \citep{ananth2021optimal}   \\ 
\hline

\multirow{5}{*}{(iv)} & \multirow{5}{2.5cm}{\centering $f(x)$ as in (i) $+$ Predator competition}
& \multirow{5}{3cm}{\centering \vspace{1mm} Nilpotent cusp BT (codim-2), Nilpotent focus BT (codim-3), Hopf (codim-2)} &\multirow{5}{3cm}{\centering $f(x,\xi,\alpha)$ and $g(x,\xi,\alpha)$ as in (i) $+$ Predator competition} &\multirow{5}{3cm}{\centering Codim-1: Saddle-node, Transcritical, Hopf; Codim-2: SNTC, CTC, PTC; BT (codim-2)} \\
& & & & \\
& & & & \\
& & & & \\
& & & & \\
& & & & \\
& &  \citep{lu2021global} & 
& \vspace{5mm} \citep{parshad2023additional, verma2026additional}
\\ 
\hline

\end{tabular}%
}

\vspace{0.1cm}
\end{table}

\subsection{Model formulation and current contributions}
Inspired by the aforementioned studies and rich dynamics exhibited by AF mediated predator competition model, we revisit the model studied in \citep{parshad2023additional, verma2026additional} with generalized competition term. While classical predator competition is typically represented as a $-y^2$ term, we replace it by the more general $-y^p$, where $1<p\leq2$. This encompasses a broader range of biological applications like hyperbolic mortality, nonlinear harvesting, and generalized competition/interference \citep{sambath2016stability, fenberg2008ecological, antwi2020dynamics, barman2023two}. The model is given as follows:
\begin{equation} \label{eq:model_BT}
\begin{aligned}
 \dot{x}& = x \left (1-\frac{x}{K}\right )-\frac{x y}{1+x+\alpha \xi}, \\
 \dot{y}& = \eta\, \left ( \frac{x+\xi}{1+x+\alpha \xi} \right)\ y -\delta\, y-c\, \xi\, y^p.  \\
\end{aligned}
\end{equation}
In particular, in this model, the classical AF competition model is recovered as the special case $(p=2)$ studied in \citep{parshad2023additional}. In this study, as summarized in Table~\ref{Table:1}, the authors uncovered rich bifurcation structure consisting of saddle-node, Hopf, transcritical, saddle-node transcritical, cusp transcritical, pitchfork-transcritical bifurcation, and Bogdanov-Takens bifurcation. Although their numerical results suggested the possibility of higher-order Bogdanov-Takens singularities, a rigorous analysis of such higher-codimension bifurcations remained open. In our previous work \citep{verma2026additional}, we analyzed model \eqref{eq:model_BT} with $\delta=0$ and proved the existence of Bogdanov-Takens bifurcation of codimension-2 theoretically. Furthermore, the existence of a focus-type degenerate Bogdanov-Takens bifurcation of codimension-3 remains largely unexplored and the precise codimension of corresponding Hopf and homoclinic bifurcation and the organizing center governing the bifurcation structure, remains unavailable. 
The current manuscript shows the following, 
\begin{itemize}
    \item We perform a comprehensive equilibrium and bifurcation analysis of system \eqref{eq:model_BT}, identifying the conditions for the existence and stability of the boundary equilibria via Lemma~\ref{lem:E4_stability_unidirectional}-\ref{lem:stability_pest_ext_drift} and the interior equilibria via Theorem~\ref{thm:interior_eq}.
    \item The existence of a cusp-type Bogdanov-Takens bifurcation of codimension at least $4$ for $1<p\leq2$ is established in Section~\ref{sec:double_positive} and the associated normal form and it's unfolding is derived via Theorem~\ref{thm:unfolding}.
    \item Local bifurcations present in the neighborhood of the cusp-type Bogdanov-Takens point are shown via Theorem~\ref{thm:local_bif} and demonstrated in Fig~\ref{BT4_figure}. The rich dynamical structures are also presented in Fig.~\ref{fig:BT_phase_portraits}.
    \item Around the cusp-type Bogdanov-Takens point, the existence of codimension-$3$ Hopf and codimension-$3$ homoclinic bifurcations is proved in Corollary~\ref{cor:hopf_3} and Corollary~\ref{cor:homo_3}, respectively.
    \item For the classical competition case $(p=2)$, a nilpotent focus-type Bogdanov-Takens bifurcation of codimension-$3$ is established in Theorem~\ref{thm:nilpotent_focus} for the triple root equilibrium with it's unfolding being derived in Theorem~\ref{thm:nf_unfolding}, and the rich dynamical structure is revealed in Fig.~\ref{fig:nilpotent_focus_dynamics}.
    \item Applications of these results to pest management strategies, particularly pertaining to the soybean aphid are discussed in Section~\ref{sec:conclusion}.
    \end{itemize}

The remainder of the paper is organized as follows. The preliminary analysis is carried out in Section~\ref{sec:prelim}. In Section~\ref{sec:stab_analysis}, we examine the existence and linear stability of boundary and interior equilibria of system \eqref{eq:model_BT}. Existence of Hopf bifurcation is proved in Section~\ref{sec:Hopf_bfn}. Section~\ref{sec:higher_order} is devoted to the analysis of the higher order bifurcations for the double root equilibrium. Section~\ref{sec:triple_eq} analyzes the codimension-$3$ degenerate focus-type Bogdanov-Takens bifurcation for the case $p=2$. Numerical simulations are presented in Section~\ref{sec:num_sim} to validate some theoretical results. Finally, concluding remarks are given in Section~\ref{sec:conclusion}.

\section{Preliminary Analysis}\label{sec:prelim}
In this section, we investigate the qualitative properties of system \eqref{eq:model_BT} in the biologically relevant region $\mathbb{R}_+^2=\{(x,y):x\geq0,\, y\geq0\}$. We first establish the positivity and invariance of solutions, followed by the boundedness results. These properties ensure that the system is well-posed from the biological perspective and provides a foundation for the analysis of equilibria and bifurcation behavior in the subsequent sections. 

\begin{theorem}(Positive invariance of $\mathbb{R}_+^2$).
    If $(x(0),y(0))\in \mathbb{R}_+^2$, then $(x(t),y(t))\in \mathbb{R}_+^2$ for all $t\geq0$ for which the solution exists.
\end{theorem}
\begin{proof}
    We analyze the direction of the vector field along the boundary of the nonnegative quadrant.
    \begin{enumerate}
        \item On the boundary $x=0$: From \eqref{eq:model_BT}, we have 
        \[
        \left.\frac{dx}{dt}\right|_{x=0}= 0.
        \]
        Hence trajectories cannot cross from $x\geq0$ into $x<0$.
        \item On the boundary $y=0$: From \eqref{eq:model_BT}, we have 
        \[
        \left.\frac{dy}{dt}\right|_{y=0}= 0.
        \]
        Hence trajectories cannot cross from $y\geq0$ into $y<0$.
    \end{enumerate}
    Since the vector field is continuous on $\mathbb{R}_+^2$, trajectories cannot leave the nonnegative quadrant. Hence, the nonnegative quadrant $\mathbb{R}_+^2$ is positively invariant. 
\end{proof}
\begin{theorem}[Uniform boundedness]
\label{thm:uniform_boundedness}
Assume $K,\alpha,\xi,\eta,\delta,c>0$. Then the following hold for
system \eqref{eq:model_BT}.

\begin{enumerate}[label=(\roman*)]
\item \textbf{(Prey bound).}
For any forward solution with initial condition in $\mathbb{R}_+^2$,
\[
0\le x(t)\le \max\{x(0),K\}\qquad\text{for all } t\ge 0.
\]

\item \textbf{(Uniform boundedness of predator).}
If $1<p\le 2$, then there exists a constant $M>0$, depending only on the parameters, such that
every solution with initial condition in $\mathbb{R}_+^2$ satisfies
\[
0\le y(t)\le \max\{y(0),M\}\qquad\text{for all } t\ge 0.
\]
Consequently, all solutions are uniformly bounded in $\mathbb{R}_+^2$ and are global
(i.e.\ $T_{\max}=\infty$).
\end{enumerate}
\end{theorem}
\begin{proof}
(i) Since $y\ge 0$ and $1+x+\alpha\xi>0$ for $x\ge 0$, we have
\[
\dot x
=
x\left(1-\frac{x}{K}\right)-\frac{xy}{1+x+\alpha\xi}
\le x\left(1-\frac{x}{K}\right).
\]
Comparison with the logistic equation yields
$x(t)\le \max\{x(0),K\}$ for all $t\ge 0$.

(ii) Assume $1<p\le 2$. Using (i), we may restrict to $x(t)\in[0,K]$ for all $t\ge 0$ and define
\[
A_{\max}:=\max_{0\le x\le K}\left(\eta\frac{x+\xi}{1+x+\alpha\xi}-\delta\right).
\]

Then
\[
\dot y
=
\left(\eta\frac{x+\xi}{1+x+\alpha\xi}-\delta\right)y-c\xi y^p
\le A_{\max}y-c\xi y^p.
\]
If $A_{\max}\le 0$, then $\dot y\le -c\xi y^p\le 0$ and hence $y(t)\le y(0)$ for all $t\ge 0$.

Assume $A_{\max}>0$. Since $p>1$, there exists $M>0$ such that
\[
A_{\max}y\le \frac{c\xi}{2}y^p
\qquad\text{for all } y\ge M.
\]
For instance, one may take
\[
M:=\left(\frac{2A_{\max}}{c\xi}\right)^{\frac{1}{p-1}}.
\]
Hence, whenever $y\ge M$,
\[
\dot y \le A_{\max}y-c\xi y^p \le -\frac{c\xi}{2}y^p<0,
\]
so $y(t)$ cannot cross above $\max\{y(0),M\}$. Therefore
$y(t)\le \max\{y(0),M\}$ for all $t\ge 0$.

Finally, for $1<p\le 2$ the vector field is locally Lipschitz in $\mathbb{R}_+^2$.
Since $(x(t),y(t))$ remains in a bounded subset of $\mathbb{R}_+^2$, standard continuation
implies $T_{\max}=\infty$.
\end{proof}
The results established in this section ensure that all biologically relevant solutions remain nonnegative and confined to a compact subset of $\mathbb{R}_+^2$. Consequently, through the analysis of equilibrium points and their bifurcations, the long-term dynamics of system \eqref{eq:model_BT} can be studied. We therefore proceed by investigating the existence and stability of the boundary and interior equilibria in the following section. 

\section{Equilibria and stability}\label{sec:stab_analysis}
Having established positivity and boundedness of solutions in Section~\ref{sec:prelim}, we now look at the equilibrium structure of system \eqref{eq:model_BT}. The existence and stability of equilibria in this section gives the foundation for the bifurcation analysis studied in the subsequent sections.

\subsection{Boundary equilibria}
We begin by investigating the boundary equilibria of system \eqref{eq:model_BT}. A direct calculation shows that the system admits the following boundary equilibria:
\begin{itemize}
\item $E_0=(0,0)$, corresponding to the extinction of both prey and predator populations;
\item $E_K=(K,0)$, representing the prey-only equilibrium, where predators extinct and prey settle to its carrying capacity $K$;
\item $E_P=(0,\hat{y})$, representing the predator-only equilibrium which occurs due to the presence of additional food, where $\hat{y} = \left( \frac{\eta \xi-  \delta(1+\alpha \xi) }{c\xi(1+\alpha \xi)}\right)^{\frac{1}{p-1}}$.
\end{itemize}

The local stability of these boundary equilibria is determined through the Jacobian matrix $(J)$ of system \eqref{eq:model_BT}, given by 
\begin{equation} 
 J = \begin{bmatrix}
1 - \frac{2 x}{K} -  \frac{y  \left(1+ \alpha \xi \right) } {\left(1+x+\alpha \xi\right)^2}  & \frac{- x}{1+x+\alpha \xi} \vspace{0.25cm}
  \\ 
\frac{\eta  \left(1 + \left(\alpha - 1\right) \xi\right) \ y}{(1+x+\alpha \xi)^2} &  \frac{\eta \left(x + \xi\right)}{1+x+\alpha \xi}  - \delta - cp \xi y^{p-1}
  \end{bmatrix}
\label{general_jacobian_uni}
 \end{equation}

Using the Jacobian matrix $(J)$ we now examine the existence and local stability of each boundary equilibrium.
\subsubsection{The extinction state \texorpdfstring{$(0,0)$  }{(0,0)}}
\begin{lemma}
The equilibrium point $ {E_0} = (0,0)$ always exists and is unstable. It will be unstable source if  $\delta < \dfrac{ \eta\xi }{1+ \alpha \xi}$ and a saddle if  $\delta > \dfrac{ \eta\xi }{1+ \alpha \xi}$.
\label{lem:E4_stability_unidirectional}
\end{lemma}
\begin{proof}
The proof is given in Appendix \ref{app:E4_stability_unidirectional}
\end{proof}

\subsubsection{Prey-only state \texorpdfstring{$(K,0)$}{ (K,0)}}
\begin{lemma}
The equilibrium point $E_K=(K,0)$ always exists. It is locally asymptotically stable if
$
\delta>\frac{\eta(K+\xi)}{1+K+\alpha\xi},
$
and is a saddle point if
$
\delta<\frac{\eta(K+\xi)}{1+K+\alpha\xi}.
$
\label{lem:E3_stability_unidirectional}
  \end{lemma}
\begin{proof}
 The proof is given in Appendix \ref{app:E3_stability_unidirectional}
\end{proof}
\subsubsection{Predator only state \texorpdfstring{$(0,\hat{y})$}{ (0,\hat{y})}}

\begin{lemma}
The equilibrium point $ {E_P} = (0,\hat{y})$ exists if $\delta < \dfrac{ \eta\xi }{1+ \alpha \xi}$. 
\label{lem:E1_existence_unidirectional}
  \end{lemma}
  \begin{proof}
The proof is given in Appendix \ref{app:E1_existence_unidirectional}
\end{proof}
\begin{lemma}
\label{lem:stability_pest_ext_drift}
The equilibrium point $ {E_P} = (0, \hat{y})$ is locally asymptotically stable if $\hat{y}>1+\alpha \xi$ and is a saddle if $\hat{y}< 1+\alpha \xi$. 
\end{lemma}
\begin{proof}
The proof is given in Appendix \ref{app:stability_pest_ext_drift}
\end{proof}
The boundary equilibria correspond to ecological states where one or both species are absent. Of more interest is the coexistence equilibria, where more interesting dynamics arise as both prey and predator populations persist and is dealt in the next section.

\subsection{Existence of Interior Equilibria}
To determine the possible coexistence equilibria, we rewrite system \eqref{eq:model_BT} in the equivalent form
\begin{equation}\label{eq:model}
\left\{
\begin{aligned}
\dot{x} &= q(x)[G(x)-y],\\
\dot{y} &= y \bigl[\eta R(x) - \delta - c\xi y^{p-1}\bigr],
\end{aligned}
\right.
\end{equation}
where
\begin{align}
q(x)&=\frac{x}{1+x+\alpha\xi},\quad G(x)=\left(1-\frac{x}{K}\right)(1+x+\alpha\xi),\quad R(x)=\frac{x+\xi}{1+x+\alpha\xi}.
\end{align}
Interior equilibrium $E^*=(x^*,y^*)$ with $x^*,y^*>0$ satisfy $y^*=G(x^*)$, together with
\[
\eta R(x^*)-\delta-c\xi (y^*)^{p-1}=0,
\]
Substituting $y^*=G(x^*)$ gives the scalar equation
\begin{equation}\label{eq:k_function}
k(x)
:=\eta R(x)-\delta
-c\xi \left(G(x)\right)^{p-1}
=0,
\qquad x\in(0,K).
\end{equation}
Thus, interior equilibria correspond to positive roots of \eqref{eq:k_function}. \\
The following theorem gives a sufficient condition ensuring that no interior equilibrium exists.
\begin{theorem}\label{thm:no_eq_existence}
If $k(x)<0$, for all $x\in(0,K)$, then system \eqref{eq:model_BT} admits no interior equilibrium.
\end{theorem}
\begin{proof}
The proof is given in Appendix \ref{app:no_eq_proof}.
\end{proof}
For simplicity, we denote $m := p-1 \in (0,1]$. In the rest of this section, we analyze the interior equilibria, provided that they exist. Hence, assume \(k(x)=0\) for at least some \(x\in(0,K)\). Then the equilibrium equation \eqref{eq:k_function} can be rewritten as
\[
F(x) = c\xi, \qquad x \in I,
\]
where
\[
F(x) := \frac{\eta R(x)-\delta}{(G(x))^m}, 
\qquad 
I := \{x\in(0,K): \eta R(x)-\delta>0\}.
\]
This reformulation transforms the equilibrium problem and allow us to study the existence and multiplicity of interior equilibria by analyzing the critical points of $c\xi$. Hence, we give the following theorem to classify the multiplicity and type of positive equilibria and use representative plots of $k(x)$ in Fig.~\ref{fig:root_configurations} to illustrate the various cases.
\begin{theorem}\label{thm:interior_eq}
Assume $1<p\leq2$. Then interior equilibria of \eqref{eq:model_BT} are in one-to-one correspondence with solutions of
\[
F(x)=c\xi, \quad x\in I.
\]
Moreover, the critical points of $F$ are determined by a quadratic equation
\[
Q(x)=Ax^2+Bx+C=0,
\]
where
\[
A=2m(\eta-\delta),
\]
\[
B=m\delta(K - 3(1+\alpha\xi)) - \eta(1+(\alpha-1)\xi) + m\eta(1 - K + (2 + \alpha)\xi),
\]
\[
C=\eta(1+(\alpha-1)\xi)K-m(K-(1+\alpha\xi))(\eta\xi-\delta(1+\alpha\xi)).
\]
Consequently, system \eqref{eq:model_BT} admits at most three interior equilibria.
Let $\Delta_Q:=B^2-4AC$. Then:
\begin{enumerate}
\item If $\Delta_Q<0$, at most one interior equilibrium exists;
\item If $\Delta_Q=0$, at most two interior equilibria exist;
\item If $\Delta_Q>0$, at most three interior equilibria exist.
\end{enumerate}

In the latter two cases, only roots of $Q$ lying in $I$ correspond to admissible equilibria.
\end{theorem}
\begin{proof}
From \eqref{eq:k_function}, we have
\[
\eta R(x)-\delta = c\xi (G(x))^m,
\]
which is equivalent to
\[
F(x)=c\xi.
\]

Thus, interior equilibria correspond to solutions of this equation on $I$.\\
To estimate the number of solutions, we analyze the critical points of $F$. A direct computation shows that
\[
F'(x)=0 \quad \Longleftrightarrow \quad A_1'(x)G(x)-mA_1(x)G'(x)=0,
\]
where $A_1(x)=\eta R(x)-\delta$.

Substituting explicit expressions for $A_1(x)$ and $G(x)$ reduces this condition to a quadratic equation
\[
Q(x)=Ax^2+Bx+C=0.
\]

Hence, $F$ has at most two critical points on $I$. By Rolle's theorem, the equation $F(x)=c\xi$ can have at most three solutions.

Each solution $x^*$ yields a unique interior equilibrium $(x^*,G(x^*))$, proving the result.
\end{proof}

\begin{figure}[ht]
    \centering
    \begin{minipage}[b]{0.32\textwidth}
        \centering
        \includegraphics[width=1.05\textwidth, height=4.5cm]{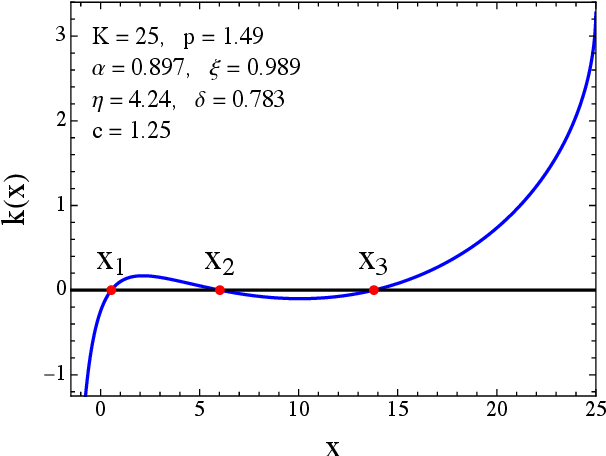}
        \\(a)
    \end{minipage}
    \hfill
    \begin{minipage}[b]{0.32\textwidth}
        \centering
        \includegraphics[width=1.05\textwidth, height=4.5cm]{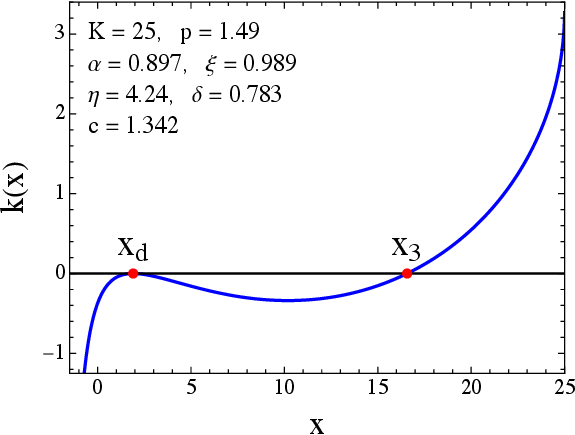}
        \\(b)
    \end{minipage}
    \hfill
    \begin{minipage}[b]{0.32\textwidth}
        \centering
        \includegraphics[width=1.05\textwidth, height=4.5cm]{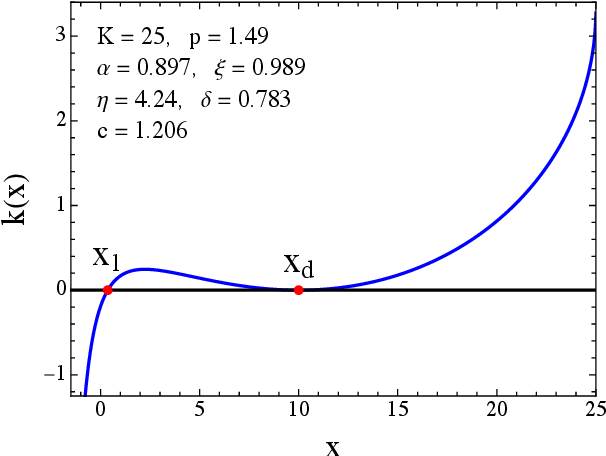}
        \\(c)
    \end{minipage}\\
\begin{minipage}[b]{0.32\textwidth}
        \centering
        \includegraphics[width=1.05\textwidth, height=4.5cm]{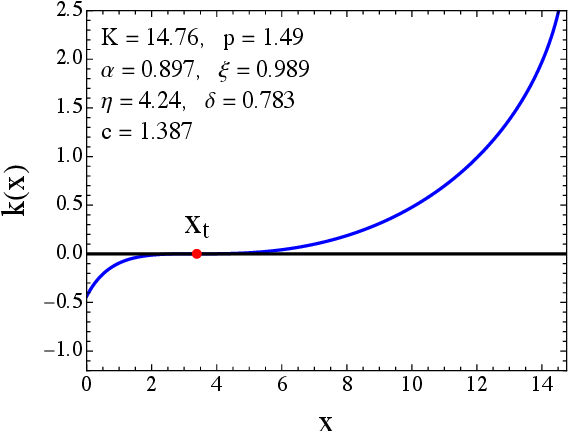}
        \\(d)
    \end{minipage}
    \hfill
    \begin{minipage}[b]{0.32\textwidth}
        \centering
        \includegraphics[width=1.05\textwidth, height=4.5cm]{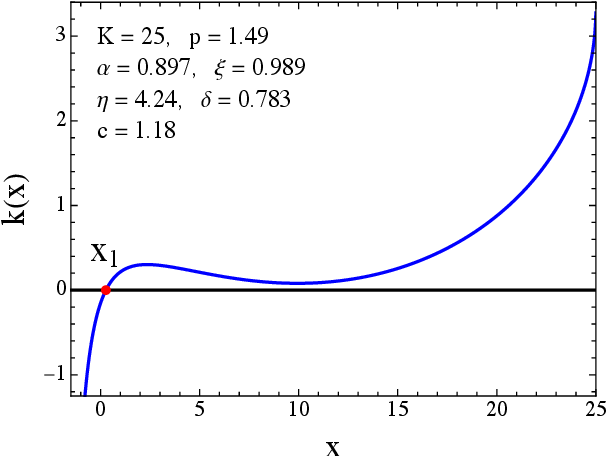}
        \\(e)
    \end{minipage}
    \hfill
    \begin{minipage}[b]{0.32\textwidth}
        \centering
        \includegraphics[width=1.05\textwidth, height=4.5cm]{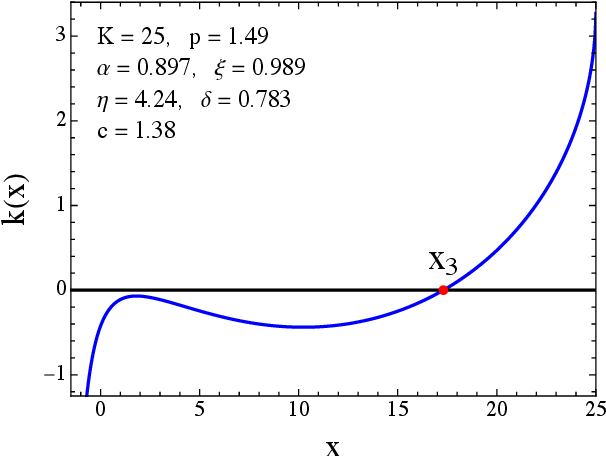}
        \\(f)
    \end{minipage}
    
    \caption{Graphs of $k(x)$ showing different positive root configurations that determine the number and multiplicity of interior equilibria of system \eqref{eq:model_BT}. Panel (a) exhibits three distinct simple roots $x_1$, $x_2$, $x_3$; panels (b) and (c) exhibit one simple root and one double root $x_d$; panel (d) exhibits a triple root $x_t$; and panels (e) and (f) exhibit a unique simple root.}
    \label{fig:root_configurations}
\end{figure}

To determine the local dynamics near an interior equilibrium, we examine the linearization of system \eqref{eq:model} and derive conditions that characterize elementary and degenerate equilibria. To do this, the Jacobian matrix of system \eqref{eq:model} is evaluated at an arbitrary positive equilibrium $(x^*,y^*)$, given by
\begin{equation} \label{vm}
J(x^*,y^*)=
\begin{bmatrix}
q(x^*)G'(x^*)+q'(x^*)(G(x^*)-y^*) & -q(x^*) \\
\eta\, R'(x^*)\,y^*  &  \eta\,  R(x^*) - \delta -p\, c\, \xi (y^*)^{(p-1)}
\end{bmatrix}.
\end{equation}
For the interior equilibrium $E^* = (x^*, y^*)$ with $y^* = G(x^*)$, we have
\begin{itemize}
    \item $\operatorname{tr}\left(J(E^*)\right)=$ $q(x^*)G'(x^*)-(p-1)\, c\, \xi\, (G(x^*))^{(p-1)}$,
    \vskip 3mm
    \item $\det\left(J(E^*)\right)=q(x^*)G(x^*)[\eta R'(x^*)-(p-1)c\xi (G(x^*))^{(p-2)}G'(x^*)]$.
\end{itemize}
Solving equation~\eqref{eq:k_function} for $c$ at the interior equilibrium, we obtain:
\begin{equation}\label{eq:cstar}
c = \frac{\eta\,  R(x^*) - \delta}{\xi \left(G(x^*)\right)^{(p-1)}}=:c^*
\end{equation}
or equivalently,
\begin{equation}\label{eq:c_star}
c^* = -\frac{1}{\xi}\left(\left(1-\frac{x^*}{K}\right)(1+x^*+\alpha  \xi)\right)^{1-p} \left(\delta -\frac{\eta  (x^*+\xi)}{1+x^*+\alpha  \xi}\right)
\end{equation}
Since for biological relevance of $c$, we need $c^*>0$. So, the above expression requires $K > 1 + 2x^* + \alpha\xi$ for $1<p\le 2$.\\
Differentiating $k(x)$, we obtain
\begin{equation}\label{eq:kprime}
k'(x^*) = \frac{\eta(1+(\alpha-1)\xi)}{(1+x^*+\alpha  \xi)^2}
- c\xi(p-1)(K - 2x^*-\alpha\xi - 1)\left(\left(1-\frac{x^*}{K}\right)(1+x^*+\alpha  \xi)\right)^{p-2}.
\end{equation}
Using \eqref{eq:cstar}, both $k'(x^*)$ and $\det\left(J(E^*)\right)$ can be written in the form
\begin{equation}\label{eq:det_compact}
\det(J(E^*)) = x^*\left(1 - \frac{x^*}{K}\right) k'(x^*).
\end{equation}
Since $x^*\left(1 - \frac{x^*}{K}\right)>0$ for $x^*\in(0,K)$, the sign of $\det (J(E^*))$ is exactly the sign of $k'(x^*)$. Consequently, the local nature of an interior equilibrium can be classified as follows.
\begin{itemize}
    \item If $k'(x^*)<0$, then $E^*$ is a hyperbolic saddle.
    
    \item If $k'(x^*)>0$, then $E^*$ is an elementary anti-saddle, i.e., a node or focus depending on the sign of the trace.
    
    \item If $k'(x^*)=0$, then $E^*$ is a degenerate equilibrium.
\end{itemize}

Theorem~\ref{thm:interior_eq} shows the existence and multiplicity of interior equilibria. From \eqref{eq:det_compact}, we can see that the sign of $\det(J(E^*))$ is determined by the sign of $k'(x^*)$. Corollary~\ref{cor:types_eq_case_1p2} below classifies the nature of these interior equilibria.

\begin{corollary}\label{cor:types_eq_case_1p2}
Assume $1<p\le 2$ and let the notation of Theorem~\ref{thm:interior_eq} hold.
\begin{enumerate}
\item If $\Delta_Q<0$, then the system has at most one interior equilibrium (see Fig. \ref{fig:root_configurations}(e)-(f))
\[
E_1=(x_1,G(x_1)).
\]
If it exists, $E_1$ is elementary. More precisely:
\begin{itemize}
\item if $k'(x_1)<0$, then $E_1$ is a saddle;
\item if $k'(x_1)>0$, then $E_1$ is an anti-saddle.
\end{itemize}

\item If $\Delta_Q=0$, then the system has at most two interior equilibria.
If $x_d$ is a repeated root of $Q$ in $I$ satisfying $k(x_d)=0$, then
\[
E_d=(x_d,G(x_d))
\]
is a degenerate equilibrium. Any other equilibrium
\[
E_1=(x_1,G(x_1))
\]
is elementary, with type determined by the sign of $k'(x_1)$ as above.

\item If $\Delta_Q>0$, then the system has at most three interior equilibria (see Fig. \ref{fig:root_configurations}(a)).
\begin{itemize}
\item If $k(x)=0$ has three distinct solutions
\[
x_1<x_2<x_3 \quad \text{in } I,
\]
then the corresponding equilibria
\[
E_i=(x_i,G(x_i)), \quad i=1,2,3,
\]
are all elementary and their types alternate according to the sign of $k'$:
\begin{itemize}
\item if $\eta>\delta$, then $E_2$ is a saddle, while $E_1$ and $E_3$ are anti-saddles;
\item if $\eta<\delta$, then $E_2$ is an anti-saddle, while $E_1$ and $E_3$ are saddles.
\end{itemize}

\item If two roots coincide, the corresponding equilibrium is degenerate (see Fig. \ref{fig:root_configurations}(b)-(c)).

\item If a triple root $x_t$ of $k(x)=0$ occurs in $I$ (see Fig. \ref{fig:root_configurations}(d)), then $E_t=(x_t,G(x_t))$ is a higher-order degenerate equilibrium.
\end{itemize}
\end{enumerate}
\end{corollary}

\section{Hopf Bifurcation}\label{sec:Hopf_bfn}
From Corollary \eqref{cor:types_eq_case_1p2} in section \ref{sec:stab_analysis}, the type of an interior equilibrium $E^*=(x^*,G(x^*))$ is determined by the sign of $k'(x^*)$. In particular:
\begin{itemize}
\item if $k'(x^*)<0$, then $E^*$ is a hyperbolic saddle;
\item if $k'(x^*)>0$, then $E^*$ is an elementary anti-saddle.
\end{itemize}
Since Hopf bifurcation can only occur at non-saddle equilibria, we restrict our attention to interior equilibria satisfying $k'(x^*)>0$. For such equilibria, the stability is determined by the trace of the Jacobian matrix. In particular, if $\operatorname{tr} J(x^*,G(x^*))=0$ and $\det J(x^*,G(x^*))>0$, then $E^*$ is a center or weak focus. Therefore, a change in the sign of $\operatorname{tr} J$ may lead to a Hopf bifurcation.

This subsection focuses on studying the occurrence of a Hopf bifurcation at the interior equilibrium $E^*=(x^*,G(x^*))$ with respect to the bifurcation parameter $\delta$. 

\medskip

We start by substituting \eqref{eq:cstar} into trace equation, we get:
\begin{equation}\label{eq.trace}
\begin{aligned}
\operatorname{tr}\left(J(x^*, G(x^*))\right)&=q(x^*)G'(x^*)+(p-1)\left(\delta-\eta R(x^*)\right)\\
\end{aligned}
\end{equation}
Solving $\operatorname{tr}\left(J(x^*, G(x^*))\right)=0$ for $\delta$ gives
\begin{equation}\label{eq.delta_star}
\delta=\eta R(x^*)-\frac{G'(x^*)q(x^*)}{(p-1)}=:\delta_H
\end{equation}
Since,
\begin{equation} \label{eq:det_1}
\det\left(J(x^*, G(x^*))\right) = q(x^*)G(x^*)[\eta R'(x^*)-(p-1)c\xi (G(x^*))^{(p-2)}G'(x^*)]
\end{equation}
At an interior equilibrium, we already have 
\begin{equation*}
    q(x^*)>0, \quad G(x^*)>0 \quad \implies q(x^*)G(x^*)>0
\end{equation*}
So, we have
\begin{equation}
    \det\left(J(x^*, G(x^*))\right)> 0 \quad \Longleftrightarrow \quad \eta R'(x^*)-(p-1)c\xi (G(x^*))^{(p-2)}G'(x^*)>0.
\end{equation}
A sufficient condition for this inequality to hold in given in the following theorem.
\begin{proposition}\label{prop:det_positive}
For $x^*>0$, $K>x^*$, $\xi>0$, $\alpha>0$, and $\rho=1+(\alpha-1)\xi>0$, if $\eta>\frac{x^*\bigl(1-K+2x^*+\alpha \xi\bigr)^2}{K(K-x^*)\rho}$ then $\det\big(J(x^*,G(x^*))\big)>0$.
\end{proposition}
Interior equilibrium $E^*=(x^*,G(x^*))$ satisfy
\begin{equation}
    k(x,\delta) :=\eta\,  R(x) - \delta - c\, \xi \left(G(x)\right)^{(p-1)} = 0
\end{equation}
Differentiating implicitly w.r.t. $\delta$, we get
\begin{equation}
    k_x(x^*,\delta)\,\frac{dx^*}{d\delta}+k_\delta(x^*,\delta)=0
\end{equation}
Since $k_\delta(x^*,\delta)=-1$, and $k_x(x,\delta)=\eta R'(x)-(p-1)c\xi (G(x))^{(p-2)}G'(x)$, then by proposition \eqref{prop:det_positive}, at the Hopf point,
\begin{equation}
    k_x(x^*,\delta_H)=\eta R'(x^*)-(p-1)c\xi (G(x^*))^{(p-2)}G'(x^*)>0.
\end{equation}
So, we have
\begin{equation}
    \frac{dx^*}{d\delta}=\frac{1}{k_x(x^*,\delta)}.
\end{equation}
We next prove the transversality condition confirming the existence of Hopf bifurcation. 
\begin{theorem}\label{thm:hopf_tranversality}
 Let $E^*=(x^*,G(x^*))$ be an interior equilibrium corresponding to $\delta=\delta_H$. If $k_x(x^*,\delta_H)\neq 0$ and $q'(x^*)G'(x^*)+q(x^*)G''(x^*)-(p-1)^2c\xi \big(G(x^*)\big)^{p-2}G'(x^*)\neq 0$, then the transversality condition holds. Consequently, system \eqref{eq:model} undergoes a Hopf bifurcation at $E^*$ as $\delta$ passes through $\delta_H$.
\end{theorem}
\begin{proof}
    The proof is given in Appendix \ref{app:hopf_transversality_proof}.
\end{proof}
To analyze the degenerate Hopf bifurcation, one must compute the corresponding Lyapunov coefficients. However, due to the algebraic complexity of the model, these computations lead to highly complicated Lyapunov quantities, making it difficult to determine their signs analytically. We therefore omit the detailed degenerate Hopf bifurcation analysis here. However, in the following section, we establish the existence of a codimension-3 Hopf bifurcation emerging from the Bogdanov-Takens bifurcation of codimension-4.


\section{Higher-Order Bifurcations}\label{sec:higher_order}
In this section, we study higher-order bifurcations possible for the system \eqref{eq:model} due to the presence of double positive equilibrium $E_d$ (corresponding to the cases Fig. \ref{fig:root_configurations}(b)-(c)) and triple positive equilibrium $E_t$ (corresponding to the case Fig. \ref{fig:root_configurations}(d)). By Corollary~\ref{cor:types_eq_case_1p2}, we know such equilibria arise when $F(x)=c\xi$ has a multiple root. We will first analyze the double positive equilibrium, $E_d$.

\subsection{Double positive equilibrium}\label{sec:double_positive}
In this section, we focus on Bogdanov-Takens singularities that may arise from the coalescing of two positive equilibria and we denote it by $E_d$. For the system \eqref{eq:model} to undergo a Bogdanov-Takens (BT) bifurcation at an interior equilibrium $E_d(x_d,y_d)$, the Jacobian matrix $J$ at that point must possess a double zero eigenvalue organized in a single Jordan block. Equivalently, this requires that both the determinant and the trace of $J$ vanish simultaneously. That is, 
\[
\det\left(J(E_d)\right) = 0 \quad \text{and} \quad \operatorname{tr}\left(J(E_d)\right) = 0.
\]
The trace of the Jacobian matrix is given by
\begin{equation}\label{eq:trace_expr}
\operatorname{tr}\!\left(J\right)
= \frac{x_d\,\bigl( K -(1+ 2x_d +\alpha\xi)\bigr)}{K\,(1+ x_d + \alpha\xi )}
- \frac{%
\begin{aligned}
&\delta\,(p-1)\,\bigl(1+(\alpha-1)\xi\bigr)\,(K - x_d)
\end{aligned}
}{%
\begin{aligned}
&K\bigl(\xi (p-\alpha) + (p-1)x_d - 1\bigr) + \xi + 2x_d\\
& - p(\xi + x_d)(1 + 2x_d + \alpha\xi)
+ (\xi + 2x_d)(x_d+\alpha\xi) 
\end{aligned}
}\,.
\end{equation}
Solving $\operatorname{tr}\left(J\right)=0$ for $\delta$ yields the critical value:
\begin{equation}\label{eq:deltastar}
\delta = 
\frac{%
\begin{aligned}
&x_d\bigl( K-(1 + 2x_d +\alpha\xi)\bigr)\Bigl(
K\bigl(\xi (p-\alpha) + (p-1)x_d - 1\bigr) + \xi + 2x_d \\
&\qquad\quad -\, p(\xi + x_d)(1 + 2x_d + \alpha \xi)
+ (\xi + 2x_d)( x_d+\alpha \xi) 
\Bigr)
\end{aligned}
}{%
\begin{aligned}
&K (p-1)\bigl(1+(\alpha - 1)\xi \bigr)\,(K - x_d)\,(1 + x_d + \alpha \xi)
\end{aligned}
}
\;:=\; \delta^* .
\end{equation}
To ensure that the parameter $\delta$ is biologically feasible, we give the following proposition.
\begin{proposition}\label{prop:eta-delta-positive}
For $1<p\leq2$, $x_d>0$, $\xi>0$, $\alpha>0$, let
$\rho:=1+(\alpha-1)\xi$. Assume that 
\begin{equation*}
    0<\rho<x_d+\xi, \quad \text{and} \quad K>
\frac{2{x_d}^2+3x_d\xi+\xi(1+\alpha \xi)}
{x_d+\xi-\rho}.
\end{equation*}
If, in addition,
\[
\max\left\{1,\,
\frac{(K-2x_d-\xi)(1+x_d+\alpha \xi)}
{(x_d+\xi)(K-(1+2x_d+\alpha \xi))}
\right\}<p\le 2,
\]
then \(\delta^{*}>0\).
\end{proposition}

Next, from \eqref{eq:det_compact}, we have
\begin{align}
    \det\left(J\right)&=x_d \left(1 - \frac{x_d}{K} \right)k'(x_d)=0\quad \implies k'(x_d)=0.
\end{align}
Solving $k'(x_d)=0$ for $\eta$, we obtain the critical value:
\begin{equation}
    \eta=\frac{x_d (K -(1+ 2x_d +\alpha\xi))^2}{K (K-x_d) (1+(\alpha -1) \xi) }:=\eta^*>0, \quad \text{if} \quad (1+(\alpha -1) \xi) >0.
\end{equation}

Under the assumptions $\eta=\eta^*$ and $\delta=\delta^*$, we reduce system \eqref{eq:model} to its normal form near $E_d$. To do so, we begin by shifting the equilibrium $E_d$ to the origin via the affine transformation $x_1=x-x_d$ and $y_1=y-G(x_d)$. Expanding the transformed system in a Taylor series about the origin then yields the following system:
\begin{eqnarray}\label{eq:s1}
\left\{\begin{aligned}
\dot{x_1}= &\ q(x_d)G'(x_d)x_1 -q(x_d)y_1 +\frac{q(x_d)G''(x_d)+2q'(x_d)G'(x_d)}{2}x_1^2\\
&\ -q'(x_d)x_1y_1 + O(|(x_1,y_1)^3|),\\
\dot{y_1}=&\ \eta G(x_d)R'(x_d)x_1-(p-1)c\xi (G(x_d))^{(p-1)}y_1+ \frac{\eta G(x_d) R''(x_d)}{2}x_1^2 \\
&\  +\eta R'(x_d)x_1y_1 -\frac{p(p-1)c\xi (G(x_d))^{p-2}}{2}y_1^2 +O(|(x_1,y_1)^3|),
\end{aligned}\right.
\end{eqnarray}
Assuming that the following Bogdanov-Takens (BT) bifurcation conditions hold,
\begin{equation}\label{eq:BT_cond}
    q(x_d)G'(x_d)-(p-1)c \xi (G(x_d))^{(p-1)}=0\quad \text{and}\quad q(x_d)G(x_d)[\eta R'(x_d)-(p-1)c\xi (G(x_d))^{(p-2)}G'(x_d)]=0,
\end{equation} 
system \eqref{eq:s1} reduces to
\begin{eqnarray}\label{eq:s2}
\left\{\begin{aligned}
\dot{x_1}= &\ m_0x_1 +n_0y_1 +m_1x_1^2+m_2x_1y_1 + O(|(x_1,y_1)^3|),\\
\dot{y_1}=&\ -\frac{m_0^2}{n_0}x_1 -m_0y_1 +n_1x_1^2+n_2x_1y_1+n_3y_1^2 + O(|(x_1,y_1)^3|),
\end{aligned}\right.
\end{eqnarray}
where $n_0=-q(x_d)\neq0$ and the other coefficients are given by
\begin{align*}
m_0&=q(x_d)G'(x_d), \quad m_1=\frac{q(x_d)G''(x_d)+2q'(x_d)G'(x_d)}{2}, \quad m_2=-q'(x_d),\\
n_1&=\frac{\eta G(x_d) R''(x_d)}{2}, \quad n_2=\eta R'(x_d), \quad n_3=-\frac{p(p-1)c\xi (G(x_d))^{p-2}}{2}.
\end{align*}
Hence, in order to transform system \eqref{eq:s2} into a more standard canonical form, we introduce the change of variables
\begin{equation}
    x_1=x_2 \quad \text{and} \quad y_1=\frac{1}{n_0}(y_2-m_0x_2)
\end{equation}
With this transformation, the system can be rewritten as follows.
\begin{eqnarray}\label{eq:s3}
\left\{\begin{aligned}
\dot{x_2}= &\ y_2 +a_{20}x_2^2+a_{11}x_2y_2 +  O(|(x_2,y_2)^3|),\\
\dot{y_2}=&\ b_{20}x_2^2+b_{11}x_2y_2+b_{02}y_2^2 + O(|(x_2,y_2)^3|),
\end{aligned}\right.
\end{eqnarray}
where
\begin{align*}
a_{20}&=\left(m_1- \frac{m_0m_2}{n_0} \right), \quad a_{11}=\frac{m_2}{n_0}, \quad b_{02}=\frac{n_3}{n_0},\\
b_{20}&=\left(n_0n_1+m_0(m_1-n_2) + \frac{m_0^2 (n_3-m_2)}{n_0} \right), \quad b_{11}=\left( \frac{m_0m_2+n_0n_2-2m_0n_3 }{n_0} \right).
\end{align*}
Further, we apply a near-identity transformation of the form
\begin{eqnarray}\label{eq.t2}
\begin{aligned}
x_2&=x_2+\left(\frac{a_{11}+b_{02}}{2}\right)x_2^2,\\
y_2&= y_2 -a_{20}x_2^2+b_{02}x_2y_2,
\end{aligned}
\end{eqnarray}
and use Remark $1$ of section $2.13$ in \citep{perko2013differential}, the system \eqref{eq:s3} is transformed to the standard normal-form:
\begin{eqnarray}\label{eq:s4}
\left\{\begin{aligned}
\dot{x_2}= &\ y_2+ O(|(x_2,y_2)^3|),\\
\dot{y_2}=&\ \gamma_1x_2^2 + \gamma_2x_2y_2 + O(|(x_2,y_2)^3|),
\end{aligned}\right.
\end{eqnarray}
where the normal-form coefficients are given by
\begin{align}
    \gamma_1&=b_{20}=\left(n_0n_1+m_0(m_1-n_2) + \frac{m_0^2 (n_3-m_2)}{n_0} \right)\notag\\ 
    &=\frac{q(x_d)\Bigl(q(x_d)G'(x_d)(pG'(x_d)^2 +G(x_d)G''(x_d))-\eta^*G(x_d)(2G'(x_d)R'(x_d)+G(x_d)R''(x_d))\Bigr)}{2G(x_d)},\\
    \gamma_2&=2a_{20}+b_{11}=2m_1+n_2-\frac{m_0(m_2+2n_3)}{n_0}\notag\\
    &=q'(x_d)G'(x_d)+ \eta^*\, R'(x_d)+ q(x_d)G''(x_d)-\frac{pq(x_d)G'(x_d)^2}{G(x_d)}.
\end{align}
Based on these coefficients, we give the following theorem to prove that $E_d$ is a cusp of codimension-$2$.
\begin{theorem}
    Let $(\eta,\delta)=(\eta^*,\delta^*)$ and the conditions in Proposition \ref{prop:eta-delta-positive} hold. If $\gamma_1\gamma_2\neq0$, then $E_d$ is a cusp singularity of codimension-$2$ for any $p\in(1,2]$.
\end{theorem}
\begin{proof}
    The result follows directly from Theorem $3$ in Section $2.11$ in \cite{perko2013differential} applied to the normal form \eqref{eq:s4}.
\end{proof}

We now analyze the behavior of the coefficients $\gamma_1$ and $\gamma_2$ in \eqref{eq:s4}. Essentially, if $\gamma_1\gamma_2=0$, then either $\gamma_1=0$ or $\gamma_2=0$. It is observed that the signs of $\gamma_1$ and $\gamma_2$ depends on the parameter values. First, setting $\gamma_2=0$ and solving for $K$ yields
\begin{equation}\label{eq:Kstar}
K = -\frac{%
\begin{aligned}
&1 + (6-4p){x_d}^2 + \alpha \xi(2 + \alpha \xi) - (2p - 7)(1+\alpha \xi)x_d \\
&\quad + (1 + x_d + \alpha \xi)\sqrt{ (12 - 8p){x_d}^2 + 4x_d(1+\alpha \xi) + (1+\alpha \xi)^2}
\end{aligned}
}{2\bigl((p-1)x_d - (1+\alpha \xi)\bigr)} 
:= K^* .
\end{equation}
Following proposition confirms the biological relevance of the parameter $K$.
\begin{proposition}
Assume that $\alpha>0$, $\xi>0$, and $1<p\leq2$. Then $K^*>0$ holds for all $0<x_d\le1$. And for $x_d>1$, the condition $K^*>0$ is satisfied whenever
\begin{enumerate}[label=(\roman*)]
    \item $\xi \geq \frac{x_d-1}{\alpha}$, for all $1<p\leq2$, or\\
    \item $0<\xi<\frac{x_d-1}{\alpha}$, then $p<\frac{1+x_d+\alpha\xi}{x_d}$. 
\end{enumerate}
\end{proposition}

We next consider the degenerate case $\gamma_2=0$. For the value of $K$ as in \eqref{eq:Kstar}, the equilibrium $E_d(x_d,G(x_d))$ remains a cusp singularity, but the vanishing of $\gamma_2=0$ introduces additional degeneracy. To determine the precise codimension, we further normalize system \eqref{eq:s4} and compute the higher-order normal form coefficients.

When $(\eta,\delta,K)=(\eta^*,\delta^*,K^*)$, system \eqref{eq:s4} undergoes $C^\infty$ change of coordinates of the form $x_3=x_2$ and $y_3= \dot{x_2}$ which further reduces system to (expanded up to fifth order) 
\begin{eqnarray}\label{eq:s5}
\left\{\begin{aligned}
\dot{x_3}= &\ y_3,\\
\dot{y_3}=&\ c_{20}x_3^2 + c_{30}x_3^3 + c_{21}x_3^2y_3 + c_{12}x_3y_3^2 + c_{03}y_3^3 + c_{40} x_3^4 + c_{31}x_3^3y_3\\ & + c_{13}x_3y_3^3 + c_{22}x_3^2y_3^2+c_{04}y_3^4+c_{50} x_3^5 + c_{41}x_3^4y_3 + c_{14}x_3y_3^4\\
&+ c_{32}x_3^3y_3^2+ c_{23}x_3^2y_3^3+c_{05}y_3^5+O(|(x_3,y_3)^6|),
\end{aligned}\right.
\end{eqnarray}
where all the coefficient expressions are given in Appendix \ref{Coeff_0}.\\
The system \eqref{eq:s5} needs to be normalized again. So, we apply a near-identity transformation of the form
\begin{equation}\label{eq:t4}
    x_3 \;\mapsto\; x_3+\phi(x_3,y_3), 
    \qquad 
    y_3 \;\mapsto\; y_3+\psi(x_3,y_3),
\end{equation}
where $\phi$ and $\psi$ are smooth polynomials chosen so as to eliminate $O(y_3^3)$ terms in the $\dot{y}_3$ equation and are given below.
\begin{eqnarray}\label{functions1}
\begin{aligned}
\phi(x_3,y_3)&=\left(\frac{c_{03}}{2}\right)x_3^2y_3 + \left(\frac{c_{04}}{2}\right)x_3^2y_3^2+ \left(\frac{c_{05}}{2}\right)x_3^2y_3^3 + \left(\frac{c_{13}}{6}\right)x_3^3y_3 \\
& + \left(\frac{3c_{03}^2+c_{14}}{6}\right)x_3^3y_3^2 + \left(\frac{c_{23}-c_{03}c_{12}}{12}\right)x_3^4y_3 - \left(\frac{2c_{20}c_{04}+c_{03}c_{21}}{10}\right)x_3^5,\\
\psi(x_3,y_3)&=c_{03}x_3y_3^2+c_{04}x_3y_3^3+c_{05}x_3y_3^4+\left(\frac{c_{13}}{2}\right)x_3^2y_3^2+ \left(\frac{3c_{03}^2+c_{14}}{2}\right)x_3^2y_3^3\\
&+ \left(\frac{c_{03}c_{20}}{2}\right)x_3^4+ \left(\frac{2c_{23}+c_{03}c_{12}}{6}\right)x_3^3y_3^2+ \left(\frac{c_{20}c_{13}+3c_{03}c_{30}}{6}\right)x_3^5.
\end{aligned}
\end{eqnarray}
Then the system \eqref{eq:s5} reduces to
\begin{eqnarray}\label{eq.reduced_system4}
\left\{\begin{aligned}
\dot{x_3}= &\ y_3,\\
\dot{y_3}=&\ d_{20}x_3^2 + d_{30}x_3^3 + d_{21}x_3^2y_3 + d_{12}x_3y_3^2 + d_{40} x_3^4 + d_{31}x_3^3y_3 + d_{22}x_3^2y_3^2\\ 
& + d_{50} x_3^5+ d_{41}x_3^4y_3+ d_{32}x_3^3y_3^2 +O(|(x_3,y_3)^6|),
\end{aligned}\right.
\end{eqnarray}
where
\begin{align*}
d_{20}&=c_{20}=\gamma_1, \quad d_{30}=c_{30}, \quad d_{21}=c_{21}, \quad d_{12}=c_{12}, \quad d_{40}=c_{40}, \quad d_{31}=c_{31}-3c_{20}c_{03},\\ 
&d_{22}=c_{22},\quad  d_{50}=c_{50}, \quad d_{41}=c_{41}-3c_{30}c_{03}-\frac{3c_{20}c_{13}}{2}, \quad d_{32}=c_{32}-2c_{20}c_{04}.
\end{align*}
Next, by introducing the transformation
\begin{align}
    x=x_3, \quad y=\frac{y_3}{\sqrt{d_{20}}}, \quad d\tau=\sqrt{d_{20}}dt,
\end{align}
we obtain the following system up to fifth order,
\begin{eqnarray}\label{eq:s6}
\left\{\begin{aligned}
\dot{x}= &\ y,\\
\dot{y}=&\ x^2 + e_{30}x^3+e_{21}x^2 y + e_{12}xy^2 + e_{40}x^4 + e_{31}x^3y + e_{22}x^2y^2 \\
&+ e_{50}x^5 + e_{41}x^4y + e_{32}x^3y^2 + O(|(x,y)^6|),\\
\end{aligned}\right.
\end{eqnarray}
where
\begin{align*}
e_{30}&=\frac{d_{30}}{d_{20}}, \quad e_{21}=\frac{d_{21}}{\sqrt{d_{20}}}, \quad e_{12}=d_{12}, \quad e_{40}=\frac{d_{40}}{d_{20}}, \quad e_{31}=\frac{d_{31}}{\sqrt{d_{20}}}, \\
e_{22}&=d_{22}, \quad e_{50} =\frac{d_{50}}{d_{20}}, \quad e_{41}=\frac{d_{41}}{\sqrt{d_{20}}}, \quad e_{32}=d_{32}.
\end{align*}
According to Proposition $5.3$ in \cite{lamontagne2008bifurcation}, the system \eqref{eq:s6} is equivalent to
\begin{eqnarray}\label{eq:s7}
\left\{\begin{aligned}
\dot{x}= &\ y ,\\
\dot{y}=&\ x^2+\gamma_3x^3y +O(|(x,y)^5|)
\end{aligned}\right.
\end{eqnarray}
where
\begin{align}
\gamma_3&=e_{31}-e_{30}e_{21}.
\end{align}
Hence, following the normal form theory in \cite{chow1994normal}, we observe that when $\gamma_3\neq0$ the equilibrium point $E_d(x_d,G(x_d))$ behaves as a cusp of codimension-$3$. If $\gamma_3=0$, then we get another higher degeneracy, however, solving \(\gamma_3\) for one of the parameters is not straightforward due to the complexity of the expression. Hence, we present a set of parameters that prove the existence of a cusp singularity of codimension-4. Indeed, if we choose \(p=1.4\) and
\begin{equation*}
(\eta, \delta, c, K, \alpha, \xi) = (3.5700505,\,1.6070648,\,1.3311815,\, 46.409597,\,0.26,\,0.675),
\end{equation*}
then the corresponding equilibrium point is
\begin{equation*}
    E_d=(2.1983868,\,3.2140684).
\end{equation*}
For this parameter set, we obtain \(\operatorname{tr}(J)=0\) and \(\gamma_2=0\). Moreover, by plotting \(\gamma_3\) as a function of \(x_d\), as shown in Fig. \ref{fig:gamma3}, we see that \(\gamma_3\) could also vanish.

\begin{figure}[htbp]
    \centering
    \includegraphics[width=0.6\textwidth, height=6cm]{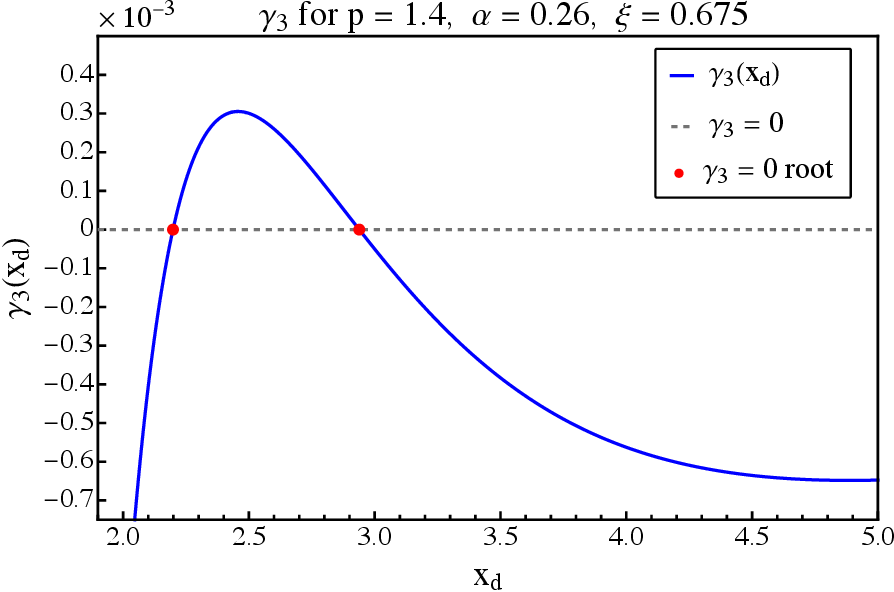}
    \caption{
    The function $\gamma_3$ as a function of $x_d$ 
    for $p=1.4$.}
    \label{fig:gamma3}
\end{figure}

Therefore, by assuming \(\gamma_3=0\) and applying Proposition $1$ in \cite{shang2024bogdanov}, the system \eqref{eq:s7} further reduces to
\begin{eqnarray}\label{eq:s8}
\left\{\begin{aligned}
\dot{x}= &\ y ,\\
\dot{y}=&\ x^2+\gamma_4x^4y +O(|(x,y)^6|)
\end{aligned}\right.
\end{eqnarray}
where
\begin{align}
\gamma_4&=e_{41}+\frac{1}{2}e_{12}e_{21}-e_{21}e_{40}.
\end{align}
For the above parameter set, $\gamma_4=-2.92196\neq0$. Therefore, again by normal form theory in \cite{chow1994normal} and by the method in \cite{li1989system, joyal1990cusp}, $E_d(x_d, G(x_d))$ corresponds to a cusp-type Bogdanov-Takens bifurcation of codimension at least $4$.

The following theorem summarizes the degeneracy of the cusp singularity.
\begin{theorem}\label{thm:codim_4}
For $1<p\leq2$, if $(\eta,\delta,K)=(\eta^*,\delta^*,K^*)$, then  the interior equilibrium $E_d(x_d, G(x_d))$ is a cusp singularity, and the following results hold:
\begin{enumerate}
  \item If $\gamma_3 \ne 0$, then $E_d(x_d, G(x_d))$ is a cusp of codimension-$3$,
  \item If $\gamma_3 = 0$, and $\gamma_4 \neq 0$, then $E_d(x_d, G(x_d))$ is a cusp of codimension-$4$.
\end{enumerate}
\end{theorem}
The preceding analysis shows that $E_d(x_d, G(x_d))$ is a cusp-type Bogdanov-Takens singularity of codimension at least $4$. Now, our aim is to determine whether system \eqref{eq:model_BT} exhibits a complete unfolding of a codimension-$4$ cusp-type BT bifurcation under small perturbations. Our next section is dedicated to this derivation.

\subsection{Unfolding of BT Bifurcation of Codimension 4}
From Theorem \ref{thm:codim_4}, we know that the system \eqref{eq:model_BT} is a cusp of codimension at least $4$, so in this subsection, we will prove that there exists a degenerate cusp type BT bifurcation of codimension at least $4$ around the double positive equilibrium $E_d$. To do so, for $p=1.4$, we use the parameter set
\begin{equation}\label{eq:paramset}
    (\eta_0, \delta_0, c_0, K_0, \alpha_0, \xi_0) = (3.5700505,\,1.6070648,\,1.3311815,\, 46.409597,\,0.26,\,0.675),
\end{equation}
at which the equilibrium is $E_d=(2.1983868,\,3.2140684)$.

To develop the universal unfolding for cusp singularity of order 4 near the point $E_d$, set
\begin{equation}
    \begin{aligned}
        \eta=\eta_0 + \lambda_1, \quad \delta=\delta_0 + \lambda_2, \quad K=K_0 + \lambda_3, \quad  c=c_0 + \lambda_4,
    \end{aligned}
\end{equation}
where $\lambda=(\lambda_1,\lambda_2,\lambda_3,\lambda_4)$ is a sufficiently small perturbation vector.

We now analyze the corresponding unfolding system
\begin{equation}\label{eq:perturbed_system}
\left\{\begin{array}{l}
\displaystyle \dot{x}=x \left (1-\frac{x}{K_0 + \lambda_3}\right )-\frac{x y}{1+x+\alpha_0\xi_0} ,\\
\displaystyle \dot{y}=(\eta_0 + \lambda_1)\,y\,\left[ \left ( \frac{x+\xi_0}{1+x+\alpha_0\xi_0} \right)\ -(\delta_0 + \lambda_2)\, - (c_0+ \lambda_4)\, \xi_0\, y^{(p-1)}\right].
\end{array}\right.
\end{equation}

\begin{theorem}\label{thm:unfolding}
    For $\lambda=(\lambda_1,\lambda_2,\lambda_3,\lambda_4)$ sufficiently small, system \eqref{eq:perturbed_system} is $C^\infty$ equivalent to 
    \begin{equation}\label{eq:unfolding_4NF}
\left\{\begin{array}{l}
\displaystyle \dot{x}=y,\\
\displaystyle \dot{y}=\mu_1(\lambda)+\mu_2(\lambda)y+\mu_3(\lambda)xy+\mu_4(\lambda)x^3y+x^2-x^4y+R(x,y,\lambda),
\end{array}\right.
\end{equation}
where 
$R(x,y,\lambda)= y^{2} O\!\bigl(|(x,y)|^{2}\bigr) + O\!\bigl(|(x,y)|^{6}\bigr) + O(\lambda)\Bigl(O(y^{2}) + O\!\bigl(|(x,y)|^{3}\bigr)\Bigr) + O(\lambda^{2})\,O\!\bigl(|(x,y)|\bigr)$
and
\begin{equation}
   \left.\frac{\partial(\mu_1,\mu_2,\mu_3,\mu_4)}{\partial(\lambda_1,\lambda_2,\lambda_3,\lambda_4)}\right|_{\lambda=(0,0,0,0)}\neq0.
\end{equation}
Hence, system \eqref{eq:perturbed_system} provides a universal unfolding of cusp singularity of codimension at least $4$.
\end{theorem}
\begin{proof}
    We first translate the equilibrium point $E_d$ to the origin by introducing the change of variables
    \[
    (X, Y) \;\mapsto\; \left(x-2.1983868,\, y-3.2140684\right)
    \]
    for system \eqref{eq:perturbed_system} and we then expand the resulting system into a Taylor series around the origin, retaining terms upto sixth order. We obtain
\begin{equation}\label{eq:unfolding_system1}
\left\{\begin{array}{ll}
\displaystyle \dot{X} &= f_{00}+f_{10}X+f_{01}Y+f_{20}X^2+f_{11}XY+f_{30}X^3+f_{21}X^2Y+f_{40}X^4+f_{31}X^3Y \\
&\ + f_{50}X^5 + f_{41}X^4Y + O(|X,Y,\lambda|^6),\\[1ex]
\displaystyle \dot{Y} &=g_{00}+g_{10}X+g_{01}Y+g_{20}X^2+g_{11}XY+g_{02}Y^2+g_{30}X^3+g_{21}X^2Y+g_{03}Y^3\\
&\ +g_{40}X^4 +g_{31}X^3Y + g_{04}Y^4 + g_{50}X^5 + g_{41}X^4Y + g_{05}Y^5 + O(|X,Y,\lambda|^6).
\end{array}\right.
\end{equation}
where $f_{ij}$ and $g_{ij}$ are smooth functions and their expressions are omitted here for brevity. Also, $f_{00}(0)=g_{00}(0)=0$.\\
We then apply the transformation $u=X$ and $v=\dot{X}$ to the above system to get
\begin{equation}\label{eq:unfolding_system2}
\left\{\begin{array}{ll}
\displaystyle \dot{u} &= v, \\
\displaystyle \dot{v} &=h_{00}+h_{10}u+h_{01}v+h_{20}u^2+h_{11}uv+h_{02}v^2+h_{30}u^3+h_{21}u^2v+h_{12}uv^2+h_{03}v^3\\
&\ +h_{40}u^4 +h_{31}u^3v+h_{22}u^2v^2+h_{13}uv^3+h_{04}v^4 + h_{50}u^5 + h_{41}u^4v + h_{32}u^3v^2\\
&\ + h_{23}u^2v^3+ h_{14}uv^4+h_{05}v^5 +  O(|u,v,\lambda|^6),\\
\end{array}\right.
\end{equation}
where $h_{ij}$ can be expressed in terms of $f_{ij}$ and $g_{ij}$ and the exact forms are given in Appendix \ref{Coeff_1}. \\
Under the near-identity transformation
\begin{equation}\label{eq:unfolding_transf_2}
    \begin{aligned}
    u&=x_1+\Phi(x_1,y_1),\\
    v&=y_1+\Psi(x_1,y_1),
\end{aligned}
\end{equation}
where the functions $\Phi$ and $\Psi$ are given in Appendix \ref{Coeff_2}, the system \eqref{eq:unfolding_system2} is transformed to
\begin{equation}\label{eq:unfolding_system3}
\left\{\begin{array}{ll}
\displaystyle \dot{x_1} &= y_1, \\
\displaystyle \dot{y_1} &=p_{00}+p_{10}x_1+p_{01}y_1+p_{20}x_1^2+p_{11}x_1y_1+p_{30}x_1^3+p_{21}x_1^2y_1+p_{40}x_1^4  \\
&\ +p_{31}x_1^3y_1 + p_{50}x_1^5 + p_{41}x_1^4y_1 + O(|x_1,y_1,\lambda|^6),\\
\end{array}\right.
\end{equation}
and one can refer to Appendix \ref{Coeff_3} for the coefficients $p_{ij}$. To remove $x_1^3$, $x_1^4$, and $x_1^5$ terms, since $p_{20}(0)\neq0$, we use the differential form $\sum_{i=2}^{5}p_{i0}x_1^idx_1=p_{20}x_2^2dx_2$ which gives the following 
\begin{equation}
    \begin{aligned}
    x_1=\zeta(x_2)&:=x_2-\frac{p_{30}}{4p_{20}}x_2^2+\frac{15p_{30}^2-16p_{20}p_{40}}{80p_{20}^2}x_2^3-\frac{175 p_{30}^3+160p_{20}^2p_{50}-336p_{20}p_{30}p_{40}}{960p_{20}^3}x_2^4\\
\end{aligned}
\end{equation}
Then by the following change of coordinates
\begin{equation}
\begin{aligned}
x_1 &= \zeta(x_2), \qquad y_1 = y_2,\\
d\tau &= \Bigl(
1 +\tfrac{p_{30}}{2\,p_{20}}\,x_2+\left(\tfrac{3\,p_{40}}{5\,p_{20}} - \tfrac{5\,p_{30}^{2}}{16\,p_{20}^{2}} \right)x_2^{2}    \\
&\quad+ \Bigl( \tfrac{7\,p_{30}^{3}}{24\,p_{20}^{3}} - \tfrac{4\,p_{30}p_{40}}{5\,p_{20}^{2}} + \tfrac{2\,p_{50}}{3\,p_{20}} \Bigr) x_2^{3} \\
&\quad+ \Bigl( \tfrac{527\,p_{30}^{4}}{768\,p_{20}^{4}} - \tfrac{13\,p_{30}^{2}p_{40}}{8\,p_{20}^{3}} + \tfrac{9\,p_{40}^{2}}{25\,p_{20}^{2}}
            + \tfrac{2p_{30}p_{50}}{3p_{20}^{2}} \Bigr) x_2^{4}\, \Bigr)dt.
\end{aligned}
\end{equation}
system \eqref{eq:unfolding_system3} reduces to
\begin{equation}\label{eq:unfolding_system4}
\left\{\begin{array}{ll}
\displaystyle \dot{x_2} &= y_2, \\
\displaystyle \dot{y_2} &=q_{00}+q_{10}x_2+q_{01}y_2+q_{20}x_2^2+q_{11}x_2y_2+q_{30}x_2^3+q_{21}x_2^2y_2+q_{40}x_2^4  \\
&\ +q_{31}x_2^3y_2 + q_{50}x_2^5 + q_{41}x_2^4y_2 + R(x_2,y_2,\lambda),\\
\end{array}\right.
\end{equation}
The coefficients $q_{ij}$ are given in Appendix \ref{Coeff_4}. Since $q_{20}(0)=0.0623854\neq0$, we may apply the following transformation and time-rescaling
\begin{equation}
    \begin{aligned}
    x_2&=x_3, \quad y_2=y_3+\frac{q_{21}}{3q_{20}}y_3^2 ,\quad d\tau= \Bigl(
1+\frac{q_{21}}{3q_{20}}y_3 \Bigr) \, dt.
\end{aligned}
\end{equation}
to get the system of the form
\begin{equation}\label{eq:unfolding_system6}
\left\{\begin{array}{ll}
\displaystyle \dot{x_3} &= y_3, \\
\displaystyle \dot{y_3} &=r_{00}+r_{10}x_3+r_{01}y_3+r_{20}x_3^2+r_{11}x_3y_3+r_{31}x_3^3y_3+r_{41}x_3^4y_3 + R(x_3,y_3,\lambda),\\
\end{array}\right.
\end{equation}
where 
\begin{align*}
r_{00}&=q_{00}, \quad r_{10}=q_{10}, \quad r_{01}=\frac{q_{01}q_{20}-q_{00}q_{21}}{q_{20}}, \quad r_{20}=q_{20}, \\
r_{11}&=\frac{q_{11}q_{20}-q_{10}q_{21}}{q_{20}}, \quad
r_{31}=\frac{q_{31}q_{20}-q_{30}q_{21}}{q_{20}}, \quad r_{41}=\frac{q_{41}q_{20}-q_{40}q_{21}}{q_{20}}.
\end{align*}
Now, notice that $r_{20}(0)=0.0623854>0$ and $r_{41}(0)=-0.0000991429<0$, so we do the following change of coordinates and time rescaling
\begin{equation}
    \begin{aligned}
    x_3&=r_{20}^{\frac{1}{7}}r_{41}^{-\frac{2}{7}}x_4, \quad y_3=-r_{20}^{\frac{5}{7}}r_{41}^{-\frac{3}{7}}y_4,\quad d\tau =-r_{20}^{-\frac{4}{7}}r_{41}^{\frac{1}{7}} dt.
\end{aligned}
\end{equation}
then the system \eqref{eq:unfolding_system6} becomes
\begin{equation}\label{eq:unfolding_system7}
\left\{\begin{array}{ll}
\displaystyle \dot{x_4} &= y_4, \\
\displaystyle \dot{y_4} &=S_0+S_1x_4+S_2y_4+x_4^2+S_3x_4y_4+S_4x_4^3y_4-x_4^4y_4 + R(x_4,y_4,\lambda),\\
\end{array}\right.
\end{equation}
where 
\begin{align*}
S_0&=\frac{r_{00}r_{41}^\frac{4}{7}}{r_{20}^{\frac{9}{7}}}, \quad S_1=\frac{r_{10}r_{41}^\frac{2}{7}}{r_{20}^{\frac{8}{7}}}, \quad S_2=-\frac{r_{01}r_{41}^\frac{1}{7}}{r_{20}^{\frac{4}{7}}}, \quad S_3=-\frac{r_{11}r_{41}^{-\frac{1}{7}}}{r_{20}^{\frac{3}{7}}}, \quad S_4=-\frac{r_{31}r_{41}^{-\frac{5}{7}}}{r_{20}^{\frac{1}{7}}}
\end{align*}
Suppose
\begin{equation}
    \begin{aligned}
    x_4&=x-\frac{S_1}{2}, \quad y_4=y,
\end{aligned}
\end{equation}
then we obtain the required universal unfolding of the system 
\eqref{eq:perturbed_system}
\begin{equation}\label{eq:unfolding_system8}
\left\{\begin{array}{ll}
\displaystyle \dot{x} &= y, \\
\displaystyle \dot{y} &=\mu_1+\mu_2y+x^2+\mu_3xy+\mu_4x^3y-x^4y + R(x,y,\lambda),\\
\end{array}\right.
\end{equation}
where 
\begin{align*}
\mu_1&=S_0-\frac{1}{4}S_1^2, \quad \mu_2=S_2-\frac{1}{16}S_1^4-\frac{1}{2}S_1S_3-\frac{1}{8}S_1^3S_4, \quad \mu_3=S_3+\frac{1}{2}S_1^3+\frac{3}{4}S_1^2S_4, \quad \mu_4=S_4+2S_1
\end{align*}
Then we have
\begin{equation}  
\left.\frac{\partial(\mu_1,\mu_2,\mu_3,\mu_4)}
{\partial(\lambda_1,\lambda_2,\lambda_3,\lambda_4)}\right|_{\lambda=0}
=
\begin{vmatrix}
0 & 0 & 0.000234825 & 0.411607 \\
-1.11385 & 1.30786 & -0.00254921 & -0.603844 \\
-2.21998\times10^{-7} & -6.46566 & -0.0122084 & -0.932422 \\
10.8648 & -11.9502 & -0.0650896 & 19.2973
\end{vmatrix}=0.293648\neq0.
\end{equation}
Therefore, system \eqref{eq:perturbed_system} is a universal unfolding of the cusp singularity of order 4. Hence, System \eqref{eq:perturbed_system} has the same bifurcation set with respect to \(\lambda_1,\lambda_2,\lambda_3,\lambda_4\) as system \eqref{eq:unfolding_system8} has with respect to \(\mu_1,\mu_2,\mu_3,\mu_4\), up to a homeomorphism in the parameter space \cite{chow1994normal}. 
\end{proof}

Note that system \eqref{eq:unfolding_system8} admits no equilibria when
\(\mu_1>0\). At \(\mu_1=0\), the system undergoes a saddle-node
bifurcation. Higher-order cusp degeneracies arise when additional unfolding
parameters vanish: a codimension-two cusp occurs for
\[
\mu_1=\mu_2=0, \qquad \mu_3\neq 0,
\]
a codimension-three cusp occurs for
\[
\mu_1=\mu_2=\mu_3=0, \qquad \mu_4\neq 0,
\]
and a codimension-four cusp occurs when
\[
\mu_1=\mu_2=\mu_3=\mu_4=0.
\]
Consequently, all bifurcation sets associated with equilibria are contained
in the half-space \(\mu_1<0.\) For $\mu_1<0$, we make the following change of variables
\begin{equation}
    x=\varepsilon^2u, \quad y=\varepsilon^3v, \quad \tau=\varepsilon t
\end{equation}
where \(\varepsilon>0\), and we also introduce the following rescaling
\begin{equation}
    \mu_1 = -\varepsilon^4, \quad
\mu_2 = \varepsilon^8 a_1, \quad
\mu_3 = \varepsilon^6 a_2, \quad
\mu_4 = \varepsilon^2 a_3,
\end{equation}
we convert system \eqref{eq:unfolding_system8} into 
\begin{equation}
    \left\{
\begin{aligned}
\dot{u} &= v, \\
\dot{v} &= -1 + u^2 + \varepsilon^7 \left(a_1v  + a_2 uv + a_3 u^3v - u^4v\right),
\end{aligned}
\right.
\end{equation}
where we truncate the higher order term $R(x,y,\lambda)$ as it does not affect the local bifurcation structure. Renaming $(u,v)$ as $(x,y)$ and rescaling $\varepsilon^7 \to \varepsilon$, we obtain the simplified system
\begin{equation}\label{eq:Hopf_s1}
\left\{
\begin{aligned}
\dot{x} &= y, \\
\dot{y} &= -1 + x^2 + \varepsilon\left(a_1 y + a_2 xy + a_3 x^3 y - x^4 y\right).
\end{aligned}
\right.
\end{equation}

In the following, we describe the bifurcation diagram of system \eqref{eq:Hopf_s1} with the \(+\) sign in the \(x^4y\)-term in \(a_1a_2a_3\)-space, given by

\begin{equation}\label{eq:Hopf_plus}
\left\{
\begin{aligned}
\dot{x} &= y, \\
\dot{y} &= -1 + x^2 + \varepsilon\left(a_1 y + a_2 xy + a_3 x^3 y + x^4 y\right).
\end{aligned}
\right.
\end{equation}

One can then obtain the bifurcation diagram of system \eqref{eq:Hopf_s1} by letting \(y \to -y\) and \(t \to -t\) in system \eqref{eq:Hopf_plus}.

To describe the local bifurcation that appears in a neighborhood of system \eqref{eq:Hopf_plus} in \(a_1a_2a_3-\)space, we adopt the notations: (\(\text{H}\)) (resp. (\(\text{H}_2\)), (\(\text{H}_3\))) denotes Hopf bifurcation of codimension 1 (resp. 2, 3); (\(\text{HL}\)) (resp. (\(\text{HL}_2\)), (\(\text{HL}_3\))) denotes homoclinic bifurcation of codimension 1 (resp. 2, 3); \(\text{H}\cap \text{HL}\) denotes the simultaneous occurrence of (H) and (HL). Similarly, \(\text{H}_2\cap \text{HL}\) and \(\text{H}\cap \text{HL}_2\) denote the corresponding simultaneous occurrences; \(C_2\) (\(C_3\) ) denotes bifurcation of a double (triple) limit cycle. Then we have the following theorem.

\begin{theorem}\cite{chow1994normal, li1989system, xiang2022degenerate}
For sufficiently small \(\varepsilon\), the bifurcation diagram of system \eqref{eq:Hopf_plus} is shown in Fig \ref{BT4_figure}. It consists of the following:
\begin{enumerate}
    \item Surfaces (codimension-\(1\) bifurcation): \(\text{H}\), \(\text{HL}\), and \(\text{C}_2\), where \(\text{C}_2\) has two smooth pieces divided by the curve \(\text{C}_3\).

    \item Curves (codimension-\(2\) bifurcation): \(\text{H}_2\), \(\text{H}\cap \text{HL}\), \(\text{HL}_2\), \(\text{H}\cap \text{C}_2\), \(\text{HL}\cap \text{C}_2\), and \(\text{C}_3\).

    \item Points (codimension-\(3\) bifurcation): \(\text{H}_3\), \(\text{HL}_3\), \(\text{H}_2 \cap \text{HL}\), and \(\text{H}\cap\text{HL}_2\).
\end{enumerate}

When \(a=(a_1,a_2,a_3)\in\Omega\), where \(\Omega\) is surrounded by the surfaces \(\text{H}\), \(\text{HL}\), and \(\text{C}_2\), and the curve \(\text{C}_3\), system \eqref{eq:Hopf_plus} has exactly three limit cycles. The region \(\Omega\)  is a “topological three simplex”. When \(a\) varies from \(\Omega\) through \(\text{C}_2\), then two of the three limit cycles merge as a semistable limit cycle and then disappear. When \(a\) varies from \(\Omega\) through the surface \(\text{H}\), the most inner limit cycle shrinks into the focus \((x,y)=(-1,0)\), which changes its stability. When \(a\) varies from \(\Omega\) through the surface \(\text{HL}\), the outermost limit cycle expands and forms a homoclinic orbit, and then the connection from the saddle point to itself breaks down, and the homoclinic loop disappears.
\label{thm:local_bif}
\end{theorem}

\begin{figure}[H]
    \centering
    \includegraphics[width=\linewidth]{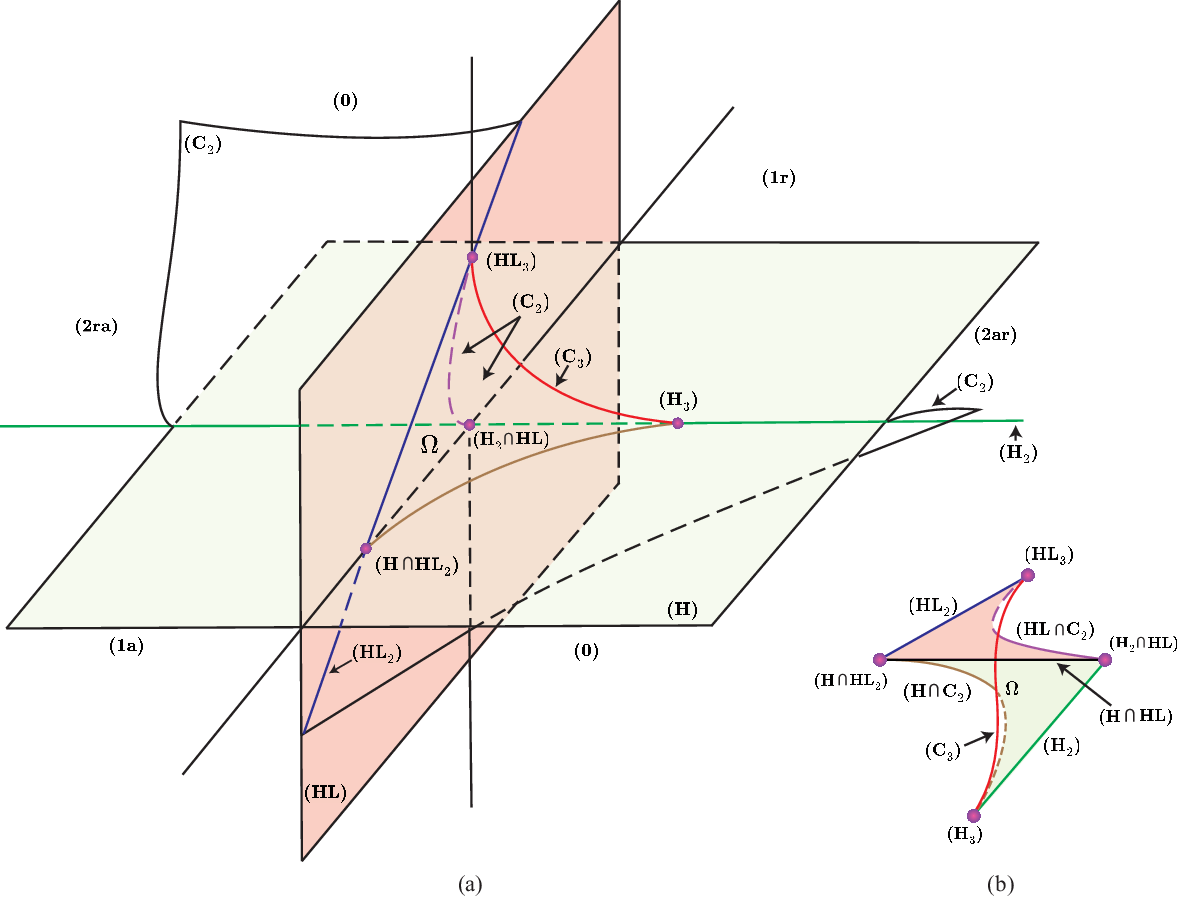}
    \caption{(a) The bifurcation diagram of system \eqref{eq:Hopf_plus} for sufficiently small \(a=(a_1,a_2,a_3)\). The codimension-one Hopf bifurcation surface is shown in light peach, while the light green surface represents the codimension-one homoclinic bifurcation surface. The number of limit cycles in each open region of the bifurcation diagram is indicated, with the indices a and r denoting attracting and repelling limit cycles, respectively, ordered from the innermost to the outermost cycle. The surface \(C_2\) represents the double-limit-cycle bifurcation, whereas the point \(C_3\) corresponds to the triple-limit-cycle bifurcation. (b) The topological three-simplex region \(\Omega\), where the system admits three limit cycles. The blue, green, purple, brown, and red curves are \(\text{HL}_2\), \(\text{H}_2\), \(\text{HL}\cap C_2\),\(\text{H}\cap C_2\) and \(C_3\), respectively. }
    \label{BT4_figure}
\end{figure}

We describe the bifurcation diagram of system \eqref{eq:unfolding_system8} with the \(+\) sign in the \(x^4y\)-term 
by taking its intersections with a 3-sphere centered at the origin in the
\(\mu\)-parameter space. The system has equilibria only on the closed half
3-sphere
\[
\left\{
(\mu_1,\mu_2,\mu_3):
\mu_1^2+\mu_2^2+\mu_3^2=1,\; \mu_1\leq 0
\right\},
\]
which is topologically equivalent to a closed 3-ball. The bifurcation
diagram inside this ball is similar to the bifurcation diagram of
\eqref{eq:Hopf_plus} with the \(+\) sign in the \(x^4y\)-term in the \(a\)-parameter space given in Fig. \ref{BT4_figure} and contains a topological
3-simplex \(\Omega\), where the system admits exactly three limit cycles.

The boundary of the ball, which is a 2-sphere, corresponds to the
saddle-node bifurcation set
\[
\mu_1=0,\qquad \mu_2\neq 0.
\]
On this boundary, the Bogdanov--Takens bifurcation occurs along the circle
\[
\mu_1=\mu_2=0,\qquad \mu_3\neq 0.
\]
Two points on this circle,
\[
\mu_1=\mu_2=\mu_3=0,\qquad \mu_4\neq 0,
\]
one with \(\mu_4>0\) and the other with \(\mu_4<0\), correspond to cusp
bifurcations of order \(3\). These two cusp points separate the two
Bogdanov--Takens cases, namely \(\mu_3<0\) and \(\mu_3>0\).

Inside the ball, the Hopf bifurcation surface HB and the homoclinic
bifurcation surface HLB branch from the Bogdanov--Takens circle.
Moreover, the codimension-two bifurcation curves inside the ball meet the
boundary of the ball at the two order-three cusp points, giving the
bifurcation set a conic structure. Finally, the bifurcation diagram of system \eqref{eq:unfolding_system8} can be obtained by letting \(y \to -y\) and \(t \to -t\). Further details are provided in
\cite{chow1994normal, li1989system, xiang2022degenerate}.

In the next two subsections, we analyze the degenerate Hopf and homoclinic
loop bifurcations arising in system \eqref{eq:unfolding_system8}. We describe
the corresponding higher-codimension bifurcation surfaces and curves in the
\(a_1a_2a_3\)-parameter space.

\subsubsection{Hopf bifurcation of codimension 3 around the BT point}\label{sec:hopf_codim3}
In this subsection, we investigate the degenerate Hopf bifurcation of the system \eqref{eq:unfolding_system8} by studying system \eqref{eq:Hopf_s1}.

\begin{corollary}\label{cor:hopf_3}
    For $\mu_1<0$ and sufficiently small $\varepsilon>0$, system \eqref{eq:Hopf_s1} undergoes a degenerate Hopf bifurcation of codimension-$3$ at
    \[
    (a_1,a_2,a_3)=\left(-\frac{11}{5},-\frac{2}{5}, -\frac{14}{5}\right)
    \]
    in the neighborhood of the codimension-$4$ Bogdanov-Takens bifurcation point.
    \end{corollary}
    \begin{proof}
For system \eqref{eq:Hopf_s1}, there exist two equilibria: a focus, $E_-=(-1,0)$ and a saddle, $E_+=(1,0)$.
\[
J(E_-)
=
\begin{bmatrix}
0 & 1\\
-2 & \varepsilon(a_1-a_2-a_3-1)
\end{bmatrix}.
\]
Hence
\[
\operatorname{tr}(J(E_-))
=
\varepsilon(a_1-a_2-a_3-1),
\qquad
\det(J(E_-))=2.
\]
Therefore $E_-=(-1,0)$ is a weak focus which makes the system undergo Hopf bifurcation when
\begin{equation}\label{eq:hopf_focus}
    a_1-a_2-a_3-1=0.
\end{equation}
This equation defined the Hopf bifurcation surface in \(a_1a_2a_3-\)space. Now to analyze the degenerate Hopf bifurcation, we shift $E_-$ to the origin by making the following transformation
\begin{equation}\label{eq:Hopf_t1}
    x=u-1,\qquad y=-\sqrt{2}\,v, \qquad \tau=\sqrt{2}t
\end{equation}
we obtain, after renaming $(u,v)$ as $(x,y)$ and $\tau$ as $t$,
\begin{equation}\label{eq:Hopf_s2}
    \left\{
\begin{aligned}
\dot{x} &= -y,\\
\dot{y} &= x-\frac{1}{2}x^2
+\frac{\varepsilon}{\sqrt{2}}\Big[
(a_2+3a_3+4)xy
+(-3a_3-6)x^2y  \\
&\hspace{3.2cm}
+(a_3+4)x^3y
-x^4y
\Big].
\end{aligned}
\right.
\end{equation}
This is the Liénard form and this system matches the normal form considered in Proposition 6.5 of \cite{lamontagne2008bifurcation}.
Hence, when the condition given in equation \eqref{eq:hopf_focus} holds, the Lyapunov constants can be computed explicitly as
\begin{equation}
\begin{aligned}
    L_1& = \frac{\sqrt{2}\varepsilon}{32}(a_2-3a_3-8),\\
    L_2&=\frac{\sqrt{2}\varepsilon}{768}\left((a_2-3a_3-8)(4 + a_2 + 3a_3)^2 \varepsilon+ \left(3855a_3 - 5a_2 + 10792 \right)\right),\\
    L_3&=-\frac{35\sqrt{2}\varepsilon}{9216}(a_3+4)
\end{aligned}
\end{equation}
The first Lyapunov coefficient \(L_1\) vanishes precisely on the curve \(a_2-3a_3-8=0\), and hence along this curve, the Hopf bifurcation becomes degenerate. By assuming $L_1=0$, we obtain
\begin{equation}
\begin{aligned}
    L_2&=\frac{\sqrt{2}\varepsilon}{192}\left(5a_3+14\right),\\
\end{aligned}
\end{equation}
Hence, if, \(L_1=0\) and \(L_2\neq 0,\) then this degeneracy is of codimension two, and the system undergoes a codimension-two Hopf bifurcation. If $L_1=L_2=0$, then
\begin{equation}
\begin{aligned}
    L_3&=-\frac{7\sqrt{2}\varepsilon}{1536}\neq0.
\end{aligned}
\end{equation}
This concludes that the Hopf bifurcation of codimension 3 occurs at 
\begin{equation}
    a_1=-\frac{11}{5}, \qquad a_2=-\frac{2}{5}, \qquad a_3=-\frac{14}{5}.
\end{equation}
\end{proof}

\subsubsection{Homoclinic bifurcation of codimension 3 around the BT point}\label{sec:homoclinic_codim3}
In the previous subsection, we established the existence of the degenerate Hopf bifurcation of codimension-$3$ near the focus equilibrium $E_-=(-1,0)$ for the reduced system \eqref{eq:Hopf_s1}. We now study the Homoclinic bifurcation associated with the saddle equilibrium $E_+=(1,0)$. 

\begin{corollary}\label{cor:homo_3}
    For $\mu_1<0$ and sufficiently small $\varepsilon>0$, a homoclinic bifurcation of codimension-$3$ exist for the system \eqref{eq:Hopf_s1} associated with it's saddle equilibrium $E_+=(1,0)$ at 
    \[
    \left(a_1,a_2,a_3\right)=\left(\frac{107}{91},\frac{94}{91},-\frac{110}{91}\right)
    \]
    in the neighborhood of the codimension-$4$ Bogdanov-Takens bifurcation point.
\end{corollary}

\begin{proof}
For $\varepsilon=0$, system \eqref{eq:Hopf_s1} reduces to the Hamiltonian system
\begin{equation}\label{eq:Hamiltonian_system}
\left\{
\begin{aligned}
\dot{x} &= y, \\
\dot{y} &= -1 + x^2.
\end{aligned}
\right.
\end{equation}
with the Hamiltonian function being
\begin{equation}
    H(x,y)=\frac{y^2}{2}+x-\frac{x^3}{3}.
\end{equation}
Indeed, 
\begin{equation}
    \frac{\partial H}{\partial y}=y, \qquad -\frac{\partial H}{\partial x}=-1+x^2.
\end{equation}
The saddle equilibrium $E_+=(1,0)$ lies on the energy level $H(1,0)=\frac{2}{3}$.\\
Therefore, the level curve $H(x,y)=h=2/3$ gives a homoclinic orbit of the unperturbed system.
Solving
\begin{equation}
    \frac{y^2}{2}+x-\frac{x^3}{3}=\frac{2}{3}
\end{equation}
for $y$ gives
\begin{equation}
    y=\pm \sqrt{\frac{2}{3}}(1-x)\sqrt{x+2},\qquad -2\leq x\leq 1.
\end{equation}
To determine whether this homoclinic loop persists under the perturbation, we compute the Melnikov function as in \cite{li1989system} (also in \cite{shang2024bogdanov}). For system \eqref{eq:Hopf_s1}, the perturbation term is
\begin{equation}
    a(x,y)=a_1y+a_2xy+a_3x^3y-x^4y.
\end{equation}
Hence the Melnikov function along the Hamiltonian level curve $\Gamma_h$ is 
\begin{equation}
    M(h)=\oint_{\Gamma_h}
\left(a_1y+a_2xy+a_3x^3y-x^4y\right)\,dx.
\end{equation}
In particular, along the homoclinic orbit $H=2/3$, we obtain
\begin{equation}
    M\left(\frac23\right)=\oint_{\Gamma_{2/3}}
\left(a_1y+a_2xy+a_3x^3y-x^4y\right)\,dx.
\end{equation}
Direct calculations show that 
\begin{equation}
    M\left(\frac23\right)=\frac{24\sqrt{2}}{5}\left(a_1-\frac57a_2-\frac{103}{77}a_3-\frac{187}{91}\right).
\end{equation}
Therefore, the codimension-one homoclinic bifurcation surface is given by
\begin{equation}\label{eq:Hom_surface1}
a_1-\frac57a_2-\frac{103}{77}a_3-\frac{187}{91}=0.
\end{equation}
To obtain higher codimension homoclinic bifurcation, we impose further degeneracy conditions on the Melnikov function.
Differentiating the Melnikov function at
\(h=2/3\) gives
\begin{equation}
    M'\left(\frac23\right)
=
2\sqrt{\frac32}
\int_{-2}^{1}
\frac{P(x)}
{(1-x)\sqrt{x+2}}\,dx.
\end{equation}
where $P(x)=a_1+a_2x+a_3x^3-x^4$, and requiring the saddle trace condition, given by
\begin{equation}\label{eq:tr_cond}
    P(1)=a_1+a_2+a_3-1=0,
\end{equation}
under which $M'(2/3)$ is finite and direct calculations give
\begin{equation}\label{eq:hom_2}
    M'\left(\frac{2}{3}\right)=6\sqrt{2}\left(-a_2-\frac{9}{5}a_3 -\frac{8}{7}\right),
\end{equation}
So, from the three conditions, $M(2/3)=M'(2/3)=P(1)=0$, we get
\begin{equation}
    \left(a_1,a_2,a_3\right)=\left(\frac{107}{91},\frac{94}{91},-\frac{110}{91}\right).
\end{equation}
To determine the maximal codimension of homoclinic bifurcation, we further investigate whether $M''(2/3)$ becomes unbounded. According to the theory of generalized homoclinic bifurcation, this is equivalent to requiring that the first saddle quantity be nonzero.\\
For system \eqref{eq:Hopf_s1}, the first saddle quantity is obtained from the formula given in \cite{joyal1989saddle} as
\begin{equation}
    V_1=\frac{1}{2}\left(-a_2+3a_3-8\right).
\end{equation}
Substituting $\left(a_1,a_2,a_3\right)=\left(\frac{107}{91},\frac{94}{91},-\frac{110}{91}\right)$ into $V_1$, we obtain
\begin{equation}
    V_1=-\frac{576}{91}\neq0.
\end{equation}
Therefore, $M''(2/3)$ becomes unbounded, implying that no higher-order degeneracy occurs. Consequently, by Proposition 2.3.1 in \cite{li1989system} together with the approach developed in \cite{joyal1989saddle}, we conclude that the highest codimension of homoclinic bifurcation is $3$.
\end{proof}

\section{A triple positive equilibrium for \texorpdfstring{$p=2$}{p = 2}}\label{sec:triple_eq}
\subsection{Quadratic predator self-limitation}
In this section, we investigate the model for the special case when $p=2$, corresponding to intraspecific predator competition. From the BT-2 normal form in \eqref{eq:s4}, the coefficient $\gamma_1$ plays a crucial role in determining the local dynamics near the nilpotent equilibrium. In particular, the condition $\gamma_1\neq0$ corresponds to the generic Bogdanov-Takens bifurcation, whereas $\gamma_1=0$ leads to further degeneracy. We therefore focus on the case $\gamma_1=0$, which signals the onset of higher order degeneracy and may give rise to codimension-3 singularities. In the present model, this condition corresponds to the coalescence of three interior equilibria and they can be characterized explicitly for $p=2$.\\

For simplification, we set $\sigma=1+\alpha\xi$ throughout this section, then for $p=2$, the scalar equilibrium equation \eqref{eq:k_function} reduces to
\[
k(x)=\eta\left(\frac{x+\xi}{x+\sigma}\right)-\delta-c\xi \left(1-\frac{x}{K}\right)(x+\sigma),
\]
A triple root $x_t$ of $k(x)=0$ is characterized by
\[
k(x_t)=0,\qquad k'(x_t)=0,\qquad k''(x_t)=0.
\]
which gives the unique positive interior equilibrium for the system \eqref{eq:model_BT} as
\begin{equation}
    E_t(x_t,y_t)=\left(\frac{K-2\sigma}{3}, \frac{2(K+\sigma)^2}{9K}\right).
\end{equation}
Furthermore, the parameters $c$ and $\eta$ satisfy
\begin{align}\label{eq:p2_triple_root_parameters}
    c=\frac{27 K \delta \left(\sigma-\xi\right)}{\xi \left(K + \sigma\right)^2 \left(K-8\sigma+9\xi\right)},
    \quad \eta=\frac{\delta \left( K + \sigma\right)}{K-8\sigma+9\xi}.
\end{align}
Also, from \eqref{eq:p2_triple_root_parameters} and $\operatorname{tr}(J(E_t))=0$, we have
\begin{equation}\label{eq:delta_p2}
    \delta=\frac{\left(K-8\sigma+9\xi\right)\left(K-2\sigma\right)}{18K\left(\sigma-\xi\right)}:=\delta_t.
\end{equation}
From $\operatorname{tr}(J(E_t))=0$ and $\det(J(E_t))=0$, we have $\left(\sigma-\xi\right)>0$.
\begin{theorem}\label{thm:nilpotent_focus}
    Suppose $K > \max\left\{\,8\sigma - 9\xi,\ 4\sigma\right\}$ and that the conditions in \eqref{eq:p2_triple_root_parameters} hold, then system \eqref{eq:model_BT} has a unique positive equilibrium $E_t\left(\frac{K-2\sigma}{3}, \frac{2(K+\sigma)^2}{9K}\right)$. Furthermore, if $\delta=\delta_t$, then $E_t$ is a degenerate nilpotent focus of codimension-$3$.
\end{theorem}
\begin{proof}
    When $\delta=\delta_t$, we apply the following transformations successively 
\begin{equation}\label{eq:nf_transf1}
    \begin{aligned}
    x&=\tilde{x}+\left(\frac{K-2\sigma}{3}\right), \qquad y=\tilde{y}+\left(\frac{2(K+\sigma)^2}{9K}\right);\\
    \tilde{x}&=X, \qquad \tilde{y}=\left(\frac{K+\sigma}{3K}\right)X-\left(\frac{K+\sigma}{K-2\sigma}\right)Y;\\
    X &=x+\left(\frac{3(K+4\sigma)}{4(K-2\sigma)(K+\sigma)}\right)x^2,\quad Y=y+\left(\frac{K-2\sigma}{K(K+\sigma)}\right)x^2 + \left(\frac{3}{2(K+\sigma)}\right)xy.
\end{aligned}
\end{equation}

that transforms system \eqref{eq:model_BT} to

\begin{eqnarray}\label{eq.rs1_p2}
\left\{\begin{aligned}
\dot{x}= &\ y + A_{30}x^3+A_{21}x^2y+A_{40}x^4+A_{31}x^3y + O(|(x,y)^5|),\\
\dot{y}=&\ B_{11}xy + B_{30}x^3 + B_{21}x^2y + B_{12}xy^2 + B_{40} x^4 + B_{31}x^3y + B_{22}x^2y^2 + O(|(x,y)^5|),
\end{aligned}\right.
\end{eqnarray}
where
\begin{align*}
    A_{30}&=-\frac{3(K+4\sigma)}{2K(K+\sigma)^2}<0,\quad A_{21}=-\frac{27(K-8\sigma)\sigma}{4(K-2\sigma)^2(K+\sigma)^2},\\
    A_{40}&=\frac{27(K^2-16\sigma^2)}{16K(K-2\sigma)(K+\sigma)^3}, \quad
    A_{31}=\frac{81 \sigma\left(K^2 - 13K\sigma + 4\sigma^2\right)}{4(K-2\sigma)^3(K + \sigma)^3},\\ 
    B_{11} &=  -\frac{5K-16\sigma}{2K(K+\sigma)}<0, \quad B_{30} = -\frac{(K-2\sigma)^2}{K^2(K+\sigma)^2}<0, \quad B_{12}=\frac{9}{2(K+\sigma)^2}>0,\\
    B_{21}&=\frac{3(11K^2- 68K\sigma  + 128\sigma^2)}{8K(K-2\sigma)(K + \sigma)^2}>0,  \quad
    B_{31}=\frac{27(K-8\sigma)\sigma}{K(K-2\sigma)(K+\sigma)^3},\\
    B_{40}&=\frac{9(K-4\sigma)}{2K^2(K+\sigma)^2}, \quad B_{22}=-\frac{27(K^2-7K\sigma+28\sigma^2)}{8(K-2\sigma)^2(K+\sigma)^3}<0.
\end{align*}
Clearly, if $K>4\sigma$, then
\begin{equation}
    B_{11}B_{30}=\frac{(K-2\sigma)^2(5K-16\sigma)}{2K^3(K+\sigma)^3}>0.
\end{equation}
Now we apply the following near identity transformation to the system \eqref{eq.rs1_p2}
\begin{equation}\label{eq.transformation4}
    \begin{aligned}
    x& = x+\frac{1}{6}\left(2 A_{21} + B_{12}\right) x^3 + \frac{1}{12}\left(3 A_{31} + B_{22}\right) x^4, 
    \\
    y &= y-A_{30}x^3 + \frac{B_{12}}{2} x^2 y -A_{40}x^4+\frac{B_{22}}{3} x^3 y,
\end{aligned}
\end{equation}
that reduces system \eqref{eq.rs1_p2} to
\begin{eqnarray}\label{eq:rs2_p2}
\left\{\begin{aligned}
\dot{x}= &\ y,\\
\dot{y}=&\ B_{11}xy + B_{30}x^3 + (B_{21}+3A_{30})x^2y + (B_{40}-A_{30}B_{11}) x^4 \\
&\ + \left(B_{31}+4A_{40}+\frac{A_{21}B_{11}}{3}+\frac{B_{11}B_{12}}{6}\right)x^3y+ O(|(x,y)^5|),
\end{aligned}\right.
\end{eqnarray}
For $K>4\sigma$, we have 
\begin{equation}
    5\,B_{30}\,(B_{21} + 3A_{30}) - 3\,B_{11}\,(B_{40} - B_{11}A_{30})=\frac{3(K-2\sigma)\left(5K^2 + 28K\sigma - 40\sigma^2\right)}
{2K^3(K + \sigma)^4}>0,
\end{equation}
and
\begin{equation}
    B_{11}^2+8B_{30}=-\frac{7K^2 + 32K\sigma - 128\sigma^2}{4K^2(K + \sigma)^2}<0.
\end{equation}
By the theory given by Dumortier et al. \cite{dumortier2006bifurcations} and Lemma 3.1 in \cite{cai2013multiparametric}, the equilibrium $E_t$ is a degenerate focus of codimension-$3$ under the stated conditions.
\end{proof}
\begin{remark}
    Under the rescaling, 
    \begin{equation}\label{eq:rescaling_zhu}x=-\frac{X}{\sqrt{-B_{30}}},\qquad y=\frac{Y}{\sqrt{-B_{30}}},\qquad t=-\tau ,
\end{equation}
which is well-defined since $B_{30}<0$, system \eqref{eq:rs2_p2} is transformed to the normal form given by Zhu and Rousseau in Section $2$ of \cite{zhu2002finite},
\begin{equation}
    \dot X=Y,\qquad 
\dot Y=e_1X^3+bXY+e_2X^2Y+dX^4+fX^3Y+O(|(X,Y)|^5),
\end{equation}
where they state that the case $b=0$ corresponds to a higher-codimension ($\geq4$) degeneracy. For our system, after the rescaling \eqref{eq:rescaling_zhu}, the corresponding coefficient is $b=-\frac{B_{11}}{\sqrt{-B_{30}}}$.
Therefore, $b=0$ would be equivalent to $B_{11}=0$. That is, it requires $K=\frac{16}{5}\sigma$. However, this is excluded by our hypothesis $K > \max\left\{\,8\sigma - 9\xi,\ 4\sigma\right\}$. Hence the
degenerate case $b=0$ does not occur in the
parameter regime considered here. Instead, our system belongs to the generic focus-type case $0<b<2\sqrt2$.
\end{remark}

\subsection{Degenerate focus type Bogdanov-Takens bifurcation of codimension 3}\label{sec:nilpotent_focus_unfolding}
As we proved in Theorem \eqref{thm:nilpotent_focus}, the system \eqref{eq:model_BT} admits a unique positive equilibrium $E_t\left(\frac{K-2\sigma}{3}, \frac{2(K+\sigma)^2}{9K}\right)$ which is a codimension-3 degenerate nilpotent focus when $(c,\eta,\delta)=(c_t,\eta_t,\delta_t)$ where
\begin{equation}\label{eq:p2_bif_parameters}
\begin{aligned}
    c_t=&\frac{3 (K-2\sigma)}{2\xi \left(K + \sigma\right)^2},
    \quad \eta_t=\frac{(K-2\sigma) \left(K + \sigma\right)}{18K(\sigma-\xi)},\\ &\delta_t=\frac{\left(K-8\sigma+9\xi\right)\left(K-2\sigma\right)}{18K\left(\sigma-\xi\right)}.
\end{aligned}
\end{equation}
satisfying $K > \max\left\{\,8\sigma - 9\xi,\ 4\sigma\right\}$.

Considering $(c,\eta,\delta)$ as the bifurcation parameters and the following unfolding system of system \eqref{eq:model_BT} for $p=2$, 
\begin{equation}\label{eq:perturbed_system_nf}
\left\{\begin{array}{l}
\displaystyle \dot{x}=x \left (1-\frac{x}{K}\right )-\frac{x y}{1+x+\alpha\xi} ,\\
\displaystyle \dot{y}=y\,\left[ (\eta_t + \lambda_1)\left ( \frac{x+\xi}{1+x+\alpha\xi} \right)\ -(\delta_t + \lambda_2)\, - (c_t+\lambda_3)\, \xi\, y\right].
\end{array}\right.
\end{equation}
Our goal is to determine whether system \eqref{eq:model_BT} admits a complete unfolding of a codimension-3 focus-type degenerate Bogdanov-Takens bifurcation under small perturbations $(\lambda_1,\lambda_2,\lambda_3)$ of the parameters $(c,\eta,\delta)$ around $(c_t,\eta_t,\delta_t)$. To achieve that, we give the following theorem.

\begin{theorem}\label{thm:nf_unfolding}
    For $\lambda=(\lambda_1,\lambda_2,\lambda_3)\sim(0,0,0)$, system \eqref{eq:perturbed_system_nf} is $C^\infty$ topologically equivalent to 
    \begin{eqnarray}\label{eq:nf_unfolding_sys}
    \left\{\begin{aligned}
    \dot{x}= &\ y,\\
    \dot{y}=&\  \nu_1(\lambda)+\nu_2(\lambda)x-x^3+y[\nu_3(\lambda)+I_1(\lambda)x+x^2]+y^2R(x,y,\lambda)+O(|(x,y)^4|), \\
    \end{aligned}\right.
    \end{eqnarray}
    where $\nu_i(\lambda)$, $(i=1,2,3)$, and $I_1(\lambda)$ are functions in $\lambda$, and $R(x,y,\lambda)$ is a $C^\infty$ function in $(x,y,\lambda)$. Moreover,
    \begin{equation} 
\left.\frac{\partial(\nu_1,\nu_2,\nu_3)}
{\partial(\lambda_1,\lambda_2,\lambda_3)}\right|_{\lambda=0}\neq0.
\end{equation}
\end{theorem}

\begin{proof}
To transform system \eqref{eq:perturbed_system_nf} to its versal unfolding, we apply a series of near-identity transformations. 

Using the same set of transformations successively as in \eqref{eq:nf_transf1}, the Taylor expansion of system \eqref{eq:perturbed_system_nf} yields
\begin{eqnarray}\label{eq:nf_sys1}
\left\{\begin{aligned}
\dot{x}= &\ y + \bar{a}_{30}x^3+ \bar{a}_{21}x^2y+O(|(x,y)^4|),\\
\dot{y}=&\  \bar{b}_{00}+\bar{b}_{10}x+\bar{b}_{01}y+\bar{b}_{20}x^2+\bar{b}_{11}xy+\bar{b}_{02}y^2+\bar{b}_{30}x^3+\bar{b}_{21}x^2y+\bar{b}_{12}xy^2+O(|(x,y)^4|), \\
\end{aligned}\right.
\end{eqnarray}
where
\begin{align*}
\bar{a}_{30}(0)&=A_{30}, \quad \bar{a}_{21}(0)=A_{21}, \quad \bar{b}_{00}(0)=\bar{b}_{10}(0)=\bar{b}_{01}(0)=\bar{b}_{20}(0)=\bar{b}_{02}(0)=0,\\
&\bar{b}_{30}(0)=B_{30},\quad  \bar{b}_{21}(0)=B_{21}, \quad \bar{b}_{12}(0)=B_{12}.
\end{align*}
The $C^\infty$ change of coordinates $x_1=x$, $y_1=\dot{x}$ gives
\begin{eqnarray}\label{eq:nf_sys2}
\left\{\begin{aligned}
\dot{x_1}= &\ y_1,\\
\dot{y_1}=&\  \bar{c}_{00}+\bar{c}_{10}x_1+\bar{c}_{01}y_1+\bar{c}_{20}x_1^2+\bar{c}_{11}x_1y_1+\bar{c}_{02}y_1^2+\bar{c}_{30}x_1^3+\bar{c}_{21}x_1^2y_1+\bar{c}_{12}x_1y_1^2+O(|(x_1,y_1)^4|), \\
\end{aligned}\right.
\end{eqnarray}
where
\begin{align*}
\bar{c}_{00}(0)&=\bar{c}_{10}(0)=\bar{c}_{01}(0)=\bar{c}_{20}(0)=\bar{c}_{02}(0)=0, \quad \bar{c}_{11}(0)=B_{11}, \quad \bar{c}_{30}(0)=B_{30},\\
\bar{c}_{21}(0)&=\frac{3 (K-20\sigma)}{2 \left(K + \sigma\right)}, \quad \bar{c}_{12}(0)=-\frac{27\sigma(5K-4\sigma)}{2(K-2\sigma)^2(K+\sigma)^2}.
\end{align*}
Further, using the following near-identity transformation
\begin{equation}
    \begin{aligned}
    x_1&=x_1+\frac{3\left(K^2 - 7K\sigma + 28\sigma^2\right)}{4\left(K - 2\sigma\right)^2 \left(K + \sigma\right)^2}x_1^3, \quad y_1=y_1+\frac{9\left(K^2 - 7K\sigma + 28\sigma^2\right)}{4\left(K - 2\sigma\right)^2 \left(K + \sigma\right)^2}x_1^2y_1,\\
\end{aligned}
\end{equation}
gives
\begin{eqnarray}\label{eq:nf_sys3}
\left\{\begin{aligned}
\dot{x_1}= &\ y_1,\\
\dot{y_1}=&\  \bar{d}_{00}+\bar{d}_{10}x_1+\bar{d}_{01}y_1+\bar{d}_{20}x_1^2+\bar{d}_{11}x_1y_1+\bar{d}_{02}y_1^2+\bar{d}_{30}x_1^3+\bar{d}_{21}x_1^2y_1+O(|(x_1,y_1)^4|), \\
\end{aligned}\right.
\end{eqnarray}
where
\begin{align*}
\bar{d}_{00}(0)&=\bar{d}_{10}(0)=\bar{d}_{01}(0)=\bar{d}_{20}(0)=\bar{d}_{02}(0)=0,\\
&\bar{d}_{30}(0)=B_{30}<0,\quad  \bar{d}_{21}(0)=-\frac{3(K^2+92K\sigma  -224\sigma^2)}{K(K-2\sigma)(K + \sigma)^2}<0.
\end{align*}
We make the following smooth coordinate transformation
\begin{equation}
    \begin{aligned}
    x_1&=-\frac{\bar{d}_{20}}{3\bar{d}_{30}} - \frac{\sqrt{-\bar{d}_{30}}}{\bar{d}_{21}}\, x, \quad y_1=-\frac{(-\bar{d}_{30})^{3/2}}{\bar{d}_{21}^2} y, \quad t=-\frac{\bar{d}_{21}}{3\bar{d}_{30}}\tau,\\
\end{aligned}
\end{equation}
then system \eqref{eq:nf_sys3} becomes (denoting $\tau$ by $t$)
\begin{eqnarray}\label{eq:nf_sys4}
\left\{\begin{aligned}
\dot{x}= &\ y,\\
\dot{y}=&\  \nu_1(\lambda)+\nu_2(\lambda)x-x^3+y[\nu_3(\lambda)+I_1(\lambda)x+x^2]+y^2R(x,y,\lambda)+O(|(x,y)^4|), \\
\end{aligned}\right.
\end{eqnarray}
where $\nu_i(\lambda)$, $(i=1,2,3)$, and $I_1(\lambda)$ can be expressed in terms of $\lambda_1$, $\lambda_2$, $\lambda_3$. Moreover, $R(x,y,\lambda)$ is a $C^\infty$ function consisting of terms up to first order in $(x,y)$, whose coefficients vary smoothly with $\lambda_1$, $\lambda_2$, $\lambda_3$; their explicit forms are omitted for conciseness.\\
Then for $K>4\sigma$, we have
\begin{equation} 
\begin{aligned}
\left.\frac{\partial(\nu_1,\nu_2,\nu_3)}
{\partial(\lambda_1,\lambda_2,\lambda_3)}\right|_{\lambda=0}
&=
-\frac{27 K^2 \xi  (\sigma- \xi) (K-4\sigma) (5 K-16 \sigma) \left(K^2+92 K \sigma-224 \sigma^2\right)^6}{262144 (K-2\sigma)^{17} \sqrt{\frac{(K-2\sigma^2)^2}{K^2 (K+\sigma )^2}}}<0.
\end{aligned}
\end{equation}
The mapping $(\lambda_1,\lambda_2,\lambda_3)\mapsto (\nu_1,\nu_2,\nu_3)$ is locally homeomorphic near the origin, ensuring that $\nu_1$, $\nu_2$, and $\nu_3$ are independent parameters. In addition, system \eqref{eq:nf_sys4} has coefficients $-1$ and $1$ for the $x^3$ and $x^2y$ terms, respectively, and the coefficient of $xy$ is $I_1(\lambda)$, where
\begin{equation}
    0<I_1(\lambda)=\frac{(5 K-16 \sigma) }{2K (K+\sigma) \sqrt{\frac{(K-2\sigma^2)^2}{K^2 (K+\sigma )^2}}}<2\sqrt{2}
\end{equation}
for $K>4\sigma$. Hence, by the criteria in \cite{dumortier2006bifurcations}, system \eqref{eq:nf_sys4} represents a generic three-parameter unfolding of a codimension-$3$ Bogdanov-Takens singularity of focus type, implying that system \eqref{eq:model_BT} undergoes a codimension-$3$ degenerate Bogdanov-Takens bifurcation.
\end{proof}

\section{Numerical Simulation}
\label{sec:num_sim}
To investigate the local unfolding of the Bogdanov-Takens singularity, we perform two-parameter numerical continuation using \textit{Matcont} with $(\delta,c)$ as the unfolding parameters while fixing $(\eta, K, \alpha, \xi, p) = (3.5700505,\, 46.409597,\,0.26,\,0.675,\, 1.4)$. The resulting bifurcation diagram of cusp-type BT bifurcation for system \eqref{eq:model_BT} is shown in Fig. \ref{fig:BT_diag}. The curves contain both a Bogdanov-Takens (BT) point and a Generalized Hopf (GH) point, which organize the local bifurcation structure. The solid and dashed portions denote supercritical and subcritical Hopf bifurcations, respectively. In details, Fig. \ref{fig:BT_diag}(a) is a magnified view of the neighborhood of the BT point. Owing to the extremely small parameter interval separating the BT and GH points, the change in Hopf criticality occurs in a very narrow region of the parameter space. Consequently, the supercritical branch occupies only a short segment of the Hopf curve before transitioning to the subcritical branch at the GH point. And Fig. \ref{fig:BT_diag}(b) shows further continuation of the Hopf bifurcation curve in the $(\delta,c)$ parameter plane. Starting from the BT point, the Hopf bifurcation is supercritical, as indicated by the negative first Lyapunov coefficient. The curve terminates at the Generalized Hopf point, where the first Lyapunov coefficient vanishes and changes sign. Beyond this point, the Hopf bifurcation becomes subcritical, represented by the red dashed branch. 

The system \eqref{eq:model_BT} undergoes a series of bifurcations and exhibit complex dynamics as parameters vary in different regions of $(\delta,c)$ plane. Some of these dynamics have been illustrated in the phase portraits given in Fig. \ref{fig:BT_phase_portraits}. Fig. \ref{fig:BT_phase_portraits}(a)-(b) shows the existence of three interior equilibria, $E_1$, $E_2$ and $E_3$, in particular, Fig. \ref{fig:BT_phase_portraits}(a) is a zoomed version showing 2 interior equilibria, $E_1$ and $E_2$ out of 3 and has $c=1.3311815$ and $\delta$ decreases to $1.6070605$ giving rise to a saddle-node bifurcation where the equilibria $E_1$ and $E_2$ coalesce into a nonhyperbolic equilibrium. This corresponds to the saddle-node curve emanating from the BT point, Fig. \ref{fig:BT_phase_portraits}(b) shows all three interior equilibria. In Fig. \ref{fig:BT_phase_portraits}(c) for $\delta=1.607065$ and $c=1.3311815$, the system lies beyond the saddle-node bifurcation curve. The equilibria $E_1$ and $E_2$ have annihilated each other through a saddle-node bifurcation, leaving a unique biologically relevant interior equilibrium $E_3$. This equilibrium is locally asymptotically stable, and all trajectories in the biologically relevant region converge to $E_3$, indicating the absence of sustained oscillations. Fig.~\ref{fig:BT_phase_portraits}(d) illustrates the phase portrait on the homoclinic bifurcation curve. The stable and unstable manifolds of the saddle equilibrium coincide to form a homoclinic loop. No periodic orbit exists at this parameter value, and the equilibrium enclosed by the loop remains a stable focus. Consequently, trajectories initiated inside the loop spiral toward the equilibrium, whereas the homoclinic orbit acts as the separatrix in the phase plane. Upon crossing the supercritical Hopf bifurcation curve at $\delta=1.60778889$ and $c=1.3298552$, the equilibrium loses stability and a stable limit cycle is created, as shown in Fig.~\ref{fig:BT_phase_portraits}(e). Trajectories initiated inside the periodic orbit spiral outward (green), while those outside spiral inward (red), with both approaching the same stable limit cycle. This behavior confirms the emergence of a stable oscillatory state through a supercritical Hopf bifurcation.

\begin{figure}
\centering
\begin{minipage}[b]{0.48\textwidth}
    \centering
    \includegraphics[width=\textwidth, height=5.2cm]{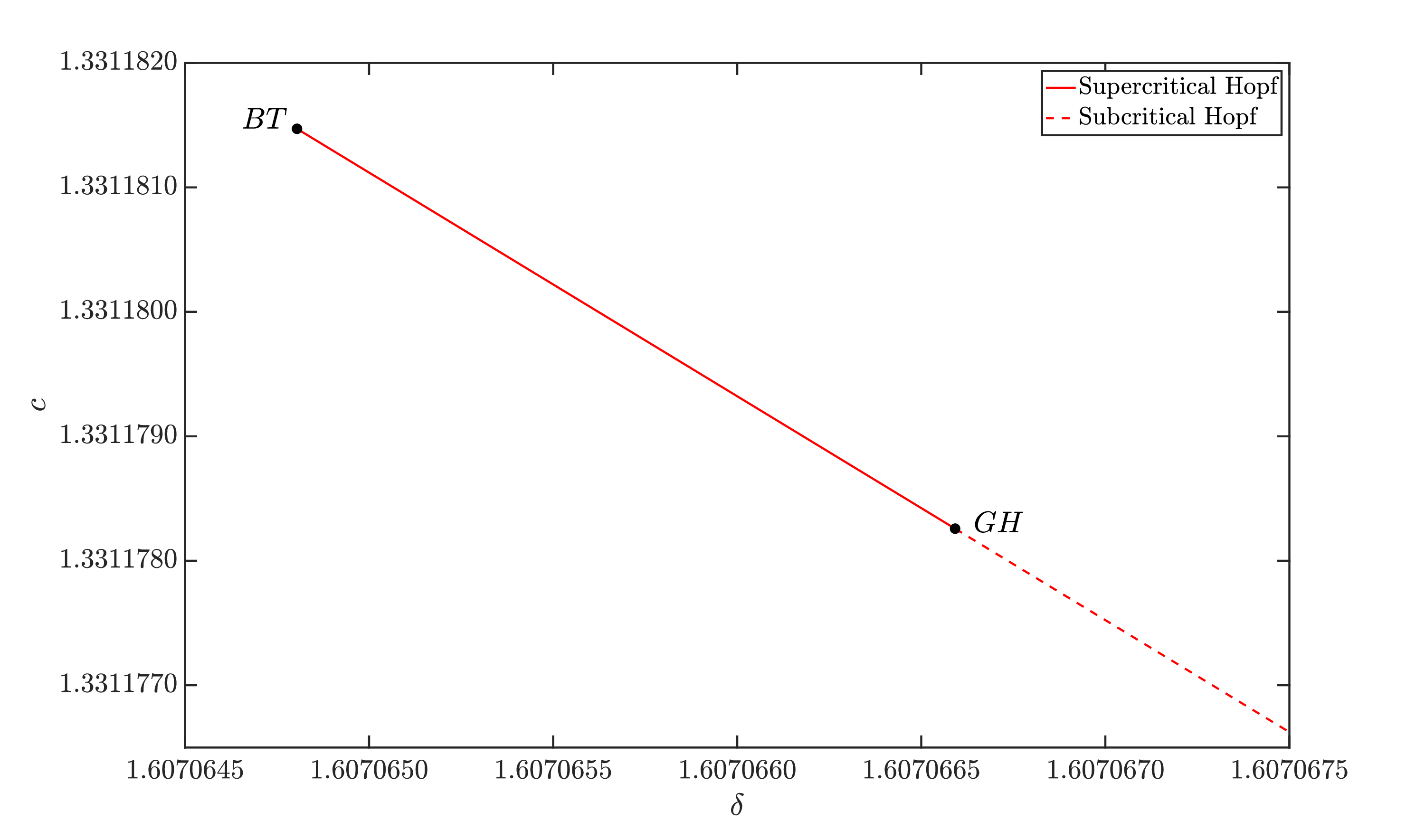}
    \\(a)
\end{minipage}
\hfill
\begin{minipage}[b]{0.48\textwidth}
    \centering
    \includegraphics[width=\textwidth, height=5.2cm]{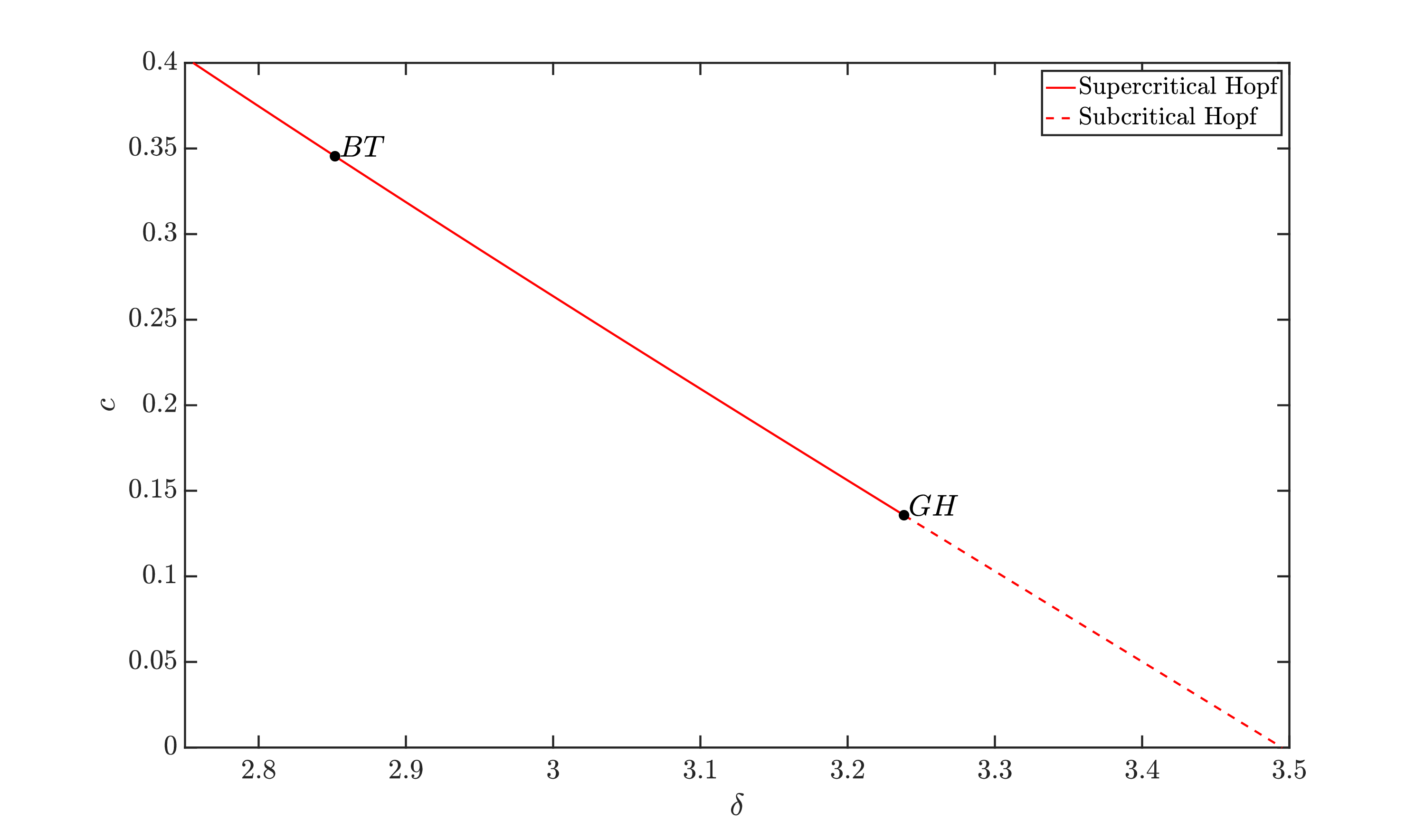}
    \\(b)
\end{minipage}
\caption{Two-parameter bifurcation diagrams of system \eqref{eq:model_BT} in the $(\delta,c)$-plane with $(\eta, K, \alpha, \xi, p)=(3.5700505,\,46.409597,\,0.26,\,0.675,\, 1.4)$. (a) Magnified view in a neighborhood of the Bogdanov-Takens (BT) point. GH denotes the Generalized Hopf (Bautin) bifurcation point. The solid and dashed red curves represent the supercritical and subcritical Hopf bifurcation curves, respectively. (b) Further continuation of the Hopf bifurcation curve.}
\label{fig:BT_diag}
\end{figure}

\begin{figure}
\centering
\begin{minipage}[t]{0.49\textwidth}
\centering
\includegraphics[width=7.2cm,height=5.3cm, keepaspectratio]{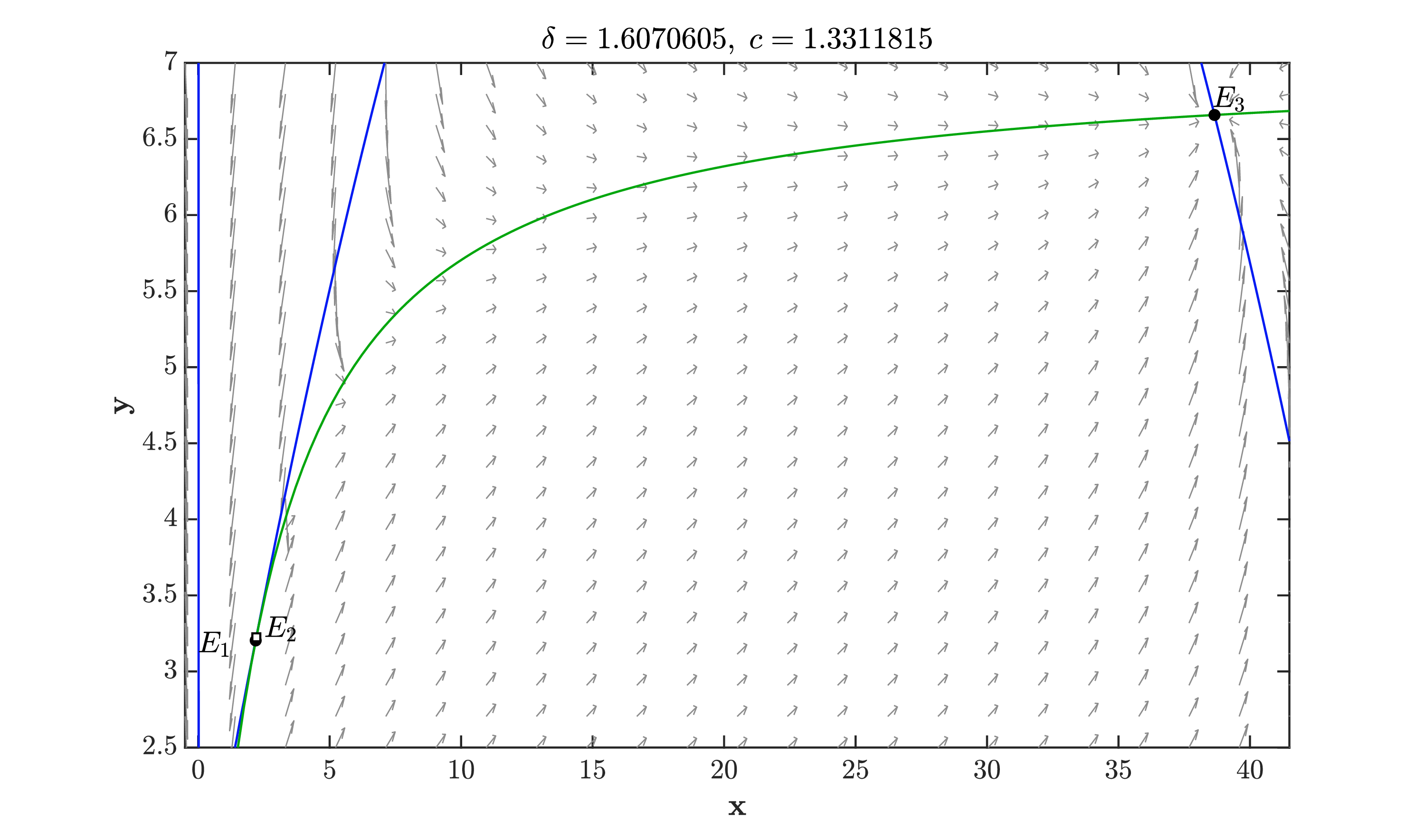}\\(a)

\end{minipage}
\hfill
\begin{minipage}[t]{0.49\textwidth}
\centering
\includegraphics[width=7.2cm,height=5.3cm, keepaspectratio]{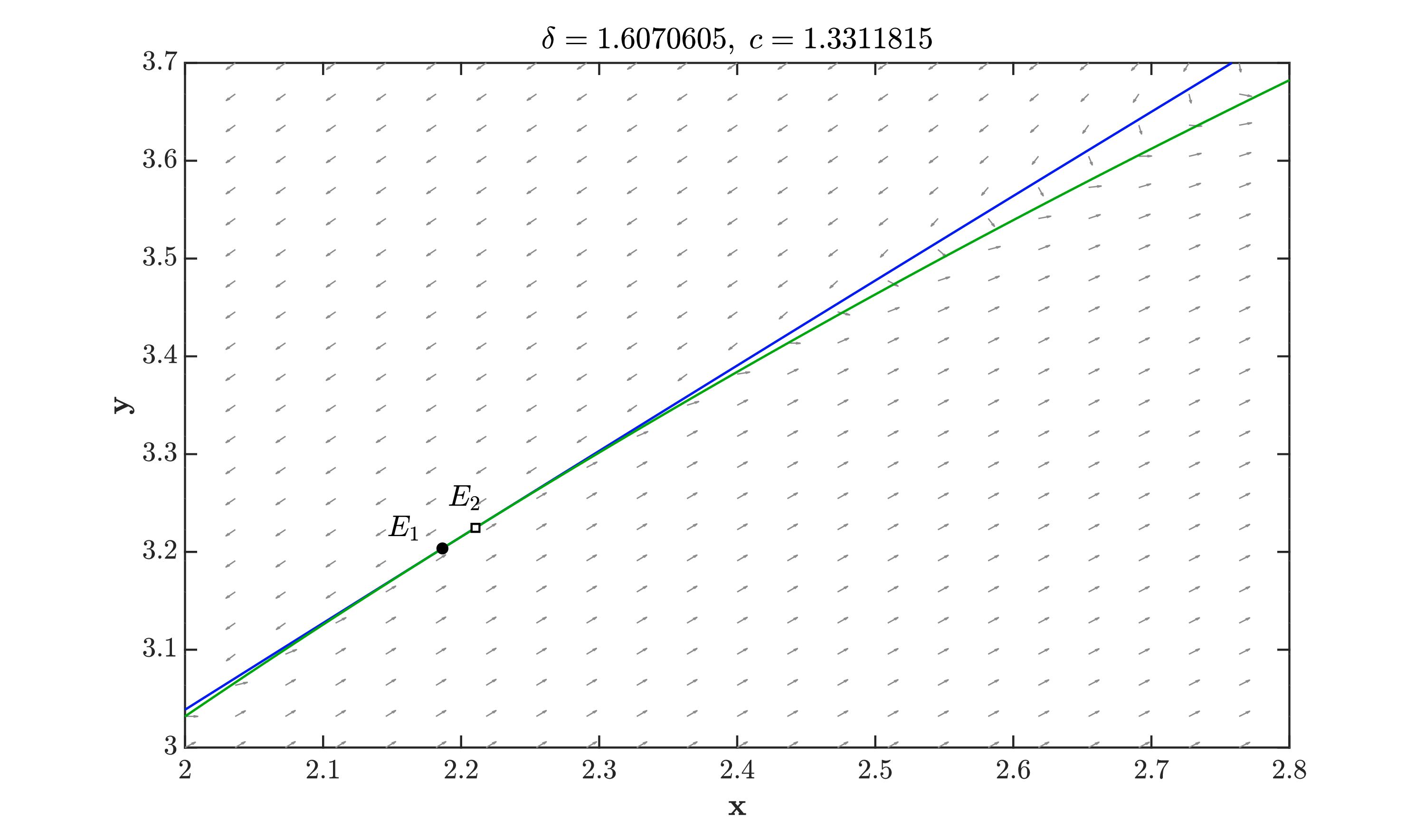}
\\(b)
\end{minipage}

\begin{minipage}[t]{0.49\textwidth}
\centering
\includegraphics[width=7.2cm,height=5.3cm, keepaspectratio]{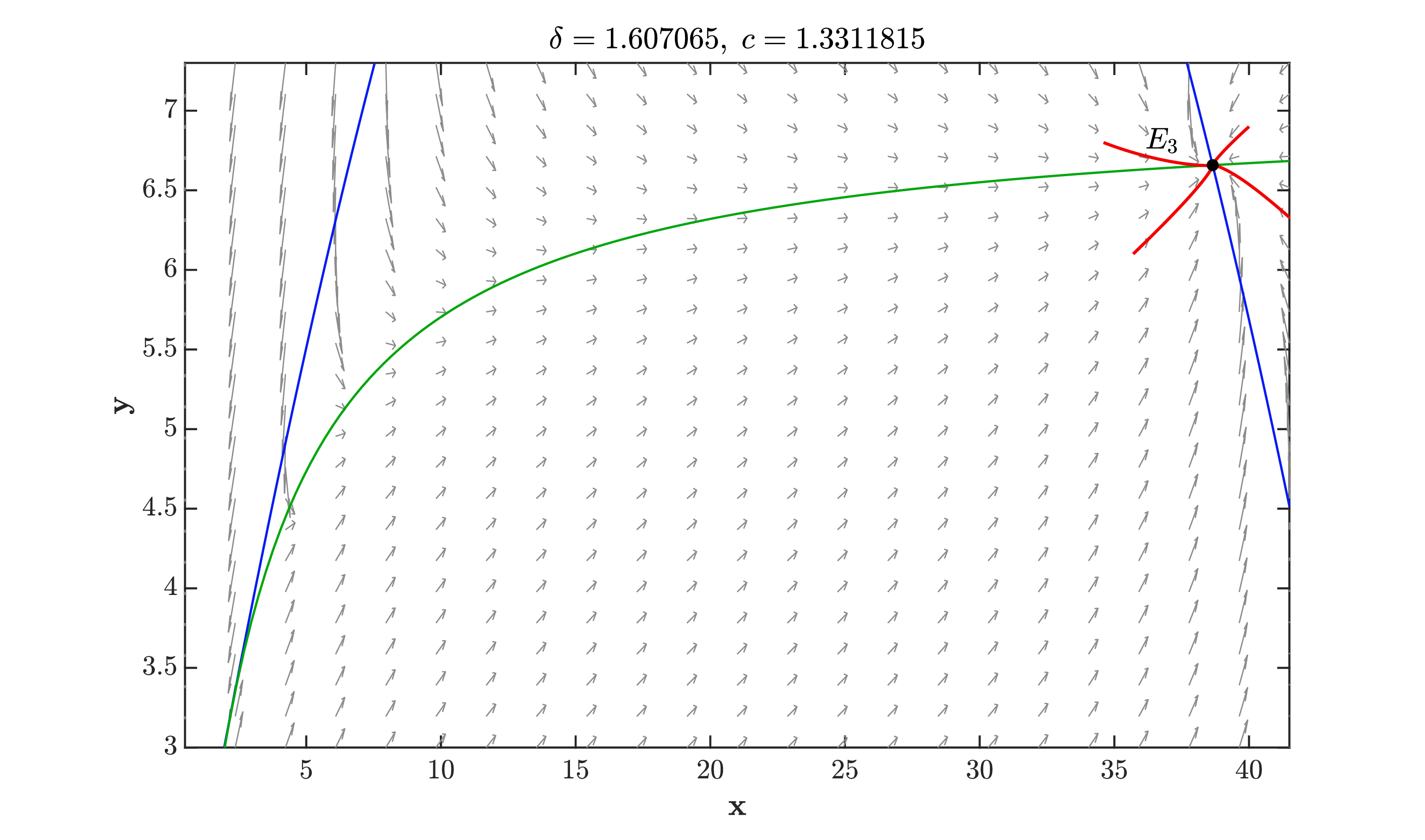}
\\(c)
\end{minipage}
\hfill
\begin{minipage}[t]{0.49\textwidth}
\centering
\includegraphics[width=7.2cm,height=5.3cm, keepaspectratio]{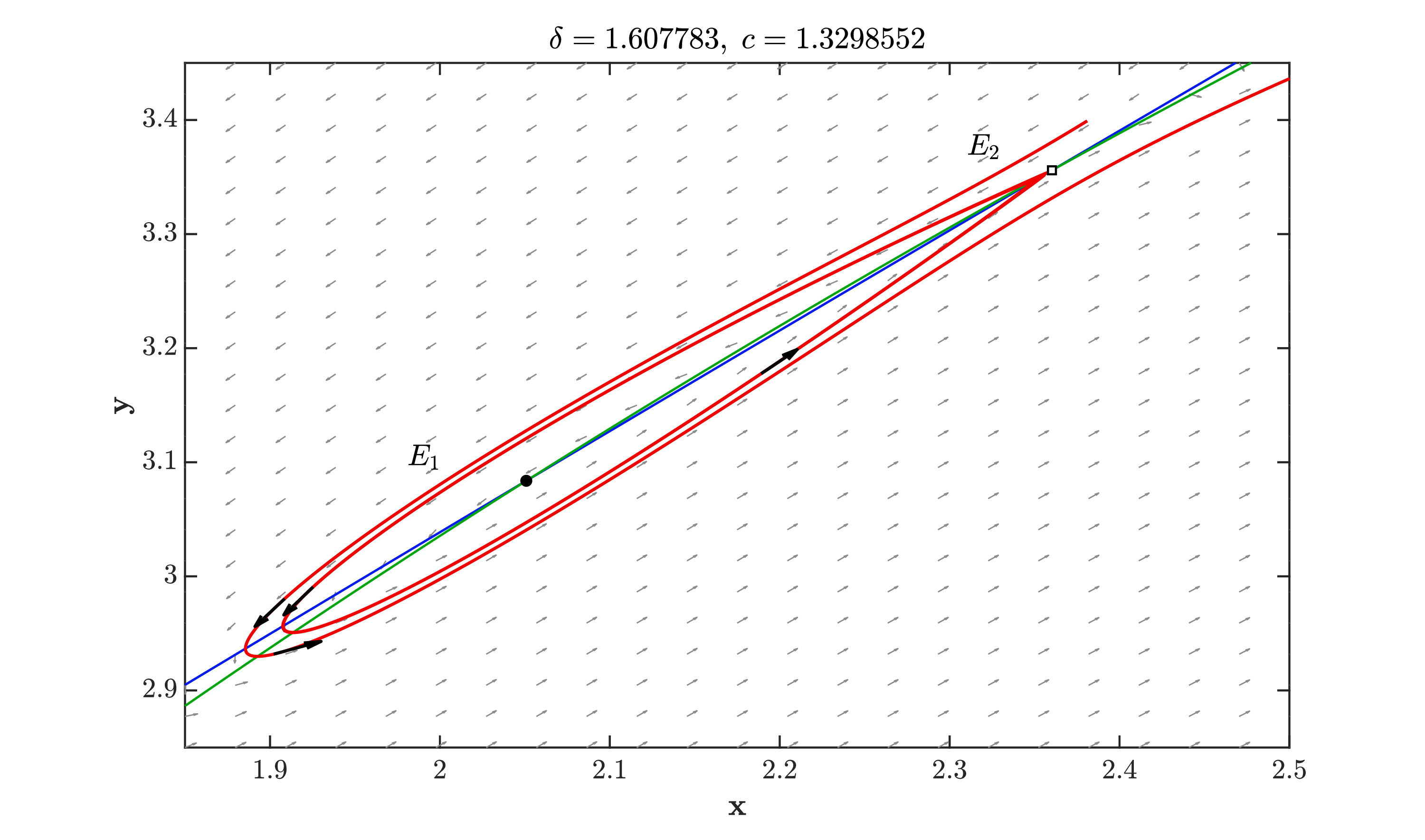}
\\(d)
\end{minipage}

\begin{minipage}[t]{0.49\textwidth}
\centering
\includegraphics[width=7.2cm,height=5.3cm, keepaspectratio]{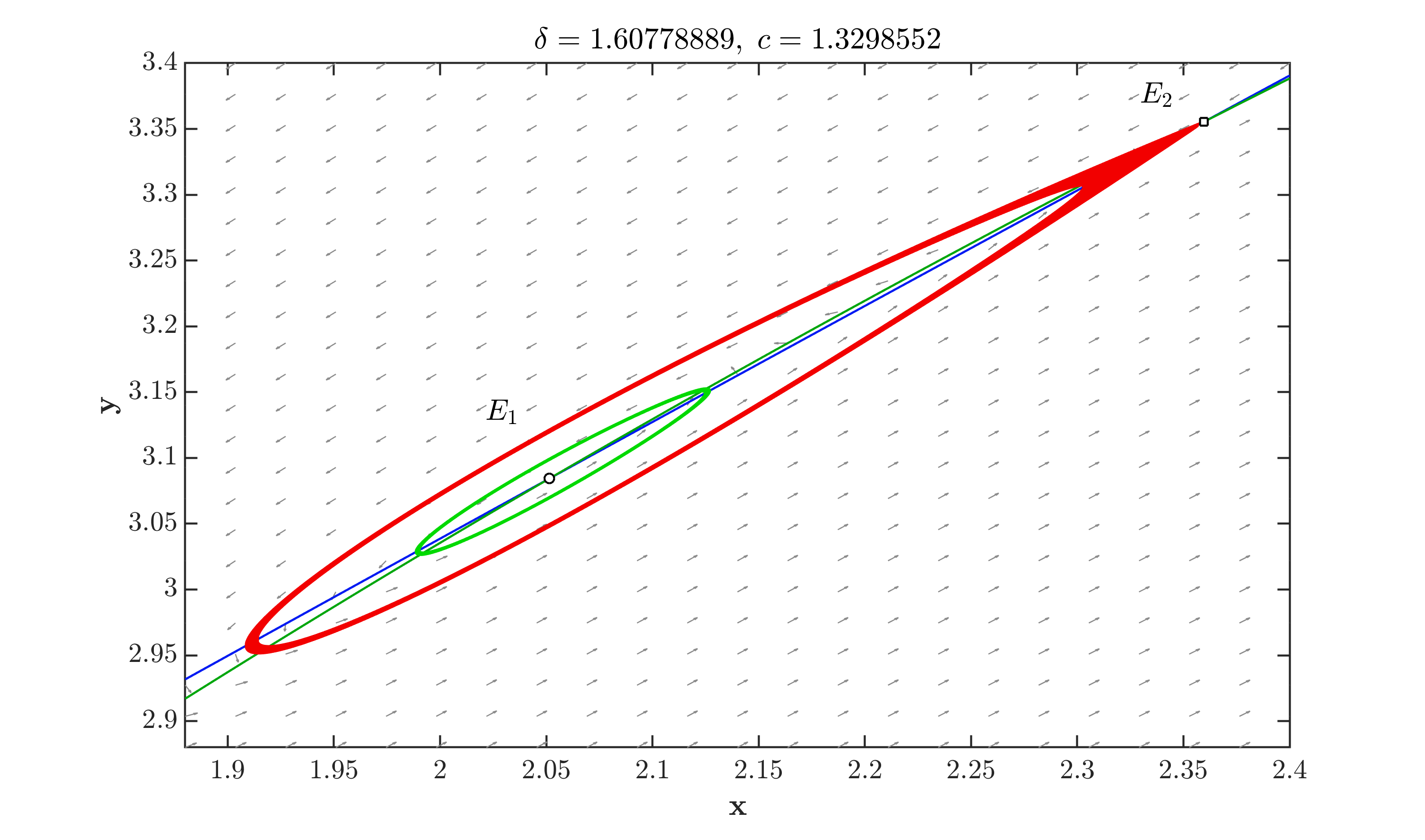}
\\(e)
\end{minipage}

\caption{Phase portraits corresponding to the bifurcation diagrams in Fig.~\ref{fig:BT_diag}. Throughout the figure, filled circles denote stable nodes/foci, open circles denote unstable nodes/foci, and open squares denote saddles. (a) Phase portrait immediately before the saddle-node bifurcation, where three interior equilibria coexist. (b) Enlarged view of the phase portrait near the Bogdanov-Takens point, showing the close proximity of the equilibria $E_1$ and $E_2$. (c) Dynamics at the Bogdanov-Takens point, where the interior equilibrium $E_3$ is a stable node. (d) Homoclinic loop surrounding the unstable equilibrium $E_1$. (e) Stable limit cycle enclosing the unstable equilibrium $E_1$ after the Hopf bifurcation.}
\label{fig:BT_phase_portraits}
\end{figure}

Further in Fig.~\ref{fig:nf_BT_diag}, we present the biparametric bifurcation diagram of focus type BT bifurcation for system \eqref{eq:perturbed_system_nf} in $(\delta,c)$ plane using \textit{Matcont} program. Here, we fix $(\eta, K, \alpha, \xi) = (0.41635802,\, 6,\,0.2,\,0.5)$. The system \eqref{eq:perturbed_system_nf} undergoes a series of bifurcations and exhibit complex dynamics as parameters vary in different regions of $(\delta,c)$ plane. Some of these dynamics have been illustrated in the phase portraits given in Fig. \ref{fig:nilpotent_focus_dynamics}. Fig.~\ref{fig:nilpotent_focus_dynamics}(a) corresponds to parameter values $(\delta,c)=(0.054965,\,0.2703)$ illustrates the phase portrait in the region preceding the Hopf bifurcation. The system possesses a unique biologically relevant interior equilibrium $E_1$, which is a stable focus. Fig.~\ref{fig:nilpotent_focus_dynamics}(b) corresponds to parameter values immediately after crossing the subcritical Hopf bifurcation curve. The equilibrium $E_1$ remains locally asymptotically stable but is surrounded by an unstable limit cycle. Initial conditions inside the unstable cycle converge to $E_1$ (red), whereas trajectories outside the unstable cycle are attracted to the distant stable equilibrium $E_3$ (green). Fig.~\ref{fig:nilpotent_focus_dynamics}(c)  illustrates the emergence of a stable limit cycle surrounding the equilibrium $E_1$. Trajectories initiated both inside (green) and outside (red) the periodic orbit converge to the same stable limit cycle, indicating that the periodic orbit is asymptotically stable. The coexistence of the stable equilibrium $E_3$ and the stable periodic orbit demonstrates bistability in the system, with the saddle equilibrium $E_2$ separating their respective basins of attraction. At the bifurcation value shown in Fig.~\ref{fig:nilpotent_focus_dynamics}(d), the stable and unstable manifolds of the saddle coincide to form a homoclinic loop, resulting in the disappearance of the stable periodic orbit. Figure~\ref{fig:nilpotent_focus_dynamics}(e) illustrates the dynamics in the neighborhood of the second BT point. The system possesses three biologically relevant equilibria: an unstable focus $E_1$, a saddle equilibrium $E_2$, and a stable equilibrium $E_3$. The equilibria $E_1$ and $E_2$ lie extremely close to each other, indicating that the system is close to the saddle-node bifurcation curve where these two equilibria coalesce. Finally, in Fig.~\ref{fig:nilpotent_focus_dynamics}(f) a stable limit cycle surrounds the equilibrium $E_1$ illustrating the oscillatory dynamics associated with the second branch of the unfolding.
\begin{figure}
    \centering
    \includegraphics[width=\linewidth]{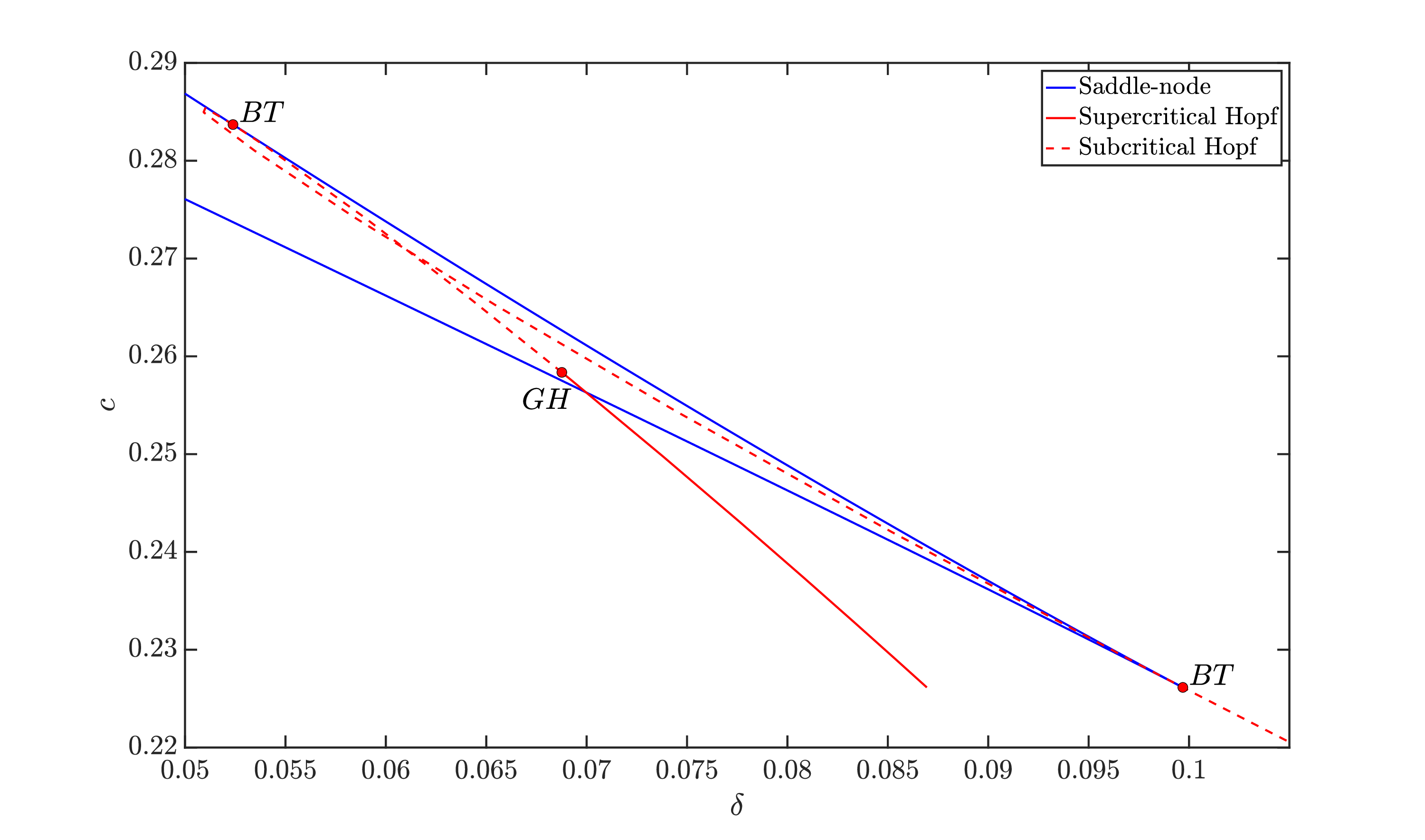}
    \caption{Bifurcation diagram of the focus-type Bogdanov-Takens bifurcation for system \eqref{eq:model_BT} in the $(\delta,c)$ parameter plane with $K=6$, $\alpha=0.2$, $\xi=0.5$, $\eta=0.41635802$, and $p=2$. The solid and dashed red curves denote the supercritical and subcritical Hopf bifurcation curves, respectively, while the blue curves denote the saddle-node bifurcation curves. The points labeled $\mathrm{BT}$ and $\mathrm{GH}$ represent the Bogdanov-Takens and generalized Hopf (Bautin) bifurcation points, respectively.}
    \label{fig:nf_BT_diag}
\end{figure}

\begin{figure}
\centering
\begin{minipage}[t]{0.49\textwidth}
\centering
\includegraphics[width=7.2cm,height=5.3cm, keepaspectratio]{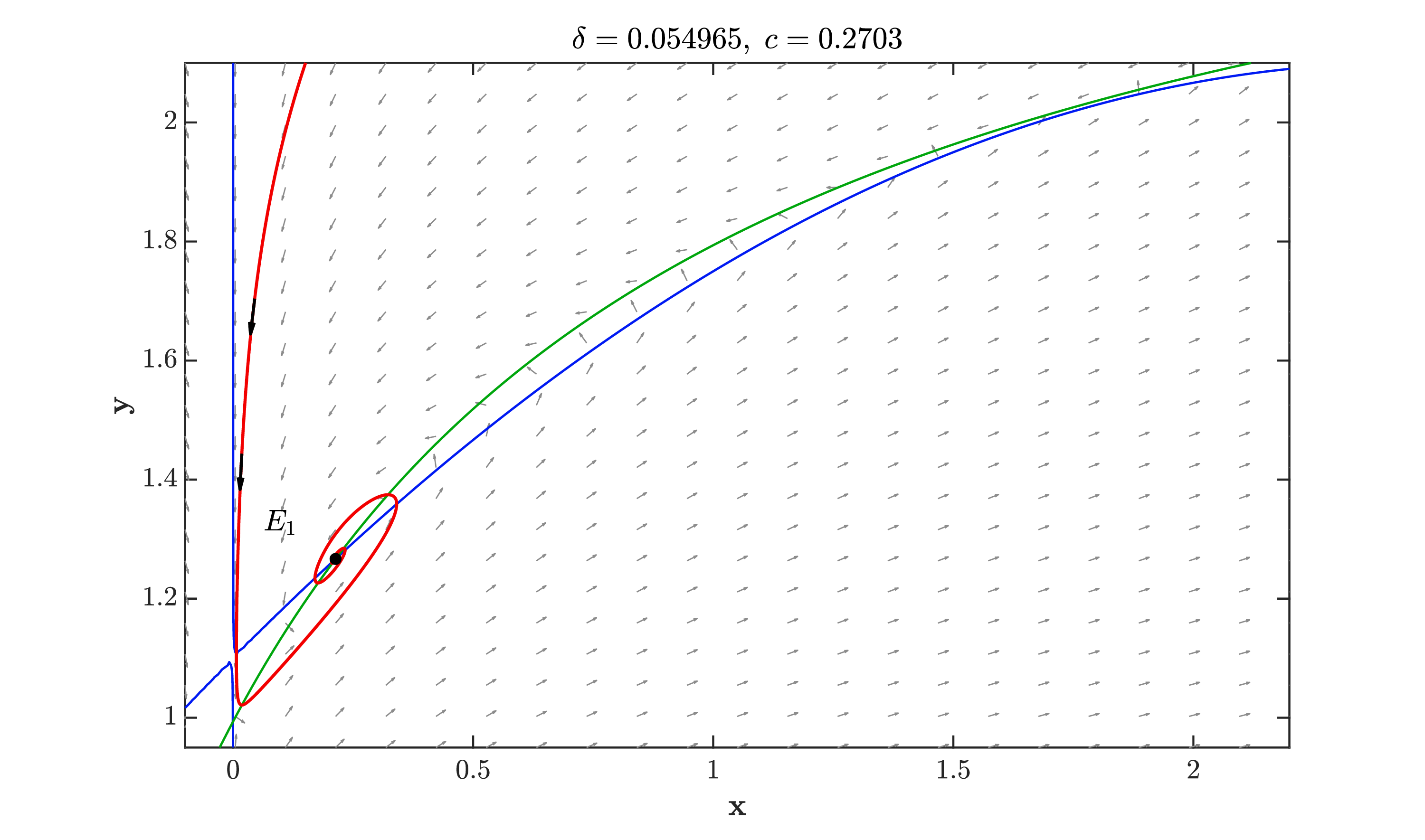}\\(a)

\end{minipage}
\hfill
\begin{minipage}[t]{0.49\textwidth}
\centering
\includegraphics[width=7.2cm,height=5.3cm, keepaspectratio]{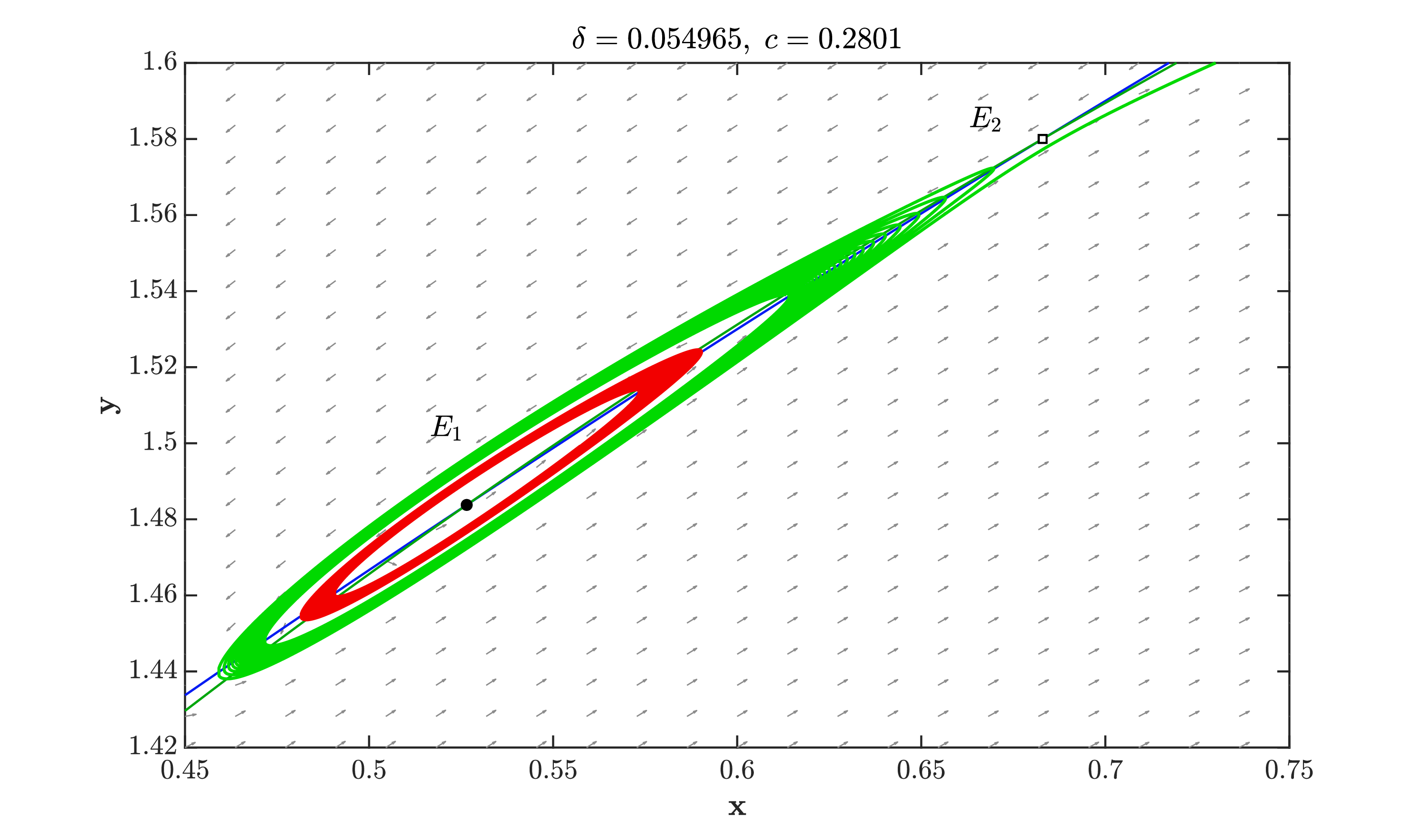}
\\(b)
\end{minipage}

\begin{minipage}[t]{0.49\textwidth}
\centering
\includegraphics[width=7.2cm,height=5.3cm, keepaspectratio]{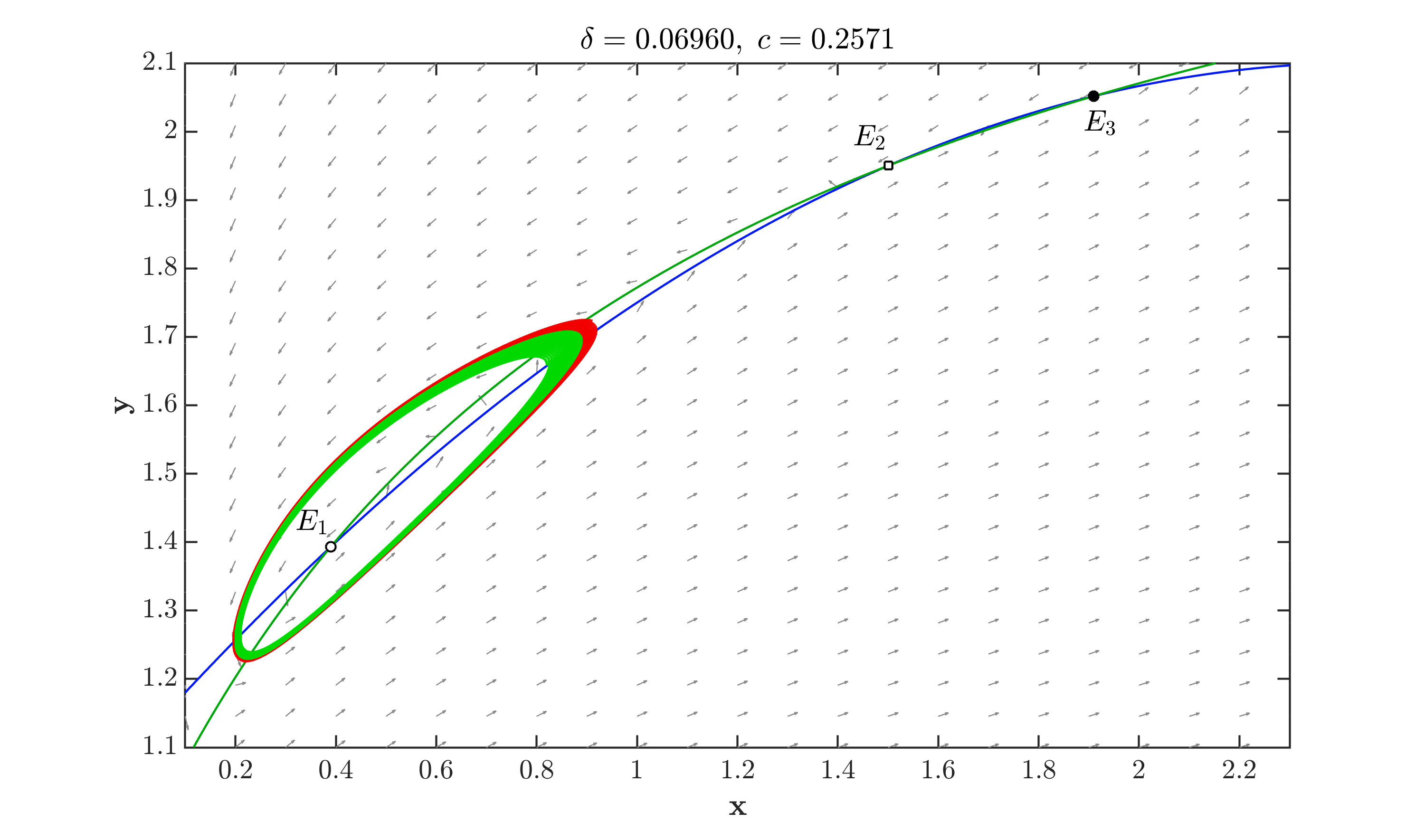}
\\(c)
\end{minipage}
\hfill
\begin{minipage}[t]{0.49\textwidth}
\centering
\includegraphics[width=7.2cm,height=5.3cm, keepaspectratio]{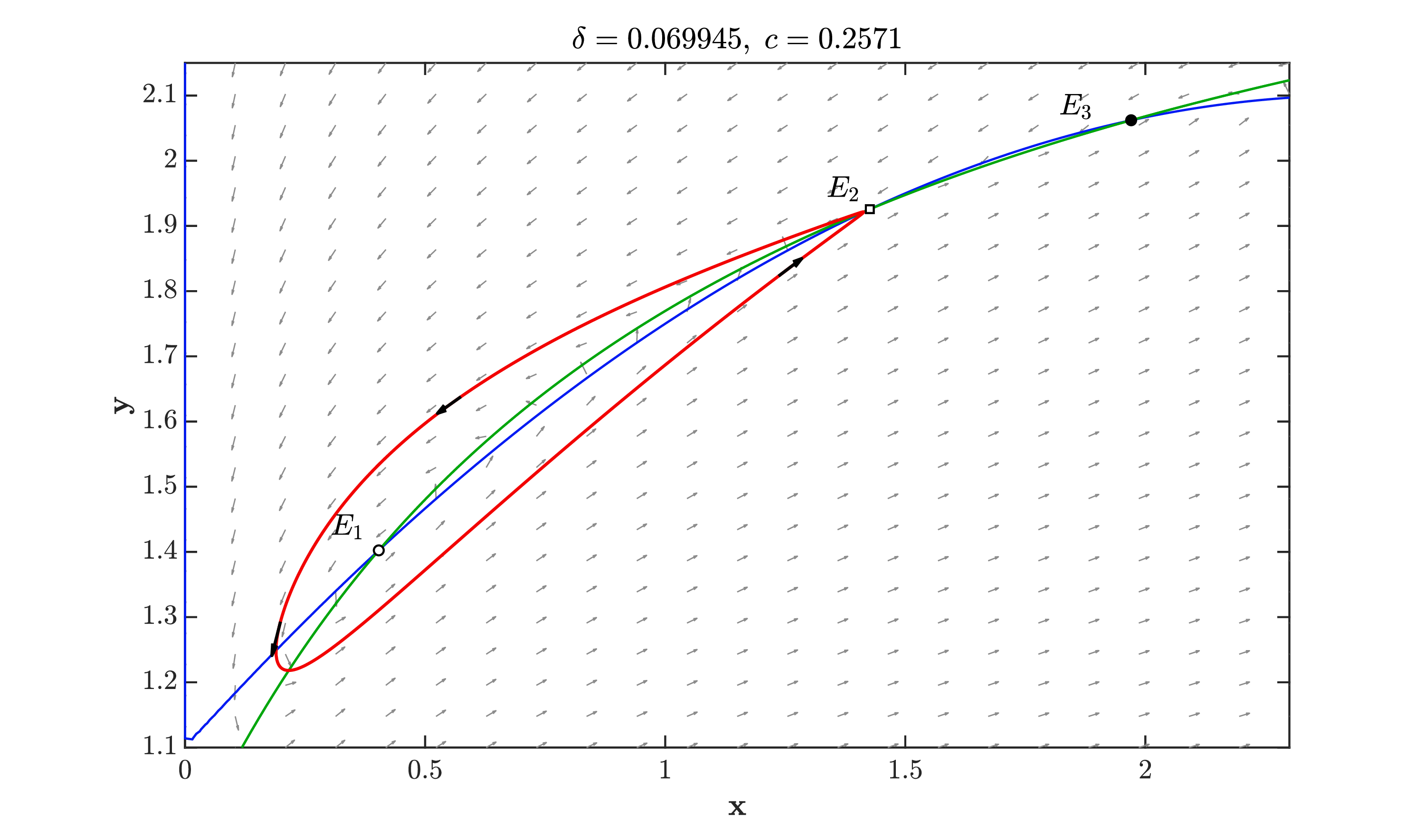}
\\(d)
\end{minipage}

\begin{minipage}[t]{0.49\textwidth}
\centering
\includegraphics[width=7.2cm,height=5.3cm, keepaspectratio]{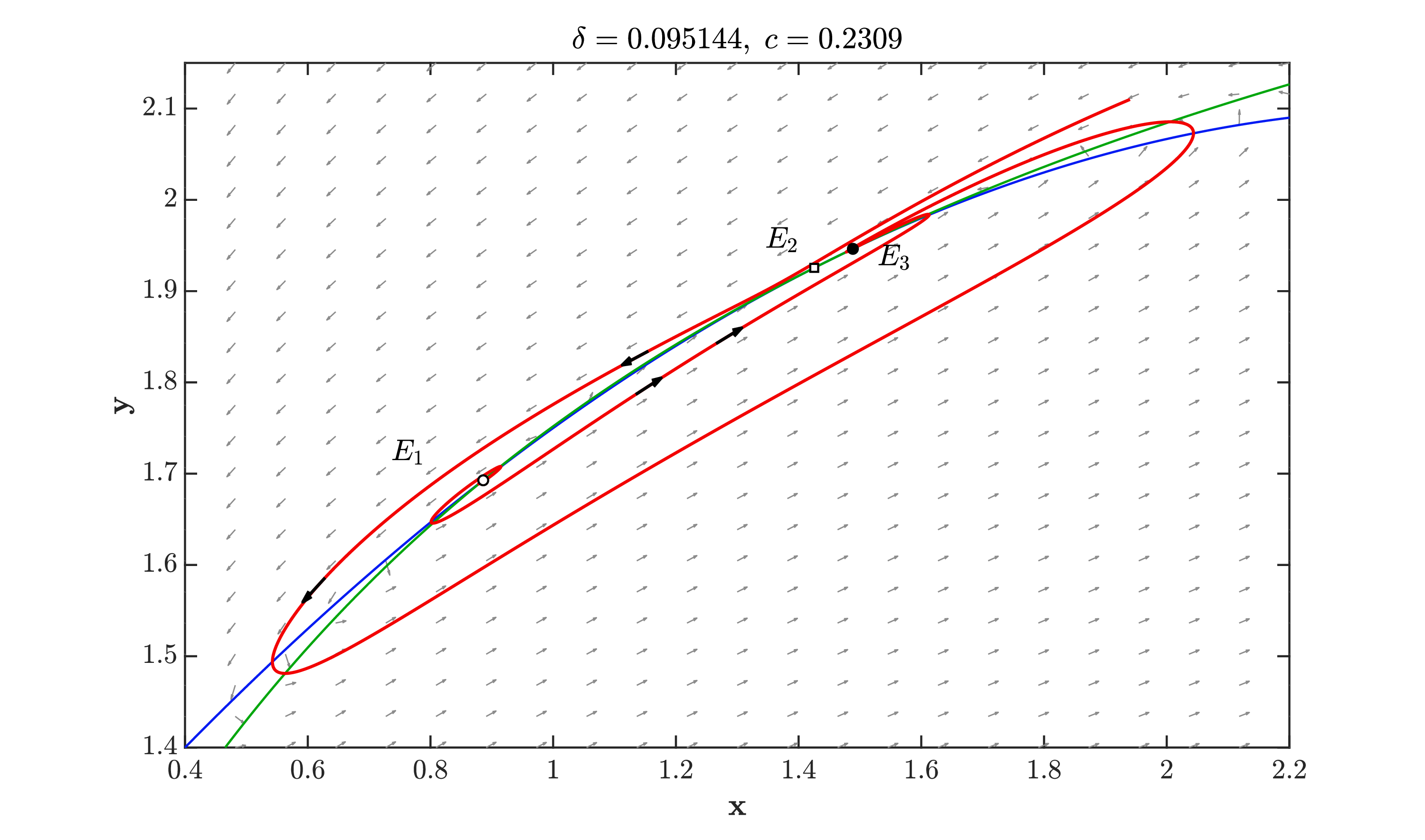}
\\(e)
\end{minipage}
\hfill
\begin{minipage}[t]{0.49\textwidth}
\centering
\includegraphics[width=7.2cm,height=5.3cm, keepaspectratio]{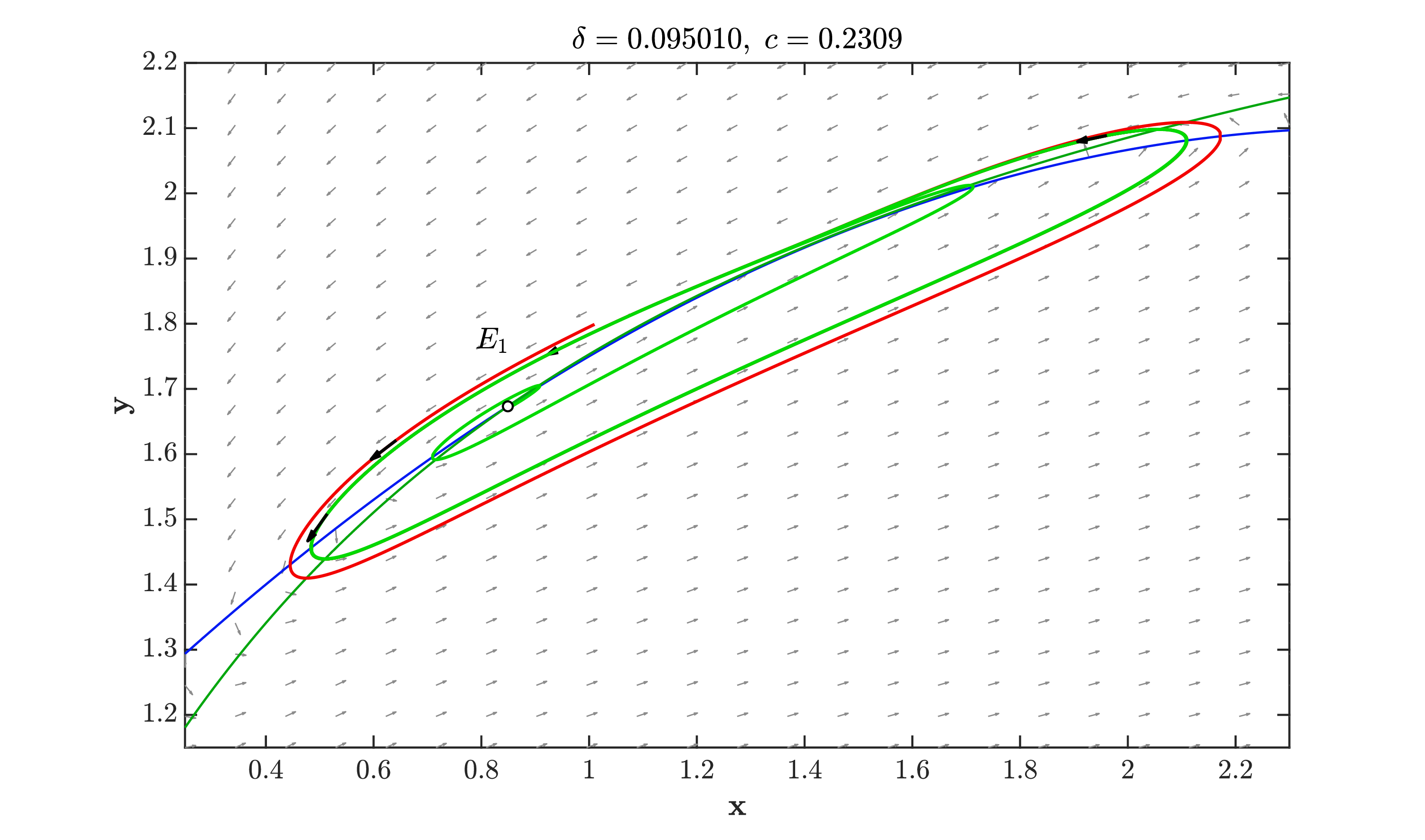}
\\(f)
\end{minipage}

\caption{Phase portraits corresponding to the bifurcation diagram in Fig.~\ref{fig:nf_BT_diag}. Throughout the figure, filled circles denote stable nodes/foci, open circles denote unstable nodes/foci, and open squares denote saddles. (a)-(f) illustrate the local dynamics for different values of $(\delta,c)$ selected from the distinct bifurcation regions. The panels show the transition from a stable focus to an unstable limit cycle, the emergence of a stable limit cycle through a supercritical Hopf bifurcation, the formation of a homoclinic loop, the coexistence of three interior equilibria near the saddle-node bifurcation curve, and the appearance of a single stable limit cycle beyond the generalized Hopf bifurcation.}
\label{fig:nilpotent_focus_dynamics}
\end{figure}

\section{Discussion \& Conclusion}\label{sec:conclusion}
In this paper, we studied an AF mediated predator-prey model with generalized predator competition. Motivated by the role of supplementary food in biological control and the significance of predator competition in the predator-prey dynamics, we considered a nonlinear competition term of the form $-y^p$, $(1<p\leq2)$, which extends the classical AF with quadratic competition model and covers a wide range of biological applications. This was first studied in \citep{verma2026additional} but with limited attention on the qualitative dynamics. With this study, we carried out a comprehensive qualitative and bifurcation analysis focusing on higher-order bifurcations and demonstrated that the interaction between AF provision and generalized predator competition can generate remarkably rich dynamical behavior.

We first analyzed the equilibrium structure of the model by establishing a one-to-one correspondence between interior equilibria. We derived explicit conditions for the multiplicity of positive equilibria in Theorem~\ref{thm:interior_eq} followed by Corollary~\ref{cor:types_eq_case_1p2}, and showed that the system can admit at most three biologically feasible interior equilibria. These positive roots and multiplicity of interior equilibria is illustrated in Fig.~\ref{fig:root_configurations}. It is shown that simple interior equilibria are hyperbolic, while double $(E_d)$ and triple $(E_t)$ positive equilibria correspond to higher-order degeneracies. Building upon this equilibrium classification, we investigated the local and global dynamics near these degenerate states through a detailed bifurcation analysis. For the analysis of the double positive equilibrium, we chose $(c,\eta,\delta,K)$ as four bifurcation parameters. By computing the universal unfolding in Theorem~\ref{thm:unfolding}, we found the existence of Bogdanov-Takens bifurcation of codimension at least $4$. In Section~\ref{sec:hopf_codim3} and Section~\ref{sec:homoclinic_codim3}, we further showed that the nilpotent singularity acts as an organizing center for the nearby dynamics, in particular, it gives rise to Hopf bifurcation of codimension $3$ (see Corollary~\ref{cor:hopf_3}) and homoclinic bifurcation of codimension $3$ (see Corollary~\ref{cor:homo_3}). The coexistence of Hopf and homoclinic bifurcations creates the potential for rich oscillatory dynamics and multiple limit cycles. For the special case $p=2$, we proved in Theorem~\ref{thm:nilpotent_focus} that the triple root equilibrium gives rise to nilpotent focus-type Bogdanov-Takens bifurcation of codimension-$3$ and its unfolding is derived in Theorem~\ref{thm:nf_unfolding}. To the best of our knowledge, such higher-order bifurcations in AF mediated generalized predator competition model has not been established previously.  

To complement the theoretical analysis, numerical simulations and two-parameter continuation using \textit{Matcont} were carried out to illustrate the bifurcation structures governed by our theoretical analysis in the paper. The computed bifurcation diagrams, Fig. \ref{fig:BT_diag} and Fig. \ref{fig:nf_BT_diag}, confirmed the unfolding of both the cusp-type and focus-type BT singularities, respectively, while the phase portraits, Fig. \ref{fig:BT_phase_portraits} and Fig. \ref{fig:nilpotent_focus_dynamics}, captured the associated transitions between qualitatively distinct dynamical regimes. These include saddle-node bifurcation, the emergence of stable and unstable limit cycle through Hopf bifurcations, and the formation of homoclinic loops. Collectively, the numerical results not only corroborate the theoretical findings but also demonstrate that AF provision, when coupled with generalized predator competition, can generate highly intricate dynamics, including bistability, sustained oscillations, and multiple coexistence states, which may have important implications for the design and management of biological control strategies.

Without AF, for the classical competition case $p=2$, system \eqref{eq:model_BT} reduces to predator-prey model with predator competition studied by Lu and Huang \citep{lu2021global}. Therein they established the existence of nilpotent cusp of codimension-$2$ and nilpotent focus of codimension-$3$ Bogdanov-Takens bifurcation, degenerate Hopf bifurcation of codimension at most $2$. In contrast, introduction of AF and generalized predator competition leads to qualitatively different degeneracies and bifurcation scenarios. Also, for the special case $p=2$ in \eqref{eq:model_BT}, the system simplifies to AF competition model studied by Parshad \textit{et. al.} \citep{parshad2023additional}. While in their study, the authors numerically identified the occurrence of BT bifurcation of codimension-$2$, a rigorous analytical proof was not provided. This gap was partially addressed in our previous study \citep{verma2026additional} with $\delta=0$, where the existence of codimension-$2$ BT bifurcation was established theoretically. The present work extends these results considerably by revealing and rigorously characterizing higher-order degeneracies, thereby providing a more complete understanding of prey-predator dynamics.  

From a biological standpoint, our results suggest that the population dynamics can behave in a highly nontrivial manner in the presence of AF and generalized predator competition. Although supplementary food may promote predator persistence and increase pest suppression, the incorporation of generalized predator competition shows the possibility of multiple coexistence states, oscillatory dynamics and complex bifurcation structures. The codimension-$4$ BT bifurcation structure identified in this work suggests that small variations in $4$ parameters simultaneously may lead to significant qualitative changes in the population dynamics. In addition, the detected limit cycles around the BT point indicate that the predator population may fluctuate over time. Consequently, this study highlights the importance of carefully balancing food supplementation while designing biological control strategies.

There are several applications of our results to the control and management of the soybean aphid. Field data strongly suggests that there are two limit cycles developing after the establishment of the Aphid by 2005. Herein, field data shows cycles of distinctly different amplitudes in the 2005-2010 period versus the 2010-onwards period. These phases correspond to the introduction of biological control agents such as the predator \emph{Harmonia axyridis}, as well as parasitoids for aphid control. This suggests that there is a natural tendency of this system to move from one cyclical pattern to another - and the introduction of AF for the predators could help drive these dynamics to the second limit cycle. The existence of which is guaranteed in Corollary~\ref{cor:hopf_3}. Note, predator competition alone perse, will not be able to achieve this, as current results \cite{lu2021global}, show degenerate Hopf bifurcation of codimension at most 2 - or the existence of 2 limit cycles. However, both of these cannot be stable, and the phase space structure would yield one stable and one unstable. The result of the current manuscript, Corollary~\ref{cor:hopf_3}, yields  a global Hopf bifurcation of codimension-3 and thus the possibility of the existence of three limit cycles around the BT point, of which 2 could be stable. This is clearly demonstrated via the $C_{3}$ point in Fig~\ref{BT4_figure} - and thus could match the two distinct limit cycles seen in the long term field data for soybean aphids and their predators via Figure \ref{fig:ltd}, and also reported in \cite{bahlai2026asynchronous, bahlai2015shifts}. This tells us that investment in AF methods for biological control could prove useful. One such conservation-based habitat management strategy is the STRIPS program. The Science-based Trials of Row-crops Integrated with Prairie Strips (STRIPS) project was initiated in Iowa to investigate whether strategically integrating small amounts of prairie strips within row-crop systems, such as corn--soybean fields, could improve both ecological and agricultural outcomes. Researchers have reported that converting as little as 10\% of a crop field into prairie strips can reduce nutrient loss, improve water quality,  enhance soil retention, and promote biodiversity within crop systems \cite{LS17} all while maintaining per-acre crop yield. These prairie areas also create favorable conditions for beneficial arthropods and other wildlife species. Studies suggest that prairie strips may increase the abundance and diversity of aphidophagous predators, although their direct effects on biological control may vary across systems \cite{cox2014impact,tooker2020balancing,kemmerling2022prairie}.
Thus, based on our results a relevant pest management strategy would be to promote conservation methods such as prairie strips, and boosting predator efficacy via the AF sources planted within them. This would drive dynamics towards the second limit cycle, where based on the local peaks of the cycles both pesticide application and parasitoid introduction could be timed optimally for maximal control. Such an approach could be generalized to other invasive pests as well, where multiple peak dynamics are reported across longer time scales.

Furthermore, this study focused on the biologically relevant regime $1<p\leq2$ of the generalized competition term. The sublinear case $0<p<1$ will be investigated in future work, as we expect it to exhibit substantially different dynamics arising due to finite time extinction \citep{banerjee2025two, PAT21}. Other future directions involve studying the model with spatial heterogeneity, time delays in the handling process, stochastic disturbances, effects of pest refuge.

\section{Appendix}
\subsection{Proof of Lemma~\ref{lem:E4_stability_unidirectional}}
\label{app:E4_stability_unidirectional}
\begin{proof}
The general Jacobian matrix  \eqref{general_jacobian_uni} at $(0,0)$ becomes,

\begin{equation*} 
 {J_0} = \begin{bmatrix}
1  & 0
  \\ 
0 &  \frac{\eta  \xi}{1+\alpha \xi}  - \delta 
  \end{bmatrix}
\label{pred_prey_free_jacobian_uni}
 \end{equation*}
Upon solving the characteristic equation we have,
 \begin{equation*}
 \lambda_1 = 1 \ \text{and} \ \lambda_2 = \frac{\eta  \xi}{1+\alpha \xi}  - \delta 
    \end{equation*}
Therefore, the extinction state $(0,0)$ always exists and is unstable. 
\label{proof_E4_stability_unidirectional}
\end{proof}

\subsection{Proof of Lemma~\ref{lem:E3_stability_unidirectional}}
\label{app:E3_stability_unidirectional}
\begin{proof}
 From the nullcline $ \Dot{x}=0$,
  
\begin{equation*}
     \hat{x}\  \left (1-\frac{ \hat{x}}{K} \right) = 0
\end{equation*}
either $ \hat{x} = 0$ \  or,\  $\hat{x}={K}$. So, the predator-free state $(K,0)$ always exists. Now, the general Jacobian matrix  \eqref{general_jacobian_uni} at $(K,0)$ becomes,

\begin{equation*} 
 {J_K} = \begin{bmatrix}
-1 & \frac{- K}{1+K+\alpha \xi} \vspace{0.15cm}
  \\ 
0 &  \frac{\eta \left(K + \xi\right)}{1+K+\alpha \xi}  - \delta 
  \end{bmatrix}
\label{pred_free_jacobian_uni}
 \end{equation*}

As $J_K$ is upper triangular, the eigenvalues are

\begin{equation*}
\lambda_1=-1,
\qquad
\lambda_2=\frac{\eta(K+\xi)}{1+K+\alpha\xi}-\delta
\end{equation*}
Clearly, $\lambda_1<0$ then for,
\begin{equation*}
\lambda_2<0
\quad\Longleftrightarrow\quad
\delta>\frac{\eta(K+\xi)}{1+K+\alpha\xi}
\end{equation*}

Then both eigenvalues are negative, and $E_K$ is locally asymptotically stable. 

But if 
$$
\delta<\frac{\eta(K+\xi)}{1+K+\alpha\xi} \quad\Longleftrightarrow\quad\lambda_2>0
$$ then the eigenvalues have opposite signs and $E_K$ is a saddle point.

\label{proof_E3_stability_unidirectional}
\end{proof}

\subsection{Proof of Lemma~\ref{lem:E1_existence_unidirectional}}
\label{app:E1_existence_unidirectional}
\begin{proof}
\label{proof_E1_existence_unidirectional}
From the nullcline $\Dot{y}=0$ we have, 
\begin{equation*}
 \hat{y}\ \left \{  \frac{\eta  \xi}{1+\alpha \xi}  - \delta -c \xi ( \hat{y})^{p - 1} \right \} = 0
\end{equation*}
either $ \hat{y} =0 $ or,
\begin{equation*}
   \frac{  \eta \xi}{1+\alpha \xi} - \delta =  c \xi ( \hat{y})^{p - 1}
\end{equation*}
\begin{equation}
    \hat{y} = \left( \frac{\eta \xi-  \delta(1+\alpha \xi) }{c\xi(1+\alpha \xi)}\right)^{\frac{1}{p-1}}
   \label{explicit_y1_comp_ext}
\end{equation}
Thus, the equilibrium point ${E_P} = (0, \hat{y})$ exists if  $\delta < \dfrac{ \eta\xi }{1+ \alpha \xi}$.
\end{proof}

\subsection{Proof of Lemma~\ref{lem:stability_pest_ext_drift}}
\label{app:stability_pest_ext_drift}
\begin{proof}
\label{proof:stability_pest_ext_drift}

Using \eqref{explicit_y1_comp_ext}, the general Jacobian matrix  \eqref{general_jacobian_uni} at $(0,\hat{y})$ becomes, 
\begin{equation*} 
\hat{J} = \begin{bmatrix}
1- \dfrac{\hat{y}} {1 +\alpha \xi}   & 0   \vspace{0.25cm}
  \\ 
\dfrac{\eta  (1 + (\alpha - 1) \xi) \ \hat{y}}{(1+\alpha \xi)^2} & c \xi (1-p)(\hat{y})^{p-1}
\end{bmatrix}
\end{equation*} 
Now the characteristic equation is given as,
\begin{equation*}
 (c \xi (1-p)(\hat{y})^{p-1}-\lambda) \left(1- \dfrac{\hat{y}} {1 +\alpha \xi} - \lambda \right)  = 0
 \end{equation*}
\begin{equation*}
\lambda_1 = c \xi (1-p)(\hat{y})^{p-1} \ < 0 \ (\text{since} \  1<p \leq 2), \quad \lambda_2 = 1  -  \dfrac{\hat{y}} {1 +\alpha \xi} < 0 \Leftrightarrow \hat{y}>(1+\alpha \xi), 
\end{equation*}

Therefore,  the predator-only equilibrium point is locally asymptotically stable if  $p>1$ (which holds) and $\hat{y}> 1+\alpha \xi$ and is a saddle if $\hat{y}<1+\alpha \xi$. 
\end{proof}
\subsection{Proof of Theorem~\ref{thm:no_eq_existence}}
\label{app:no_eq_proof}
\begin{proof}
At an interior equilibrium, $y^*=G(x^*)$, so along the curve $y=G(x)$, 
\[
\dot y(x,G(x))
=G(x)\left[\eta R(x)-\delta-c\xi G(x)^{p-1}\right]
=G(x)k(x).
\]
Since $G(x)>0$ for all \(x\in(0,K)\), it follows that 
\[
\dot y(x,G(x))<0
\quad\Longleftrightarrow\quad
k(x)<0.
\]
Observe that 
\[
k(x)
=\eta R(x)-\delta-c\xi G(x)^{p-1}
\leq \eta R(x)-\delta,
\]
so a sufficient condition for $k(x)<0$ for all $x\in(0,K)$ is
\[
\delta \ge \eta \max_{x\in[0,K]}R(x).
\]
Indeed, under this condition,
\[
\dot y(x,G(x))
\le G(x)\left[\eta\max_{x\in[0,K]}R(x)-\delta-c\xi G(x)^{p-1}\right]
\le -c\xi G(x)^p<0.
\]
Now
\[
R(x)=\frac{x+\xi}{1+x+\alpha\xi},
\qquad
R'(x)=\frac{1+(\alpha-1)\xi}{(1+x+\alpha\xi)^2}.
\]
Hence \(R\) is increasing on \([0,K]\) if \(1+(\alpha-1)\xi>0\), decreasing on \([0,K]\) if
\(1+(\alpha-1)\xi<0\), and constant if \(1+(\alpha-1)\xi=0\). Therefore,
\[
\max_{x\in[0,K]}R(x)=
\begin{cases}
\dfrac{K+\xi}{1+K+\alpha\xi}, & \text{if } 1+(\alpha-1)\xi>0,\\[2ex]
\dfrac{\xi}{1+\alpha\xi}, & \text{if } 1+(\alpha-1)\xi\leq0.
\end{cases}
\]
Consequently, $k(x)<0$ for all $x\in(0,K)$ whenever
\[
\delta \ge
\begin{cases}
\eta\,\dfrac{K+\xi}{1+K+\alpha\xi}, & \text{if } 1+(\alpha-1)\xi\ge 0,\\[2ex]
\eta\,\dfrac{\xi}{1+\alpha\xi}, & \text{if } 1+(\alpha-1)\xi<0.
\end{cases}
\]
\end{proof}
\subsection{Proof of Theorem~\ref{thm:hopf_tranversality}}
\label{app:hopf_transversality_proof}
\begin{proof}
At $\delta=\delta_H$, the equilibrium $E_d=(x^*,G(x^*))$ satisfies
\[
\operatorname{tr}( J(x^*,G(x^*)))=0,
\qquad
\det( J(x^*,G(x^*)))>0,
\]
For a planar system, the eigenvalues satisfy
\[
\lambda_{1,2}(\delta)=\frac{\operatorname{tr}(J(\delta))}{2}\pm i\,\sqrt{\det(J(\delta))-\frac{1}{4}\operatorname{tr}^2(J(\delta))}.
\]
In particular, at $\delta=\delta_H$ we have $\operatorname{tr}\big(J(x^*,G(x^*))\big)=0$ and
$\det\big(\hat J(x^*,G(x^*))\big)>0$, hence
\[
\lambda_{1,2}(\delta_H)=\pm i\sqrt{\det\big( J(x^*,G(x^*))\big)}=: \pm i\omega_H,
\qquad \omega_H>0,
\]
so the eigenvalues are purely imaginary. Moreover, for $\delta$ near $\delta_H$,
\[
\Re\big(\lambda_{1,2}(\delta)\big)=\frac{1}{2}\operatorname{tr}\big(J(x^*(\delta),G(x^*(\delta)))\big).
\]
Consequently, the Hopf transversality condition
\[
\frac{d}{d\delta}\Re(\lambda(\delta))\Big|_{\delta=\delta_H}\neq 0
\]
is equivalent to
\[
\frac{d}{d\delta}\operatorname{tr}\big( J(x^*(\delta),G(x^*(\delta)))\big)\Big|_{\delta=\delta_H}\neq 0.
\]
Since interior equilibria satisfy $k(x,\delta)=0$ and $k_x(x^*,\delta_H)\neq 0$, the Implicit Function Theorem yields a $C^1$ branch $x^*(\delta)$.\\
Differentiating $k(x^*(\delta),\delta)=0$ gives
\[
\frac{dx^*}{d\delta}=\frac{1}{k_x(x^*,\delta)}.
\]
Define
\[
T(x)=q(x)G'(x)-(p-1)c\xi (G(x))^{p-1},
\]
so that
\[
\operatorname{tr}\big( J(x^*(\delta)\big),G(x^*(\delta)))=T(x^*(\delta)).
\]
By the chain rule,
\[
\frac{d}{d\delta}\operatorname{tr}( J) = T'(x^*)\,\frac{dx^*}{d\delta}.
\]
Since $\frac{dx^*}{d\delta}=\frac{1}{k_x(x^*,\delta_H)}$, it follows that
\[
\frac{d}{d\delta}\Re(\lambda(\delta_H))=
\frac{1}{2}\frac{T'(x^*)}{k_x(x^*,\delta_H)}.
\]
Therefore, the transversality condition $\frac{d}{d\delta}\Re(\lambda(\delta_H))\neq 0$ holds provided
\[
k_x(x^*,\delta_H)\neq 0
\quad\text{and}\quad
T'(x^*)\neq 0.
\]
Finally, differentiating $T(x)$ gives
\[
T'(x)=q'(x)G'(x)+q(x)G''(x)-(p-1)^2c\xi (G(x))^{p-2}G'(x),
\]
so that the sufficient condition for transversality at $E_d$ is
\[
T'(x^*)=q'(x^*)G'(x^*)+q(x^*)G''(x^*)-(p-1)^2c\xi (G(x^*))^{p-2}G'(x^*)\neq 0,
\]
together with $k_x(x^*,\delta_H)\neq 0$.
Hence the real part of the eigenvalues crosses the imaginary axis with nonzero speed,
and the transversality condition is satisfied.
\end{proof}

\label{sec:A4}
\subsection{Coefficients of system~\texorpdfstring{\eqref{eq:s5}}{}}
\label{Coeff_0}
\begin{align*}
        c_{20}&=b_{20}, \; c_{30}=b_{30}+2a_{20}^2+a_{11}b_{20}, \;  c_{03}=b_{03}, c_{21}=b_{21}-a_{11}a_{20}+3a_{30}-2a_{20}b_{02}, \; c_{04}=b_{04},\\ c_{12}&=b_{12}+a_{11}^{2}+2a_{21}+3a_{11}b_{02}+2b_{02}^{2}, \; c_{13}=b_{13}+2b_{02}b_{03}, \; c_{14}=b_{14}+3b_{02}b_{04}, \; c_{05}=b_{05},\\
        c_{40}&=b_{40}+\left(2a_{30}-b_{21}\right)a_{20}+\left(2a_{20}^{2}+\frac{1}{2}b_{30}\right)b_{02}+\left(\frac{3}{4}a_{11}^{2}+a_{21}+\frac{1}{4}b_{02}^{2}\right)b_{20}\\
        &\quad+a_{11}\left(a_{20}^{2}+b_{02}b_{20}+\frac{3}{2}b_{30}\right),\\
        c_{31}&=b_{31}-2a_{20}a_{21}+2a_{11}a_{30}+4a_{40}-a_{11}a_{20}b_{02}+a_{30}b_{02}-2a_{20}b_{02}^{2}-2a_{20}b_{12}+a_{11}b_{21}+b_{02}b_{21}, \\
        c_{22}&=b_{22}+\frac{1}{2}\Bigl(-3a_{11}^{3}+6a_{31}-10a_{11}^{2}b_{02}+4a_{21}b_{02}-2b_{02}^{3}-6a_{20}b_{03}+3b_{02}b_{12}+a_{11}\left(-9b_{02}^{2}+b_{12}\right)\Bigr),\\ 
        c_{50}&=b_{50}+b_{02}b_{40}+a_{31}b_{20}+a_{21}b_{02}b_{20}-a_{30}b_{21}+a_{21}b_{30}+\frac{1}{4}\Big(a_{20}^{2}\left(6b_{02}^{2}+4b_{12}\right)-b_{02}^{3}b_{20}+b_{02}^{2}b_{30} \\
        &\quad +a_{11}^{2}\left(2a_{20}^{2}-b_{02}b_{20}+5b_{30}\right) +4a_{20}\left(2a_{40}+4a_{30}b_{02}-b_{02}b_{21}-b_{31}\right)\\
        &\quad +a_{11}\left(8a_{20}a_{30}+4a_{20}^{2}b_{02}+4a_{21}b_{20}-2b_{02}^{2}b_{20}-4a_{20}b_{21}+6b_{02}b_{30}+8b_{40}\right)\Big), \\
        c_{41}&=\frac{1}{4}\Big(a_{11}^{3}a_{20}+16a_{40}b_{02}-9a_{30}b_{02}^{2}-2a_{20}b_{02}^{3}+4\left(-3a_{20}a_{31}+5a_{50}+b_{41}\right)+a_{11}^{2}\left(a_{30}+6a_{20}b_{02}+b_{21}\right)  \\
        &\quad -4a_{21}\left(2a_{30}+a_{20}b_{02}\right)+12a_{20}^{2}b_{03}-8a_{30}b_{12}-12a_{20}b_{02}b_{12}+b_{02}^{2}b_{21}-8a_{20}b_{22}+6b_{02}b_{31} \\
        &\quad+a_{11}\left(20a_{40}-4a_{30}b_{02}+7a_{20}b_{02}^{2}-4a_{20}b_{12}+2b_{02}b_{21}+6b_{31}\right)\Big),\\
        c_{32}&=\frac{1}{2}\Big(3a_{11}^{4}-4a_{21}^{2}+8a_{41}+11a_{11}^{3}b_{02}+10a_{31}b_{02}+2b_{02}^{4} -2a_{21}\left(b_{02}^{2}+b_{12}\right)-6a_{20}b_{13}+4b_{02}b_{22} \\
        &\quad-6a_{30}b_{03}-12a_{20}b_{02}b_{03}+b_{02}^{2}b_{12}-a_{11}^{2}\left(2a_{21}-13b_{02}^{2}+b_{12}\right) \\
        &\quad+a_{11}\left(4a_{31}-8a_{21}b_{02}+7b_{02}^{3}-2b_{02}b_{12}+2b_{22}\right)+2b_{32}\Big),\\
        c_{23}&=b_{23}-\left(a_{11}^{2}+2a_{21}+3a_{11}b_{02}-b_{02}^{2}\right)b_{03}-4a_{20}b_{04}+\frac{1}{2}\left(a_{11}+5b_{02}\right)b_{13}.
    \end{align*}
    
\subsection{Coefficients of system~\texorpdfstring{\eqref{eq:unfolding_system2}}{}}
\label{Coeff_1}
\begin{align*}
    h_{00}(0)&=h_{10}(0)=h_{01}(0)=h_{11}(0)=0, \quad h_{20}(0)=0.0623854, \quad h_{02}(0)=0.350128,\\
    h_{30}(0)&=0.00335219, \quad h_{21}(0)=-0.0544774, \quad h_{12}(0)=- 0.180918, \quad h_{03}(0)=0.0183018, \\
    h_{40}(0)&=-0.00329723, \quad h_{31}(0)=0.0195721, \quad h_{22}(0)=0.0945538, \quad h_{13}(0)=- 0.01381,   \\
    h_{04}(0)&=0.00349563, \quad h_{50}(0)=0.00118459, \quad h_{41}(0)=-0.0061707, \quad h_{32}(0)=- 0.0468387, \\ 
    h_{23}(0)&=0.00868923, \quad h_{14}(0)=- 0.0041502, \quad h_{05}(0)=0.00086796.
\end{align*}
\subsection{Functions in the transformation~\texorpdfstring{\eqref{eq:unfolding_transf_2}}{}}\label{Coeff_2}
    \begin{align*}
        \Phi(x_1,y_1)&=\frac{h_{02}}{2}x_1^2
+\frac{h_{03}}{2}x_1^2y_1
+\frac{h_{04}}{2}x_1^2y_1^2
+\frac{h_{05}}{2}x_1^2y_1^3
+\frac{1}{6}
\left(2h_{02}^2-2h_{01}h_{03}-5h_{00}h_{04}+h_{12}
\right)x_1^3\\
&
+\frac{1}{24}
\Big(
6h_{02}^3
+6h_{01}^2h_{04}
-11h_{04}h_{10}
-4h_{03}h_{11}
+h_{02}(-12h_{01}h_{03}-24h_{00}h_{04}+7h_{12})\\
&\qquad\qquad
-2h_{01}h_{13} +h_{00}(9h_{03}^2+32h_{01}h_{05}-5h_{14})
+2h_{22}
\Big)x_1^4\\
&
+\frac{1}{120}
\Big(
24h_{02}^4
-63h_{00}^2h_{04}^2
-24h_{01}^3h_{05}
-357h_{00}^2h_{03}h_{05}
+33h_{03}^2h_{10}
+33h_{00}h_{05}h_{11}\\
&\quad
-2h_{00}h_{04}h_{12}
+7h_{12}^2
+h_{02}^2(-72h_{01}h_{03}+48h_{00}h_{04}+46h_{12})
+78h_{00}h_{03}h_{13}
-6h_{11}h_{13}\\
&\quad
-16h_{10}h_{14}
+6h_{01}^2(4h_{03}^2+h_{14})
-36h_{04}h_{20}
-12h_{03}h_{21}+6h_{32}\\
&\quad
+h_{02}
\Big(
48h_{01}^2h_{04}
-85h_{04}h_{10}
-38h_{03}h_{11}
-20h_{01}h_{13}
+h_{00}(243h_{03}^2-58h_{01}h_{05}+11h_{14})
\\
&\quad
+22h_{22}
\Big)
+h_{01}
\Big(
-278h_{00}h_{03}h_{04}
+92h_{05}h_{10}
+30h_{04}h_{11}
-4(6h_{03}h_{12}+h_{23})
\Big)
\Big)x_1^5\\
&
+\frac{1}{6}
\left(
5h_{02}h_{03}
-4h_{01}h_{04}
-7h_{00}h_{05}
+h_{13}
\right)x_1^3y_1
+\frac{1}{6}
\left(
3h_{03}^2
+6h_{02}h_{04}
-6h_{01}h_{05}
+h_{14}
\right)x_1^3y_1^2\\
&
+\frac{1}{24}
\Big(
26h_{02}^2h_{03}
-15h_{01}h_{03}^2
-28h_{01}h_{02}h_{04}
+31h_{00}h_{03}h_{04}
+18h_{01}^2h_{05}-16h_{05}h_{10}
-9h_{04}h_{11}\\
&\qquad\qquad
+8h_{03}h_{12}
+9h_{02}(h_{00}h_{05}+h_{13})
-4h_{01}h_{14}
+2h_{23}
\Big)x_1^4y_1 ,\\
        \Psi(x_1,y_1)&=h_{02}x_1y_1
+h_{03}x_1y_1^2
+h_{04}x_1y_1^3
+h_{05}x_1y_1^4
+\frac{1}{2}
\left(
2h_{02}^2-h_{01}h_{03}-3h_{00}h_{04}+h_{12}
\right)x_1^2y_1\\
&
+\frac{1}{6}
\Big(
6h_{02}^3
+2h_{01}^2h_{04}
-5h_{04}h_{10}
-h_{03}h_{11}
+h_{02}(-7h_{01}h_{03}-18h_{00}h_{04}+7h_{12})
-h_{01}h_{13}\\
&\qquad
+h_{00}(6h_{03}^2+13h_{01}h_{05}-3h_{14})
+2h_{22}
\Big)x_1^3y_1\\
&\quad
+\frac{1}{2}
\left(
5h_{02}h_{03}
-2h_{01}h_{04}
-4h_{00}h_{05}
+h_{13}
\right)x_1^2y_1^2\\
&
+\frac{1}{24}
\Big(
24h_{02}^4
-27h_{00}^2h_{04}^2
-6h_{01}^3h_{05}
-183h_{00}^2h_{03}h_{05}
+21h_{03}^2h_{10}
+5h_{00}h_{05}h_{11}-14h_{00}h_{04}h_{12}\\
&\quad
+7h_{12}^2
+46h_{02}^2(-h_{01}h_{03}+h_{12})
+48h_{00}h_{03}h_{13}
-2h_{11}h_{13}
-8h_{10}h_{14}-12h_{04}h_{20}\\
&\quad
+6h_{32}+h_{01}^2(9h_{03}^2+2h_{14})+h_{02}
\Big(
20h_{01}^2h_{04}
-49h_{04}h_{10}
-12h_{03}h_{11}
-11h_{01}h_{13}
+22h_{22}\\
&\quad
+h_{00}(159h_{03}^2-h_{01}h_{05}+3h_{14})\Big)\\
&\quad
+h_{01}
\Big(
28h_{05}h_{10}-127h_{00}h_{03}h_{04}
+5h_{04}h_{11}
-2(8h_{03}h_{12}+h_{23})
\Big)
\Big)x_1^4y_1\\
&
-\frac{1}{24}
\Big(
42h_{00}^3h_{04}h_{05}
-14h_{02}^2h_{03}h_{10}
+7h_{01}h_{03}^2h_{10}
-18h_{01}^2h_{05}h_{10}
+16h_{05}h_{10}^2\\
&\quad
+9h_{04}h_{10}h_{11}
-4h_{03}h_{10}h_{12}
+h_{00}^2
\Big(
12h_{03}^3
+24h_{01}h_{04}^2
+9h_{02}^2h_{05}
-14h_{05}h_{12}
-6h_{04}h_{13}\\
&\quad
+h_{03}(h_{02}h_{04}+68h_{01}h_{05}-6h_{14})
\Big)
+4h_{01}h_{10}h_{14}
+16h_{01}h_{04}h_{20}
-4h_{13}h_{20}\\
&
+h_{02}
\left(
20h_{01}h_{04}h_{10}
-7h_{10}h_{13}
-20h_{03}h_{20}
\right)
-2h_{10}h_{23}-12h_{03}h_{30}\\
&
+h_{00}
\Big(
26h_{02}^3h_{03}
+28h_{01}^2h_{03}h_{04}
+2h_{12}h_{13}
+h_{02}^2(-28h_{01}h_{04}+9h_{13})\\
&\quad
-8h_{01}(h_{04}h_{12}+h_{03}h_{13})
+28h_{05}h_{20}
+h_{03}(-49h_{04}h_{10}-8h_{03}h_{11}+4h_{22})\\
&\quad
+h_{02}
\Big(
-35h_{01}h_{03}^2
+18h_{01}^2h_{05}
-39h_{05}h_{10}
-9h_{04}h_{11}
+20h_{03}h_{12}
-4h_{01}h_{14}
+2h_{23}
\Big)
\Big)
\Big)x_1^5 .
    \end{align*}
\noindent
\subsection{Coefficients of system~\texorpdfstring{\eqref{eq:unfolding_system3}}{}}\label{Coeff_3}
\begin{align*}
        p_{00}&=h_{00}, \; p_{10}=h_{10}-h_{00}h_{02}, \; p_{01}=h_{01}, \; p_{20}=h_{20}-\frac{h_{02}h_{10}}{2}+\frac{1}{2}h_{00}(2h_{01}h_{03}+3h_{00}h_{04}-h_{12}),\\
        p_{11}&=h_{11}-3h_{00}h_{03}, \; p_{21}=h_{21}+\frac{1}{2}\left(-6h_{03}h_{10}+h_{02}h_{11}+h_{00}(8h_{01}h_{04}+15h_{00}h_{05}-3h_{13})\right), \\
        p_{30}&=h_{30}+\frac{1}{6}\Big(-h_{10}(h_{02}^2-4h_{01}h_{03}+2h_{12})+h_{00}^2(-6h_{03}^2-20h_{01}h_{05}+3h_{14})\\
        &\qquad\qquad+h_{00}(-6h_{01}^2h_{04}+9h_{04}h_{10}+4h_{03}h_{11}-h_{02}h_{12}+2h_{01}h_{13}-2h_{22})\Big),\\
        p_{40}&=h_{40}+\frac{1}{24}\Big(3h_{00}^3(27h_{04}^2+61h_{03}h_{05})-2h_{02}^3h_{10}-18h_{01}^2h_{04}h_{10}+9h_{04}h_{10}^2+12h_{03}h_{10}h_{11}\\
        &\qquad+6h_{01}h_{10}h_{13}-h_{00}^2(108h_{02}^2h_{04}-230h_{01}h_{03}h_{04}+33h_{05}h_{11}+22h_{04}h_{12}
+48h_{03}h_{13}\\
&\qquad+h_{02}(105h_{03}^2-170h_{01}h_{05}+27h_{14}))
-2h_{02}^2h_{20}
+8h_{01}h_{03}h_{20}
-4h_{12}h_{20}
-6h_{10}h_{22}\\
&\qquad
+h_{02}(8h_{01}h_{03}h_{10}-7h_{10}h_{12}+12h_{30})+h_{00}\Big(
24h_{01}^3h_{05}
-36h_{03}^2h_{10}
+23h_{02}h_{04}h_{10}\\
&\qquad
+6h_{02}h_{03}h_{11}
-2h_{02}^2h_{12}-h_{12}^2
+6h_{11}h_{13}
-6h_{01}^2h_{14}
+15h_{10}h_{14}
+8h_{04}h_{20}
+12h_{03}h_{21}
-6h_{02}h_{22}\\
&\qquad
+h_{01}(-92h_{05}h_{10}-30h_{04}h_{11}+4h_{02}h_{13}+4h_{23})
-6h_{32}
\Big)
\Big),\\
p_{31}&=h_{31}+\frac{1}{6}\Big(
-120h_{00}^2(h_{03}h_{04}+h_{02}h_{05})
-9h_{02}h_{03}h_{10}
+24h_{01}h_{04}h_{10}\\
&\qquad
+2h_{02}^2h_{11}
-2h_{01}h_{03}h_{11}
+h_{11}h_{12}
-9h_{10}h_{13}
-18h_{03}h_{20}
+6h_{02}h_{21}\\
&\qquad
+h_{00}(-30h_{01}^2h_{05}+75h_{05}h_{10}+10h_{04}h_{11}
+3h_{03}h_{12}+8h_{01}h_{14}-6(h_{02}h_{13}+h_{23}))\Big),\\
p_{50}&=h_{50}+\Bigg(
\frac{1}{120}\Big(
-6h_{02}^4h_{10}
+96h_{01}^3h_{05}h_{10}
-72h_{03}^2h_{10}^2
-8h_{10}h_{12}^2+24h_{10}h_{11}h_{13}+24h_{10}^2h_{14}
-10h_{02}^3h_{20}\\
&\qquad
-10h_{00}^3(36h_{03}^2h_{04}
+197h_{02}h_{03}h_{05}
+6h_{04}(20h_{01}h_{05}-3h_{14}))+14h_{04}h_{10}h_{20}
+40h_{03}h_{11}h_{20}\\
&\qquad
+20h_{01}h_{13}h_{20}
-4h_{01}^2(4h_{03}^2h_{10}+6h_{10}h_{14}+15h_{04}h_{20})
+48h_{03}h_{10}h_{21}
-20h_{20}h_{22}+492h_{04}^2h_{10}\\
&\qquad
+h_{00}^2\Big(
360h_{02}^3h_{04}
-40h_{01}^2(9h_{04}^2+20h_{03}h_{05})
+558h_{03}h_{05}h_{10}
+345h_{03}h_{04}h_{11}
+150h_{03}^2h_{12}\\
&\qquad
+5h_{02}(-260h_{01}h_{03}h_{04}
+61h_{05}h_{11}
-76h_{04}h_{12}
+100h_{03}h_{13})+10h_{02}^2(148h_{03}^2-50h_{01}h_{05}+9h_{14})\\
&\qquad
-60h_{12}h_{14}-24h_{10}h_{32}
-40h_{01}(6h_{03}^3-10h_{05}h_{12}-3h_{04}h_{13}-3h_{03}h_{14})
-140h_{05}h_{21}
-120h_{04}h_{22}
\Big)\\
&\qquad
+8h_{01}h_{10}(-16h_{05}h_{10}-15h_{04}h_{11}
+2(h_{03}h_{12}+h_{23}))
+h_{02}^2(28h_{01}h_{03}h_{10}-29h_{10}h_{12}+30h_{30})\\
&\qquad
+h_{00}\Big(
120h_{02}^5
-240h_{01}^3h_{03}h_{04}
-212h_{05}h_{10}h_{11}
-45h_{04}h_{11}^2-132h_{04}h_{10}h_{12}
-30h_{03}h_{11}h_{12}\\
&\qquad
+160h_{02}^3(-3h_{01}h_{03}+2h_{12})
-162h_{03}h_{10}h_{13}
+10h_{01}^2(9h_{05}h_{11}+12h_{04}h_{12}+8h_{03}h_{13})-30h_{03}^2h_{20}\\
&\qquad
+10h_{14}h_{20}
+20h_{13}h_{21}+40h_{12}h_{22}
+h_{01}\Big(
65h_{03}^2h_{11}
-20(2h_{12}h_{13}+h_{11}h_{14}+4h_{05}h_{20}+4h_{04}h_{21})\\
&\qquad
+16h_{03}(57h_{04}h_{10}-5h_{22})
\Big)
+2h_{02}^2(180h_{01}^2h_{04}-351h_{04}h_{10}-50h_{03}h_{11}
-80h_{01}h_{13}+90h_{22})\\
&\qquad
+10h_{11}h_{23}
-120h_{04}h_{30}
+60h_{03}h_{31}
+h_{02}\Big(
-240h_{01}^3h_{05}
-147h_{03}^2h_{10}
+872h_{01}h_{05}h_{10}
+120h_{12}^2\\
&\qquad
-5h_{11}h_{13}
-149h_{10}h_{14}
+60h_{01}^2(6h_{03}^2+h_{14})
-40h_{01}(-3h_{04}h_{11}+10h_{03}h_{12}+h_{23})
\\
&\qquad
+60h_{32}
\Big)\Big)-80h_{04}h_{20}
+40h_{03}h_{21}+h_{02}\Big(
-42h_{01}^2h_{04}h_{10}
+65h_{04}h_{10}^2
+52h_{03}h_{10}h_{11}
+30h_{01}h_{10}h_{13}\\
&\qquad
+40h_{01}h_{03}h_{20}
-30h_{12}h_{20}
-38h_{10}h_{22}
+120h_{40}
\Big)
\Big)
\Bigg).
\end{align*}
\subsection{Coefficients of system~\texorpdfstring{\eqref{eq:unfolding_system4}}{}}\label{Coeff_4}
    \begin{align*}
        q_{00}&=p_{00}, \; q_{10}=p_{10}-\frac{p_{00}p_{30}}{2p_{20}}, \; q_{01}=p_{01}, \; q_{20}=\frac{80p_{20}^{3}-60p_{10}p_{20}p_{30}+45p_{00}p_{30}^{2}-48p_{00}p_{20}p_{40}}{80p_{20}^{2}},\\
        q_{11}&=p_{11}-\frac{p_{01}p_{30}}{2p_{20}}, \; q_{30}=\frac{
210p_{10}p_{20}p_{30}^{2}
-175p_{00}p_{30}^{3}
-192p_{10}p_{20}^{2}p_{40}
+336p_{00}p_{20}p_{30}p_{40}
-160p_{00}p_{20}^{2}p_{50}
}{
240p_{20}^{3}
},\\
q_{21}&=\frac{
80p_{20}^{2}p_{21}
-60p_{11}p_{20}p_{30}
+45p_{01}p_{30}^{2}
-48p_{01}p_{20}p_{40}
}{
80p_{20}^{2}
}, \; q_{40}=\frac{
-55p_{10}p_{30}^{3}
+96p_{10}p_{20}p_{30}p_{40}
-40p_{10}p_{20}^{2}p_{50}
}{
48p_{20}^{3}
},\\
q_{31}&=\frac{1}{240p_{20}^{3}}
\Big(-240p_{20}^{2}p_{21}p_{30}
+210p_{11}p_{20}p_{30}^{2}
-175p_{01}p_{30}^{3} +240p_{20}^{3}p_{31}
-192p_{11}p_{20}^{2}p_{40} \\
&\qquad\qquad+336p_{01}p_{20}p_{30}p_{40}
-160p_{01}p_{20}^{2}p_{50}
\Big),\\
q_{50}&=\frac{1}{12800p_{20}^{5}}
\Big(4850p_{10}p_{20}p_{30}^{4}
-625p_{00}p_{30}^{5}
-9600p_{10}p_{20}^{2}p_{30}^{2}p_{40} +8800p_{00}p_{20}p_{30}^{3}p_{40}
+1536p_{10}p_{20}^{3}p_{40}^{2} \\
& \qquad\qquad\qquad -14592p_{00}p_{20}^{2}p_{30}p_{40}^{2} + 640p_{20}^{2}
\left(
5p_{10}p_{20}p_{30}
-5p_{00}p_{30}^{2}
+16p_{00}p_{20}p_{40}
\right)p_{50}
\Big)
    \end{align*}

\bibliographystyle{unsrtnat}
\bibliography{references}

\end{document}